\documentclass[11pt,letterpaper]{article}

\usepackage[english]{babel}
\usepackage[margin=1in]{geometry}
\usepackage{amsfonts} 
\usepackage{xspace}
\usepackage{amsmath,amssymb,amsthm}
\usepackage{graphicx}
\usepackage{xcolor}
\usepackage[colorlinks=true, allcolors=blue]{hyperref}
\usepackage{cleveref}
\usepackage{multirow}

\usepackage{algorithm}
\usepackage{algorithmicx}
\usepackage{algpseudocode}
\algrenewcommand\algorithmicrequire{\textbf{Input:}}
\algrenewcommand\algorithmicensure{\textbf{Output:}}

\newtheorem{theorem}{Theorem}[section]
\makeatletter

\newtheorem{claim}[theorem]{Claim}
\newtheorem{remark}[theorem]{Remark}

\newtheorem{lemma}[theorem]{Lemma}

\newtheorem{corollary}[theorem]{Corollary}

\newtheorem{definition}[theorem]{Definition}

\newtheorem{observation}[theorem]{Observation}

\newcommand{\eps}{\varepsilon}
\renewcommand{\epsilon}{\varepsilon}

\newcommand{\eat}[1]{}

\newcommand{\R}{\mathbb{R}}

\newcommand{\Z}{\mathbb{Z}}

\newcommand{\ddim}{\ensuremath{\mathrm{ddim}}}

\newcommand{\calA}{\mathcal{A}}
\newcommand{\calB}{\mathcal{B}}
\newcommand{\calC}{\mathcal{C}}

\newcommand{\calF}{\mathcal{F}}
\newcommand{\calG}{\mathcal{G}}
\newcommand{\calH}{\mathcal{H}}

\newcommand{\calM}{\mathcal{M}}
\newcommand{\calN}{\mathcal{N}}
\newcommand{\calR}{\mathcal{R}}
\newcommand{\calU}{\mathcal{U}}

\newcommand{\calX}{\mathcal{X}}
\newcommand{\calZ}{\mathcal{Z}}

\newcommand{\OPT}{\ensuremath{\mathsf{OPT}}}

\newcommand{\Exp}{\ensuremath{\mathbb{E}}}

\newcommand{\HUS}{\ensuremath{\mathsf{HUS}}\xspace}

\newcommand{\cover}{\ensuremath{N_{\calX}}}
\newcommand{\avg}{\ensuremath{h}}
\newcommand{\AT}{\ensuremath{\mathsf{OAT}}\xspace}
\newcommand{\cc}[1]{\ensuremath{\Lambda_{#1}(\mathcal{X})}}
\newcommand{\conv}{\ensuremath{\mathsf{conv}}}

\newcommand{\lip}{\ensuremath{\mathsf{Lip}}}
\newcommand{\rank}{\ensuremath{\mathsf{rank}}}

\newcommand{\tw}{\ensuremath{t}}
\newcommand{\augG}{\ensuremath{G_{\text{ex}}}}

\renewcommand{\eqref}[1]{(\ref{#1})}

\newcommand{\ProblemName}[1]{\textsc{#1}}
\newcommand{\kzC}{\ProblemName{$(k, z)$-Clustering}\xspace}
\newcommand{\kzmC}{\ProblemName{$(k, z)$-Clustering} with $m$ outliers\xspace}

\newcommand{\kMedian}{\ProblemName{$k$-Median}\xspace}
\newcommand{\kMeans}{\ProblemName{$k$-Means}\xspace}

\makeatletter
\newcommand*{\rom}[1]{\expandafter\@slowromancap\romannumeral #1@}
\makeatother

\DeclareMathOperator{\poly}{poly}
\DeclareMathOperator{\cost}{cost}

\DeclareMathOperator{\Ball}{Ball}
\DeclareMathOperator{\Balls}{Balls}

\DeclareMathOperator{\ring}{ring}

\newcommand{\pclose}{\ensuremath{p_{\mathrm{close}}\xspace}}
\newcommand{\pfar}{\ensuremath{p_{\mathrm{far}}\xspace}}
\newcommand{\Cclose}{\ensuremath{C_{\mathrm{close}}\xspace}}
\newcommand{\Cfar}{\ensuremath{C_{\mathrm{far}}\xspace}}
\newcommand{\barCfar}{\ensuremath{\overline{C}_{\mathrm{far}}\xspace}}

\title{Coresets for Constrained Clustering:\\
	General Assignment Constraints and Improved Size Bounds}

\author{Lingxiao Huang\thanks{
		Email: \texttt{huanglingxiao1990@126.com}
	}\\
	State Key Laboratory of Novel Software Technology \\ Nanjing University
	\and
	Jian Li\thanks{
		Email: \texttt{lapordge@gmail.com}
	}\\
	Tsinghua University
	\and
	Pinyan Lu\thanks{ 
		Email: \texttt{lu.pinyan@mail.shufe.edu.cn}}\\
	Shanghai University of Finance and Economics \\
	Key Laboratory of Interdisciplinary Research of Computation and Economics (SUFE), \\ Ministry of Education \\
	and Huawei TCS Lab
	\and
	Xuan Wu\thanks{ 
		Email: \texttt{wu3412790@gmail.com}}\\
	Nanyang Technological University
}
\date{}

\begin{document}
	
	\maketitle
	
	\begin{abstract}%
		Designing small-sized \emph{coresets}, which approximately preserve the costs of the solutions for large datasets, has been an important research direction for the past decade. We consider coreset construction for a variety of
		general constrained clustering problems.
		We introduce a general class of assignment constraints,
		including capacity constraints on cluster centers, and assignment structure constraints for data points (modeled by a convex body $\calB$).
		We give coresets for clustering problems with such general assignment constraints that
		significantly generalize and improve known results.
		Notable implications include
		the first $\eps$-coreset for capacitated and fair \kMedian with $m$ outliers in Euclidean spaces whose size is $\tilde{O}(m + k^2 \epsilon^{-4})$, generalizing and improving  upon the prior
		bounds in \cite{braverman2022power,Huang2022NearoptimalCF}
		(for capacitated \kMedian, the coreset size bound obtained in \cite{braverman2022power} 
		is $\tilde{O}(k^3 \epsilon^{-6})$, and for \kMedian with $m$ outliers, the coreset size bound obtained in \cite{Huang2022NearoptimalCF} is 
		$\tilde{O}(m + k^3 \epsilon^{-5})$), and
		the first $\eps$-coreset of size $\poly(k \eps^{-1})$ for fault-tolerant clustering for various types of metric spaces.
		
		Our algorithm improves upon the hierarchical uniform sampling framework in~\cite{braverman2022power, Huang2022NearoptimalCF}
		by employing new adaptive sampling steps,
		resulting in better coreset size upper bounds for \kzC subject to various capacity constraints.
		In addition, we introduce novel techniques 
		to handle assignment structure constraints.
		Specifically, we relate the coreset size to a complexity measure $\lip(\calB)$ of the structure constraint,
		where $\lip(\calB)$ for convex body $\calB$
		is the Lipschitz constant of a certain transportation problem constrained in $\calB$, called {\em optimal assignment transportation problem}.
		We prove nontrivial upper bounds of $\lip(\calB)$
		for various polytopes, including the general matroid basis polytopes,
		and laminar matroid polytopes (with a better bound).


	\end{abstract}
	
	\newpage
	\tableofcontents
	\newpage

	\section{Introduction}
	\label{sec:intro}
	
	We study coresets for clustering with general assignment constraints.
	Clustering is a fundamental data analysis task that receives significant attention in various areas.
	In the (center-based) clustering problem,
	given as input a metric space $(\calX, d)$
	and a finite point set $P \subseteq \calX$,
	the clustering cost is defined for a set of clustering centers $C$ of $k$ points from $\calX$,
	and an assignment function $\sigma : P \times C \to \mathbb{R}_+$
	that assigns each data point $p$ to the centers (such that $\| \sigma(p, \cdot) \|_1 = 1$ 
	which ensures the point $p$ is fully assigned 
	\footnote{Here, for some function of the form $\sigma : X \times Y \to \R_{\geq 0}$,
		we write $\sigma(x, \cdot)$ as the vector $u \in \R^Y$
		such that $\forall y\in Y$, $u_y = \sigma(x, y)$. The notation $\sigma(\cdot, y)$ is defined similarly.}%
	),
	as an $\ell_z$ 
	aggregation 
	($z\geq 1$) of the
	distances from the data points to the centers $C$ weighted by $\sigma$, i.e.,
	\begin{equation}
		\label{eqn:cost_sigma}
		\cost_z^\sigma(P, C) := \sum_{x \in P}\sum_{c \in C} \sigma(x, c) \cdot (d(x, c))^z.
	\end{equation}
	The goal of the clustering problem is to find a center set $C$ and assignment $\sigma$ that minimizes the clustering cost.
	If there is no additional constraint on $\sigma$, then it is optimal to assign the data points to the (unique) nearest center, which is the case of vanilla \kzC.
	However, the general formulation above allows a point to be assigned fractionally to multiple centers (which is sometimes called soft clustering)
	if we impose some additional constraint on $\sigma$.

	\paragraph{General Assignment Constraints}
	In this paper, we 
	introducing the {\em general assignment constraints},
	which impose various constraints on the assignment function $\sigma$. 
	The formal definition can be found in Section~\ref{sec:pre}.
	The general assignment constraint is very general and 
	unifies several important clustering problems as follows: 
	\begin{enumerate}
		\item Capacitated clustering: The problem is
		an important class of constrained clustering problems and has been studied extensively (see e.g.,~\cite{CharikarGTS02,levi2004lp,cygan2012lp,li2016approximating}). In this problem, we impose a \emph{capacity upper bound} on every center $c \in C$, i.e., $\|\sigma(\cdot, c)\|_1 \leq u_c$.
		\item Fair clustering:
		Recently, a group-membership fairness definition, called fair clustering, has received significant interest
		\cite{Chierichetti0LV17,BeraCFN19,BackursIOSVW19}.  Each point has a color indicating its group membership, and we need to find an assignment subject to the constraint that each cluster has each color represented within some pre-specified proportions. 
		\item Clustering with outliers:
		In some application domains, some data points may be 
		regarded as outliers and does not need to be clustered.
		The problem has also been studied extensively for various objectives     \cite{charikar2001algorithms,chen2008constant,gupta2017local,statman2020kmeans,bhaskara2018low,mount2014on}.
		This can be captured by imposing total capacity constraints $\|\sigma(\cdot,\cdot) \|_1 = n-m$ 
		(i.e., the number of outliers is $m$) and $\|\sigma(p,\cdot)\|_1\le 1$ for every point $p$.
		\item Fault-tolerant clustering:
		Another significant extension of clustering is the {\em fault tolerant} clustering problem, which requires each point to be assigned to at least $l$ centers $(l \geq 1)$ (see e.g.,~\cite{khuller2000fault,SwamyS08,hajiaghayi2016constant}).
		This can be captured by the 
		{\em assignment structure constraints}
		which requires that the vector $\sigma(p,\cdot)$
		lies in a constrained convex set 
		$\calB=\left\{x\in \Delta_k: x_i\leq \frac{1}{l}, \forall i\in [k]\right\}$.
		
	\end{enumerate}

	\paragraph{Coresets}
	We focus on coresets for constrained clustering with general assignment constraints. 
	Roughly speaking, an $\epsilon$-coreset $S$ is a tiny proxy of the data set $P$,
	such that the cost evaluated on both $S$ and $P$ are within $(1 \pm \eps)$ factor for all potential centers $C$ and assignments satsifying 
	assignment constraints.
	The coreset is a powerful technique that can be used to 
	compress large datasets,
	speed up existing algorithms, and design efficient approximation schemes~\cite{cohen-addad2019on,BandyapadhyayFS21,BJKW21}.
	The concept of coreset has also found many applications in modern sublinear models,
	such as streaming algorithms~\cite{harpeled2004on},
	distributed algorithms~\cite{BalcanEL13} and dynamic algorithms~\cite{HenzingerK20},
	by combining it with the so-called merge-and-reduce framework~\cite{harpeled2004on}.
	
	The study of coreset size bounds for clustering has been very fruitful,
	especially the case of vanilla \kzC (i.e., without constraints).
	A series of works focus on the Euclidean spaces~\cite{harpeled2004on,DBLP:journals/dcg/Har-PeledK07,feldman2011unified,FSS20,sohler2018strong,huang2020coresets,cohen2021new,cohen2022towards}
	and near-optimal size bounds have been obtained in more recent works for \kMedian and \kMeans~\cite{cohen2021new,cohen2022towards,cohenaddad2022improved,Huang2024OnOC}.
	Recently, coreset size bounds in small dimensional Euclidean spaces have also been investigated \cite{Huang2023OnCF}.
	Going beyond Euclidean spaces, another series of works provide coresets of small size 
	in other types of metrics, such as doubling metrics~\cite{huang2018epsilon,cohen2021new} and graph shortest-path metrics~\cite{baker2020coresets,BJKW21}.
	More interestingly, throughout the line of research, various fundamental techniques have been proposed,
	including importance sampling~\cite{feldman2011unified,FSS20} which can be applied
	to several problems including clustering,
	and a more recently developed hierarchical sampling framework~\cite{Chen09,cohen2021new,braverman2022power} that employs uniform sampling in a novel way.
	
	\paragraph{Coreset for Constrained Clustering}
	Unfortunately, coresets for constrained clustering has been less understood.
	In particular, only the capacity constraints were considered,
	and the research focus was on coresets for fair clustering~\cite{Chierichetti0LV17} and capacitated clustering~\cite{CharikarGTS02}.
	Earlier works~\cite{SSS19,huang2019coresets,cohen-addad2019on,BandyapadhyayFS21}
	achieved coresets of size either depending on $n$ or exponential in $d$ (which is the Euclidean dimension).
	Recently, a breakthrough was made in~\cite{braverman2022power}
	where the first coresets for both fair and capacitated \kMedian in Euclidean $\mathbb{R}^d$ with size $\poly(k \eps^{-1})$ were obtained,
	via an improved hierarchical uniform sampling framework.
	The framework has also been adapted to the outlier setting in a more recent work~\cite{Huang2022NearoptimalCF}.
	This framework certainly provides a good starting point,
	but several fundamental issues still remain.
	One issue is that coresets obtained through this framework are still somewhat ad-hoc,
	and it is unclear if the result
	can be adapted to more general assignment constraints such as the aforementioned structure constraints,
	and/or other metrics such as graph shortest-path metrics.
	Indeed, a perhaps more fundamental question is that,    
	a systematic characterization of what types of assignment constraints allow small coresets, is still missing in the literature.
	In addition, the framework and analysis in~\cite{braverman2022power} only lead to 
	$\poly(k\epsilon^{-1})$ size bound with high degree polynomial, which is also sub-optimal.

	\subsection{Our Contributions}
	\label{subsec:contribution}
	
	Our main contribution is two-fold.
	(1) We propose a very general model of assignment constraints
	(including capacity constraint, outliers constraint, and the aforementioned 
	structure constraint), and provide a characterization of
	families of assignment constraints that admit small coresets.
	(2) Our new analysis leads to improved coreset size upper bounds, even for the important cases of fair/capacitated clustering and clustering with outliers (without additional structure constraints), achieving state-of-the-art bounds for these problems.
	Next, we discuss our contributions in more detail.

	\paragraph{A General Model of Assignment Constraints} 
	Our new model for the assignment constraints, called the \emph{general assignment constraints},
	is a combination of three types of constraints on the assignment function 
	$\sigma(\cdot,\cdot)$: (1) the \emph{assignment structure constraint} (see \Cref{def:structure_constraint}) which is a new notion proposed in this work, (2) the standard {\em capacity constraint} (\Cref{def:capacity_constraint}) that constrains the total weight assigned to each center,
	and (3) the {\em total capacity constraint} (\Cref{def:total_capacity_constraint}), which can be used to capture the number
	of outliers.
	%
	The new structure constraint (\Cref{def:structure_constraint})
	specifies a convex body $\calB \subseteq \Delta_k$ where $\Delta_k := \{x\in \mathbb{R}^k_{\geq 0} : \|x\|_1=1\}$ is the $k$-dimensional simplex,
	and it requires that for every point $p$ the assignment vector $\sigma(p,\cdot)$ must lie in $\calB$. 
	Note that the case $\calB = \Delta_k$ corresponding to the standard constraint that a point is fully assigned to the centers.
	%
	We focus on the well-known cases for convex body $\calB$, such as matroid basis polytopes and knapsack polytopes.
	Indeed, these types of $\calB$ are already general enough to capture many of the above-mentioned constraints, including the fault-tolerance constraint (in which case $\calB=\left\{x\in \Delta_k: x_i\leq \frac{1}{l}, \forall i\in [k]\right\}$).
	See \Cref{def:structure_constraint} and the subsequent discussions.
	In addition, to capture the outliers in clustering, we introduce 
	the {\em total capacity constraint} which 
	requires $\|\sigma\|_1 = n-m$
	(i.e., the number of outliers is $m$).

	\paragraph{Main Theorem}
	In our main theorem (Theorem~\ref{thm:intro_main}), we show that a small coreset exists,
	as long as the structure constraints have bounded complexity $\lip(\calB) = O_k(1)$
	(see Definitions~\ref{def:transportation} and~\ref{def:mass}), 
	which only depends on the convex body $\calB$,
	and the {\em covering exponent} $\cc{\eps}$ 
	(see \Cref{def:covering_wo})
	can be bounded by a number that is independent of $n$.
	Hence, the main theorem systematically reduces the problem of constructing coresets of size $\poly(k \eps^{-1})$, to the mathematical problems of bounding the parameters $\lip(\calB)$ and $\cc{\eps}$.
	The parameter $\lip(\calB)$ is new, and is defined as the Lipschitz constant of a certain transportation procedure inside the convex body $\calB$.
	On the other hand, the covering exponent $\cc{\eps}$ (also known as 
	the log covering number or the metric entropy) of metric space $(\calX,d)$ is closely related to several combinatorial dimension notions, such as the VC-dimension and the (fat) shattering dimension of the set system formed by all metric balls
	which have been extensively studied in previous works, e.g.~\cite{langberg2010universal,feldman2011unified,FSS20,huang2018epsilon,BJKW21}.
	Bounds for $\cc{\eps}$ are known for multiple metric spaces, including Euclidean metrics, doubling metrics, general discrete metrics, and shortest-path metrics (see \Cref{remark:covering_exponent}).
	Now, we state our main theorem.

	
	
	\begin{theorem}[\bf{Informal; see Theorem~\ref{thm:coreset}}]
		\label{thm:intro_main}
		We consider \kzC with capacity upper/lower bound constraint for each center, assignment structure constraint for each point
		(specified by convex body $\calB \subseteq \Delta_k$),
		and a total capacity constraint $\|\sigma\|_1=n-m$ (i.e., $m$ outliers).
		For any $0 < \eps < 1$, there is a near linear-time algorithm that computes an $\eps$-coreset of size 
		$O(m) + \tilde{O}_z( \lip(\calB)^2\cdot (\cc{\eps} + k + \eps^{-1}) \cdot  k^2\eps^{-2z})$. 
		\footnote{Throughout, the notation $\tilde O_z(f)$ hides factors $\poly \log f$ and $2^{\poly(z)}$}
		%
		
		Moreover, if there is no additional structure constraint (i.e., $\calB = \Delta_k$), the coreset size bound can be improved to $O(m) + \tilde{O}_z((\cc{\eps} + \eps^{-1}) \cdot  k^2\eps^{-2z})$.
	\end{theorem}
	
	\noindent
	We provide an overview of the proof in Section \ref{sec:tech_overview}, and the complete proof appears in Section \ref{sec:main_thm}.
	The theorem is general and provides improved coreset size bounds
	for a variety of constrained clustering problems.
	Unlike many previous works that use ad-hoc methods
	to deal with constraints in specific metric spaces (such as Euclidean $\mathbb{R}^d$)~\cite{SSS19,huang2019coresets,cohen-addad2019on,BandyapadhyayFS21,braverman2022power},
	we completely decouple  the parameter of constraints $\lip(\calB)$ and the complexity $\cc{\eps}$ of the metric space,
	so that they can be dealt with independently.
	Moreover, our coreset size bound is optimal (up to constant factor) in the dependence of $m$, due to an $\Omega(m)$ lower bound for coresets for clustering with $m$ outliers~\cite{Huang2022NearoptimalCF},
	and the dependence in $k$ and $\epsilon$ improves over the bounds in~\cite{braverman2022power}  by a factor of $k\eps^{-z}$ (which works for only the capacity constraint).
	Hence, even without any constraint on $\calB$, i.e., $\calB = \Delta_k$ where we have $\lip(\calB) = 1$,
	and only considering the capacity constraints,
	we can already obtain several new/improved coreset results, by simply using known bounds of $\cc{\eps}$ (see \Cref{remark:covering_exponent}).
	We list more concrete implications of our general theorem, 
	which can be found at the end of this section and Table~\ref{tab:result}.
	%
	
	
		
		
	
	\newcommand{\specialcell}[2][c]{%
		\begin{tabular}[#1]{@{}c@{}}#2\end{tabular}}
	
	\begin{table}[t]
		\centering
		
		\scriptsize
		
		\begin{tabular}{|c|c|c|c|}
			\hline
			\multicolumn{2}{|c|}{Variants of constrained clustering in $\R^d$} & Prior results & Our results \\ \hline
			\multicolumn{1}{|c|}{\kzC} & \multicolumn{1}{|c|}{with outliers} & $O(m) + \tilde{O}(k^3 \eps^{-3z - 2})$~\cite{Huang2023OnCF} & $O(m) + \tilde{O}(k^2 \eps^{-2z-2})$  \\ \hline 
			\multicolumn{1}{|c|}{\multirow{2}{*}{Capacitated \kzC}} & \multicolumn{1}{|c|}{without outliers} &  \specialcell{$O(k^2 \eps^{-3} \log^2 n)$ ($z=1$)~\cite{cohen-addad2019on} \\ $\tilde{O}(k^3 \eps^{-6})$ ($z=1$)~\cite{braverman2022power} \\ $O(k^5 \eps^{-3} \log^5 n)$ ($z=2$)~\cite{cohen-addad2019on} \\ $\mathrm{poly}(k, \eps^{-1})$ ($z=2$)~\cite{braverman2022power}} & $\tilde{O}( k^2 \eps^{-2z-2})$ \\ \cline{2-4}
			\multicolumn{1}{|c|}{} & \multicolumn{1}{|c|}{with outliers} & $/$ & $O( m) + \tilde{O}(k^2 \eps^{-2z-2})$ \\ \hline
			\multicolumn{1}{|c|}{\multirow{2}{*}{Fair \kzC}} & \multicolumn{1}{|c|}{without outliers} &  \specialcell{$O(\Gamma k^2 \eps^{-3} \log^2 n)$ ($z=1$)~\cite{cohen-addad2019on} \\ $\tilde{O}(\Gamma k^3 \eps^{-6})$ ($z=1$)~\cite{braverman2022power} \\ $O(\Gamma k^5 \eps^{-3} \log^5 n)$ ($z=2$)~\cite{cohen-addad2019on} \\ $\Gamma\cdot \mathrm{poly}(k, \eps^{-1})$ ($z=2$)~\cite{braverman2022power}} & $\tilde{O}( k^2 \eps^{-2z-2})$ \\ \cline{2-4}
			\multicolumn{1}{|c|}{} & \multicolumn{1}{|c|}{with outliers} & $/$ & $O(\Gamma m) + \tilde{O}(\Gamma k^2 \eps^{-2z-2})$ \\ \hline
			\multicolumn{1}{|c|}{\multirow{2}{*}{Fault-tolerant \kzC}} & \multicolumn{1}{|c|}{without outliers} & $/$ & $\tilde{O}(k^2 \eps^{-2z} (k + \eps^{-2}) )$ \\ \cline{2-4}
			\multicolumn{1}{|c|}{} & \multicolumn{1}{|c|}{with outliers} & $/$ &  $O(m) + \tilde{O}(k^2 \eps^{-2z} (k + \eps^{-2}) )$ \\ \hline
		\end{tabular}
		\caption{Comparison of the state-of-the-art coreset sizes and our results for different variants of constrained \kzC in Euclidean spaces $\R^d$.
			The factor $\Gamma$ in fair clustering denotes the number of distinct collections of groups that a point may belong to; also see Claim \ref{claim:capacity}. 
			We assume that $z\geq 1$ is a constant and ignore $2^{O(z)}$ or $z^{O(z)}$ factors in the coreset size, and also ignore the logarithmic terms in $\tilde{O}(\cdot)$.
		}
		\label{tab:result}
	\end{table}

	\paragraph{New Coreset Results and Improved Bounds}
	
	We list several notable results followed by our main theorem (Theorem~\ref{thm:intro_main}),
	using new or existing upper bounds on $\lip(\calB)$
	and $\cc{\eps}$.
	Also, see Table \ref{tab:result} for a summary.
	
	
	\begin{itemize}
		\item Coresets for fair/capacitated clustering with outliers under various metrics.
		This corresponds to the case $\calB = \Delta_k$ and $\lip(\calB)=1$.
		We take capacitated clustering as an example.
		In the context of fair clustering, the coreset size includes an additional factor \( \Gamma \) compared to that of capacitated clustering.
		For metric spaces with bounded doubling dimension, shortest-path metrics of planar graphs, or more generally graphs that exclude a fixed minor,
		the covering exponent $\cc{\eps}$ can also be bounded independent of $n$ (see \Cref{remark:covering_exponent}). 
		Hence, we obtain the first $O(m + \poly(k \eps^{-1}))$-sized 
		coreset of capacitated \kMedian with $m$ outliers under the above metric spaces.
		Previously, for capacitated clustering, even without outliers,
		coresets of size $\poly(k \eps^{-1})$ are known only for Euclidean \kMedian \cite{braverman2022power}%
		\footnote{We remark that Braverman et al.~\cite{braverman2022power} also provided coresets for various metric spaces, but only for vanilla clustering. For constrained clustering (such as fair/capacitated clustering), only Euclidean $\R^d$ was considered and it turns out to be nontrivial to generalize to other metrics. See \Cref{sec:tech_overview} for a more detailed discussion.
		}
		and our size bound is already better by a factor of $k \eps^{-2}$ for this special case (e.g., the bound in~\cite{braverman2022power} for capacitated \kMedian is $\tilde{O}(k^3 \eps^{-6})$
		and our bound is $\tilde{O}(k^2 \eps^{-4})$).
		Our result also generalizes the recent work~\cite{Huang2022NearoptimalCF} which provides coresets for clustering with outliers (but cannot handle e.g., fairness constraints).
		Our bound also achieves a tight linear dependence in $m$,
		and the other term is a factor of $k\epsilon^{-z}$ better
		(the bound in~\cite{Huang2022NearoptimalCF} is $O(m) + \tilde{O}(k^3 \eps^{-3z-2})$
		and our new bound is $O(m) + \tilde{O}(k^2 \eps^{-2z-2})$ in this setting).
		\item 
		Coresets for fault-tolerant clustering. 
		In fault-tolerant clustering, we require each point to be assigned to at least $l \geq 1$ centers.
		In this case, $\calB=\left\{x\in \Delta_k: x_i\leq 1/l, \forall i\in [k]\right\}$ is a (scaled) uniform matroid basis polytope. 
		By Theorem~\ref{thm:laminar_matroid}, $\lip(\calB)$ is bounded by $2$ since uniform matroid is a laminar matroid of depth $1$. 
		Hence, we obtain the first coreset of size $\poly(k \eps^{-1})$ 
		for the fault-tolerant $k$-median in Euclidean space (and 
		other metric spaces with bounded covering exponents).
		\item 
		Coresets for clustering with more general fault-tolerance requirements.
		In the variant of clustering defined in~\cite{bredin2005deploying}, suppose the points are partitioned into several groups based on  some geographical regions.
		The goal is to choose $k$ centers subject to
		the constraint that $k_i$ center are chosen from 
		the $i$-th group ($\sum_i k_i=k$).
		In addition, it is required  that each point is connected to centers in at least $l$ different groups.
		For this variant, $\calB$ corresponds to 
		a Laminar matroid basis polytope of depth $2$ (see the discussion 
		after Definition~\ref{def:structure_constraint}), and by Theorem~\ref{thm:laminar_matroid}, $\lip(\calB)$ is bounded by $3$. Hence, we obtain the 
		the first coreset of size $\poly(k \eps^{-1})$ for this clustering problem. 
		\item 
		Simultaneous coresets. As another corollary of the main theorem, 
		we obtain coresets that hold simultaneously for a set of $m$ structure constraints, which only requires enlarging the coreset by a $\log m$ factor (see \Cref{sec:simultaneous}).
		This is particularly useful when the parameter $\calB$ is to be picked from a family that is not known in advance and is subject to experiment. In this scenario, the same coreset can be re-used, which avoids recomputing a new coreset every time a new $\calB$ is tested.
	\end{itemize}

	\color{black}
	
	\paragraph{Upper Bound of $\lip(\calB)$ when $\calB$ is Matroid Basis Polytope}
	In light of Theorem~\ref{thm:intro_main}, one can see that
	proving upper bound of $\lip(\calB)$ is crucial for bounding the coreset size.
	In fact, one can easily construct polytope $\calB$, which is simply defined by 
	some linear constraints, such that $\lip(\calB)$ is unbounded
	(See Section~\ref{sec:Lipschitz}).
	In this paper, we focus on an important class of polytopes, called
	matroid basis polytopes (see e.g., \cite{LexSchrijver2003CombinatorialO}).
	A matroid basis polytope is the convex hull of all 0/1 indicator vectors of the basis
	of a matroid. 
	Matroids generalize many combinatorial structures and are popular ways to model the structure of an assignment/solution
	in various contexts, such as online computation~\cite{BabaioffIK07,BabaioffIKK18},
	matching~\cite{lawler2001combinatorial,LeeSV13}, diversity maximization~\cite{AbbassiMT13}
	and variants of clustering problems~\cite{chen2016matroid,KrishnaswamyKNSS11,Krishnaswamy0NS15,Swamy16}.
	For a general matroid basis polytope $\calB$,
	we prove $\lip(\calB) \leq k-1$ (see Theorem~\ref{thm:matroid}).
	We also provide an improved $\lip(\calB)\leq \ell+1$ bound for the special case of laminar matroids of depth $\ell$ (see~\Cref{thm:laminar_matroid}).
	This readily implies coresets of size $\poly(k \eps^{-1})$ for any general assignment constraints with a matroid basis polytope $\calB$,
	under various types of metric spaces.

	\subsection{Related Work}
	Approximation algorithms have been extensively studied for constrained clustering problems.
	We focus on the $k$-median case for several notable problems in the following discussion.
	For fair $k$-median,~\cite{BeraCFN19} provided a bi-criteria $O(1)$-approximation in general metric spaces,
	but the solution may violate the capacity/fair constraint by an additive error.
	\cite{BackursIOSVW19} gave $O(\log n)$-approximation without violating the constraints in Euclidean spaces.
	For capacitated $k$-median, $O(1)$-approximation were known in general metrics,
	but they either need to violate the capacity constraint~\cite{DBLP:conf/icalp/DemirciL16,DBLP:conf/ipco/ByrkaRU16}
	or the number of centers $k$~\cite{DBLP:journals/talg/Li17,li2016approximating}
	by a $(1 + \epsilon)$ factor.
	Both problems admit polynomial-time algorithms that have better  approximation and/or no violation of constraints
	when $k$ is not considered a part of the input~\cite{DBLP:conf/esa/AdamczykBMM019,cohen-addad2019on,DBLP:conf/isaac/FengZH0020,BandyapadhyayFS21}.
	For fault-tolerant $k$-median, a constant approximation was given in~\cite{SwamyS08},
	and $O(1)$-approximation also exists even when the number of centers
	that each data point needs to connect can be different\cite{hajiaghayi2016constant}.
	Finally, we mention a variant of clustering called matroid $k$-median
	which also admit $O(1)$-approximation in general metrics~\cite{KrishnaswamyKNSS11}.
	In this problem, a matroid is defined on the vertices of the graph,
	and only the center set that form independent sets may be chosen.
	While this sounds different from the constraints that we consider,
	our coreset actually captures this trivially since our coreset preserves the cost for all centers (not only those that form independent sets).
	
	Apart from constrained clustering, coresets were also considered for clustering with generalized objectives (but without constraints).
	Examples include
	projective clustering~\cite{feldman2011unified,FSS20,DBLP:conf/aistats/Tukan0ZBF22},
	clustering with missing values~\cite{NEURIPS2021_90fd4f88}, ordered-weighted clustering~\cite{braverman2019coresets} and clustering with panel data~\cite{huang2021coresets}. 

	
	\subsection{Technical Overview}
	\label{sec:overview}
	
	\color{black}
	Our approach is a generalization and improvement over a recent hierarchical uniform sampling framework developed in \cite{braverman2022power}.
	Our contribution is two-fold: 1) We improve the coreset size of the framework of~\cite{braverman2022power}, even without the new assignment structure constraints (i.e., in the same setting as in~\cite{braverman2022power}), and this is achieved by employing a new adaptive sampling step; 2) We incorporate the additional assignment structure constraints into the framework, which can handle more general
	clustering problems such as fault-tolerant clustering.
	
	\paragraph{Size Improvement}
	At a high level, the framework of \cite{braverman2022power} decomposes the dataset $P$ into disjoint rings $R$,
	and takes a uniform sample $S_R$ on every ring with a \emph{uniform} size. This uniform size bound on rings directly affects the size of the coreset,
	and our idea is to improve this sample size for rings.
	%
	As observed in~\cite{braverman2022power}, this simple uniform sampling is very powerful and can preserve the coreset error $|\cost_z(R,C,\Delta_k, h)-\cost_z(S_R,C,\Delta_k,h)|$ incurred on the ring $R$, for every center set $C\subset \R^d$ and every capacity constraint $h$, by charging to a certain additive error $\mathrm{err}(R)$ that only depends on $R$ (Inequality~\eqref{eq:single_error}).
	This charging is worst-case optimal over the choice of $C$, and hence, we cannot expect to improve the error analysis for a single ring.
	However, we find that such additive error $\mathrm{err}(R)$ is only incurred when the ring $R$ is ``close enough'' to center set $C$, while the number of such rings is always small (say $\tilde{O}(k)$) for every choice of $C$ (Lemmar \ref{lem:total_level}).
	This novel geometric observation enables us to adaptively tune the sample size for each ring,
	which leads to a bounded total error of rings (Lemma \ref{lm:ring}) and significantly improves the coreset size.

	\paragraph{Handling Assignment Structure Constraints}
	\sloppy
	To handle the assignment structure constraints, the main technical step is still to bound the coreset error $|\cost_z(R,C,\calB,\avg) - \cost_z(S_R,C,\calB,\avg)|$ incurred on a ring $R$.
	As a central step, we need to show it is possible to modify the optimal assignment from $R$ to $C$, to an assignment from $S_R$ to $C$, with small additional cost subject to the constraint $\calB$.
	We reduce the problem of bounding the extra cost of such conversion to a so-called \emph{optimal assignment transportation} (OAT) problem, which aims to bound the total transportation cost from a given assignment $\sigma$ to any assignment $\sigma'$ that is consistent with $(\calB, h')$.
	We define a new notion $\lip(\calB)$ as the universal upper bound of the OAT cost, and prove nontrivial upper bounds of $\lip(\calB)$ for various polytopes, including the general matroid basis polytopes (Theorem \ref{thm:matroid}), and laminar matroid polytopes with a better bound (Theorem \ref{thm:laminar_matroid}).
	These bounds imply that our algorithm produces coreset for fault-tolerant clustering and even more general assignment structure constraints.
	Bounding $\lip(\calB)$ for other convex set $\calB$,
	as well as the efficient computation of it (which we do not need) may be of independent interest for future research.

	\section{Modeling and Definitions}
	\label{sec:pre}

	Throughout, we are given an underlying metric space $(\calX, d)$, 
	integer $k\geq 1$, constant $z\geq 1$, and precision parameter $\eps\in (0,1)$.
	For integer $n \geq 1$, let $[n] := \{1, \ldots, n\}$.
	For a function $\sigma : X \times Y \to \R$, for $i \in X$ we write $\sigma(i, \cdot)$ as the vector $u_i \in \R^Y$ such that $u_i(j) := \sigma(i, j)$, and define $\sigma(\cdot, j)$ similarly. 
	We say $\sigma$ is an assignment function if $\|\sigma(i,\cdot)\|_1\leq 1$ for every point $i\in X$.
	Sometimes, we also interpret $\sigma$ as a vector in $\R^{X \times Y}$,
	so $\|\sigma\|_1 = \sum_{i \in X, j \in Y} |\sigma(i, j)|$.
	\footnote{Here, we abuse the notation by using $\sigma\in \R^{X \times Y}$ to represent a vector instead of a matrix such that $\|\sigma\|_1$ is well defined.}
	We use $\Delta_k$ to denote the simplex $\{x\in \mathbb{R}^k_{\geq 0} : \|x\|_1=1\}$.
	Given a point $a\in \calX$ and a radius $r>0$, we define $\Ball(a,r):=\{x\in \calX,d(x,a)\leq r\}$ to be the ball of radius $r$ centered at $a$. Moreover, for two positive real numbers $r_1,r_2>0$, define  $\ring(a,r_1,r_2) := \Ball(a,r_2)\setminus \Ball(a,r_1)$.
	Throughout this paper, we assume there exists an oracle that answers $d(p,q)$ in $O(1)$ time for any $p,q\in \calX$.

	\subsection{General Assignment Constraints}
	We consider three types of assignment constraints in clustering literature:
	(1) \emph{capacity constraints} on cluster centers and 
	(2) \emph{assignment structure constraints} on $\sigma(p,\cdot)$ for points $p$,
	and 
	(3) {\em total capacity constraints} which can capture outliers in clustering.
	
	First, we model the capacity constraints in a way similar to previous works~\cite{huang2019coresets,cohen-addad2019on,BandyapadhyayFS21,braverman2022power}.
	We simply consider a vector $\avg \in \R_{\geq 0}^k$ ($k$ is the number of centers), and require that the mass assigned to each center $c$ equals $h_c$.
	
	
	\begin{definition}[\bf{Capacity constraint}]
		\label{def:capacity_constraint}
		Given a set $C$ of $k$ centers, a
		capacity constraint can be specified by 
		a vector  $\avg \in \R_{\geq 0}^C$.
		We say an assignment function $\sigma : P \times C \to \R_{\geq 0}$ is consistent with $\avg$, denoted as $\sigma \sim \avg$,
		if $\| \sigma(\cdot, c)\|_1 = \avg_c$ for every center $c \in C$,
		which means the total assignment to a center $c$ is exactly $\avg_c$.\footnote{We require the capacity to be exactly $\avg_c$
			instead of placing a lower and/or upper bound,
			since we would preserve the cost for all $\avg$ \emph{simultaneously} in our coreset.
			See \Cref{def:coreset}.
		}
		Equivalently, we can write $\sum_{p} \sigma(p,\cdot) = h$
		(both sides are $k$-dimensional vectors).
	\end{definition}
	
	\noindent
	The above capacity constraint consists of equality constraints, which seem different from inequality constraints for capacitated clustering or fair clustering.
	However, we can show that capacitated clustering and fair clustering can be captured by such equality constraint; see Section~\ref{subsec:capacitated} (Claims \ref{claim:capacity} and \ref{claim:fair}) for the detailed reductions.
	In addition, we further allow the total capacity to be less than 
	the total weight $n$, and this is useful for dealing with outliers in clustering.

	\begin{definition}[\bf{Total capacity constraint}]
		\label{def:total_capacity_constraint}
		Given an integer $0\leq m\leq n$, we 
		can impose a total capacity constraint of the form
		$\|\sigma\|_1=\|h\|_1 =n - m$.
		In an integral assignment, the total capacity constraint
		says that we can exclude $m$ points as outliers.
	\end{definition}
	
	\noindent
	In fact, if the capacity vector $\avg\in \R_{\geq 0}^k$ is given,
	the total capacity constraint is already determined by $\avg$.
	Due to the special meaning of $m$ (i.e., the number of outliers),
	we make this parameter explicit and we will introduce new ideas 
	to analyze the $m>0$ case in Lemmas~\ref{lm:ring} and \ref{lm:group}.
	
	%
	
	On the other hand, an assignment structure constraint 
	concerns the range of the assignment vector $\sigma(p,\cdot)$ for each point $p$.
	We model an assignment structure constraint by a convex body $\calB \subseteq\Delta_k$, where $\Delta_k:=\left\{x\in \R_{\geq 0}^k: \sum_{i\in [k]} x_i = 1\right\}$ denotes the simplex in $\R^k$.
	
	\begin{definition}[\bf{Assignment structure constraint}]
		\label{def:structure_constraint}
		Given convex body $\calB \subseteq \Delta_k$,
		we say an assignment function $\sigma : P \times C \to \mathbb{R}_{\geq 0}$
		is consistent with $\calB$, denoted as $\sigma \sim \calB$,
		if 
		$$
		\sigma(p, \cdot) \,\,
		\in  \,\, 
		\calB \quad
		\text {for every }\,\,\, p \in P.
		$$
	\end{definition}
	
	\noindent
	We define $\calB^o$ as $\calB\cup \{0\}$
	($0$ is used to capture the assignment for the outlier).
	Given that $\|\sigma(p,\cdot)\|_1\leq 1$, we can infer that $\sigma(p,\cdot)\in \conv(\calB^o)$ as per the definition above.

	We list a few examples of assignment structure constraints as follows.
	\begin{enumerate}
		
		\item $\calB=\left\{x\in \Delta_k: x_i\leq \frac{1}{l}, \forall i\in [k]\right\}$ for some integer $l\in [k]$. 
		If $l=1$, $\calB = \Delta_k$, and in this case a point is assigned to its nearest center.
		If $l>1$, the cheapest way to assign $p$ is to connect it to the $l$ nearest centers, which captures fault-tolerant clustering. 
		\sloppy
		Such fault-tolerant constraints have been studied extensively in a variety of clustering problems~\cite{khuller2000fault,SwamyS08,hajiaghayi2016constant}.
		\item $\calB$ is a (scaled) matroid basis polytope. This is a significant generalization of the above constraint which corresponds to a uniform matroid polytope.
		As alluded before \cite{bredin2005deploying, ZhongG03,baev2008approximation}, 
		more advanced fault-tolerance requirements can be captured by  {\em laminar matroid} basis polytope.
		For example, suppose $P_1,\ldots, P_g$ is a partition of $[k]$ such that $|P_j|=k_j$ (thus $\sum_j k_j=k$). Consider the partition matroid
		polytope
		$
		\calB=\left\{x\in \Delta_k: \sum_{i\in P_j} x_i\leq \frac{1}{l}, 
		\forall j\in [g], \sum_i x_i=1\right\}.
		$
		This captures the clustering problem with advanced fault-tolerant constraints (mentioned in Section~\ref{subsec:contribution}).
	\end{enumerate}
	%
	
	\noindent
	{\bf General Assignment Constraints:}
	In this paper, we consider the case where all points $p$ are subject to the same
	assignment structure constraint $\calB$.
	When an assignment function $\sigma$ satisfies both capacity constraint $\avg$
	and assignment structure constraint $\calB$, we denote it as $\sigma \sim (\calB, \avg)$,
	and call it a \emph{general assignment constraint}.
	\color{black}
	Given a general assignment constraint $(\calB,\avg)$, we define the cost of center set $C$ 
	for point set $P$ as
	\begin{equation}
		\label{eqn:cost}
		\cost_z(P, C, \calB, \avg) := \min_{\sigma \sim (\calB, \avg)} \cost_z^\sigma(P, C).
	\end{equation}
	That is, for the center set $C$,
	the cost is computed via the min-cost assignment $\sigma$ that is consistent with both constraints $\calB$ and $\avg$.


	\subsection{Coresets for Clustering with General Assignment Constraints}
	
	%
	In our paper, a coreset $S\subset P$ is a weighted subset with weight
	$w_S(\cdot): S\rightarrow \mathbb{R}_{\geq 0}$.
	In particular, the weight should satisfy 
	$w_S(S):= \sum_{p\in S} w_S(p) = n$.
	We need to extend the 
	general assignment constraint to handle weighted points.
	In particular, for a point set $S$ with weight $w_S(\cdot)$,
	if we require a point $p$ fully assigned,
	the corresponding constraint is $\|\sigma(p,\cdot)\|_1=w_S(p)$.
	For a center set $C$ and a capacity constraint $h$,
	the capacity and total capacity constraints are the same: 
	we still require that $\sum_{p} \sigma(p,\cdot) = h$
	and $\|\sigma\|_1=\|h\|_1 =n - m$.
	For the assignment structure constraint, given convex body $\calB \subseteq \Delta_k$,
	we say an assignment function $\sigma : P \times C \to \mathbb{R}_{\geq 0}$
	is consistent with $\calB$, denoted as $\sigma \sim \calB$,
	if 
	$    \sigma(p, \cdot) \,\,
	\in  \,\, 
	w_S(p)\cdot \calB $
	for every  $p \in S$ (note that this implies
	$\|\sigma(p, \cdot)\|_1\leq w_S(p)).$

	Now, we provide the definition of coreset for our problem.
	
	\sloppy
	\begin{definition}[\bf Coreset]
		\label{def:coreset}
		For a dataset $P \subseteq \calX$ of size $n$, an outlier number $0\leq m\leq n$ and an assignment structure constraint $\calB$,
		an \emph{$(\eps,\calB,m)$-coreset for clustering with general assignment constraints} 
		is a (weighted) set $S \subseteq P$ that satisfies
		\begin{equation}
			\label{eqn:coreset}
			\forall C \in \calX^k, \avg \in (n-m) \cdot \calB,  \quad \cost_z(S, C, \calB, \avg) \in (1 \pm \eps) \cdot \cost_z(P, C, \calB, \avg),
		\end{equation}
		where $\cost_z(P, C, \calB, h)$ is defined as in \eqref{eqn:cost}.
	\end{definition}

	\noindent
	%
	%
	Note that 
	for coreset $S$, the cost $\cost_z(S, C, \calB, \avg)$ is also computed according to \eqref{eqn:cost}, with the understanding of 
	$\sigma \sim (\calB, \avg)$ defined for weighted points.
	The coreset guarantee states that 
	\eqref{eqn:coreset} should hold for {\em all} capacity vectors $\avg \in (n - m) \cdot \calB$. 
	Requiring $\avg \in (n - m) \cdot \calB$ is necessary since we only need to focus on  \emph{feasible} capacity vector $\avg$ for which there exists $\sigma$ such that 
	$\sigma \sim (h,\calB)$.\footnote{To see this, note that we require $\sigma(p,\cdot)\in  \frac{n-m}{n}\cdot \calB$, hence $\avg= \sum_{p\in P} \sigma(p,\cdot) \in (n - m) \cdot \calB$.
	}

	\subsection{Handling Capacitated and Fair Clustering}
	\label{subsec:capacitated}
	
	It is important to note that our notion of coreset preserves the cost for all centers $C$ and all feasible capacity constraints $\avg$ {\em simultaneously},
	and for an assignment structure constraint $\calB$ given in advance.
	Hence, as also noted in previous works~\cite{BandyapadhyayFS21,braverman2022power},
	the guarantee that the cost is preserved for all $h$ simultaneously
	implies that such a coreset simultaneously captures all types of upper/lower bound capacity constraints
	of the form $\ell_c \leq \|\sigma(\cdot, c)\|_1 \leq u_c$; summarized by the following claim.

	\begin{claim}[\bf{Capacitated Clustering}]
		\label{claim:capacity}
		Consider the \kzC problem with capacity upper/lower bound constraint for each center $c$ of the form $\ell_c \leq \|\sigma(\cdot, c)\|_1 \leq u_c$, assignment structure constraint $\calB \subseteq \Delta_k$,
		and a total capacity constraint $\|\sigma\|_1=n-m$.
		An $(\eps,\calB,m)$-coreset $S$ is an $\eps$-coreset for this constrained clustering problem, i.e., for any center set $C\in \calX^k$,
		\[
		\min_{\substack{\sigma: ~ \|\sigma\|_1 = n - m, \sigma\sim \calB \\ \quad \quad \ell_c \leq \|\sigma(\cdot, c)\|_1 \leq u_c, \forall c\in C} } \cost_z^\sigma(S,C) \in (1\pm \eps)\cdot \min_{\substack{\sigma: ~ \|\sigma\|_1 = n - m, \sigma\sim \calB \\ \quad \quad \ell_c \leq \|\sigma(\cdot, c)\|_1 \leq u_c, \forall c\in C} } \cost_z^\sigma(P,C).
		\]
	\end{claim}
	
	\begin{proof}
		The proof can be found in Appendix \ref{sec:sufficient}.
	\end{proof}
	
	\noindent
	This claim implies that we do not need to specify capacity upper/lower bounds in the definition of our coreset.
	%
	
	Similarly, the following claim shows that our coreset works for fair clustering via a reduction of \cite[Theorem 4.2]{huang2019coresets}.
	Recall that in the fair clustering problem, 
	there are $s$ groups $G_1,\ldots, G_s\subseteq P$  (groups can be non-disjoint).
	We constrain that the proportion of points from a group $G_j$ 
	contained in a cluster centered at $c$, $\frac{|\sigma^{-1}(c) \cap G_i|}{|\sigma^{-1}(c)|}$, has a group-wise upper bound $u_i$ and lower bound $\ell_i$ with $0\leq \ell_i\leq u_i\leq 1$.

	\begin{claim}[\bf{Fair Clustering}]
		\label{claim:fair}
		Let $P\subseteq \calX$ be a dataset and $G_1,\ldots, G_s\subseteq P$ be $s$ groups (groups can be non-disjoint).
		We denote $\mathcal{G}_p = \left\{j\in [s]: p\in G_j\right\}$ to be the collection of groups that point $p\in P$ belongs to, and let $\Gamma:= |\{\mathcal{G}_p: p\in P\}|$ denote the number of distinct $\mathcal{G}_p$'s. 
		Consider the \kzC problem with fairness upper/lower bound constraint for each group $G_i$ and each center $c$ of the form $\ell_i \leq \frac{|\sigma^{-1}(c) \cap G_i|}{|\sigma^{-1}(c)|} \leq u_i$, and a total capacity constraint $\|\sigma\|_1=n-m$. 
		Suppose there exists $S$ of size at most $A$, which is an $(\eps,\Delta_k,m')$-coreset for clustering with general assignment constraints for all $0\leq m'\leq m$. 
		Then there exists an $\eps$-coreset of size at most $\Gamma A$ for this fair clustering problem with outliers, i.e., for any center set $C\in \calX^k$,
		\[
		\min_{\substack{\sigma: ~ \|\sigma\|_1 = n - m \\ \quad \ell_i \leq \|\sigma(\cdot, c)\|_1 \leq u_i, \forall i\in [s], c\in C} } \cost_z^\sigma(S,C) \in (1\pm \eps)\cdot \min_{\substack{\sigma: ~ \|\sigma\|_1 = n - m \\ \quad \ell_i \leq \|\sigma(\cdot, c)\|_1 \leq u_i, \forall i\in [s], c\in C} } \cost_z^\sigma(P,C).
		\]
	\end{claim}
	
	\begin{proof}
		The proof can be found in Appendix \ref{sec:fairness}.
	\end{proof}
	
	\noindent
	Intuitively, we begin by partitioning the set \( P \) into \( \Gamma \) subsets \( P_1, \ldots, P_\Gamma \), where each subset \( P_j \) consists of points \( p \) sharing the same group label \( \mathcal{G}_p = \mathcal{G}_j \). 
	Let \( \avg^{(j)} \) denote the capacity vector constrained to subset \( P_j \). For every \( i \in [s] \) and \( c \in C \), we can establish that
	\[
	|\sigma^{-1}(c)| = \sum_{j=1}^{\Gamma} \avg^{(j)}_c \quad \text{and} \quad |\sigma^{-1}(c) \cap G_i| = \sum_{j=1}^{\Gamma} \mathrm{I}\left[i \in \mathcal{G}_j \right] \cdot \avg^{(j)}_c,
	\]
	where \( \mathrm{I}\left[\cdot \right] \) is the indicator function. 
	Thus, the capacity vectors \( \avg^{(1)}, \ldots, \avg^{(\Gamma)} \) determine whether an assignment satisfies fairness constraints, allowing us to represent fairness constraints through a set of capacity constraints on \( \avg^{(1)}, \ldots, \avg^{(\Gamma)} \). 
	This reduction approach introduces an additional factor \( \Gamma \) in the size of the coreset.
	
	The claim indicates that the coreset size for fair clustering only contains an additional multiplicative factor $\Gamma$.
	Different from Claim \ref{claim:capacity}, we do not handle assignment structure constraints and fairness constraints simultaneously since the reduction of \cite[Theorem 4.2]{huang2019coresets} relies on a decomposition of dataset $P$ into $\Gamma$ pieces while it is unknown how to decompose assignment constraints accordingly. 
	By our coreset construction (Algorithm \ref{alg:main}), if $S$ is an $(\eps,\Delta_k,m)$-coreset, then $S$ is always an $(\eps,\Delta_k,m')$-coreset for all $0\leq m'\leq m$.
	Thus, combining with our coreset results for capacitated clustering, Claim \ref{claim:fair} leads to the results for fair clustering in Table \ref{tab:result}.

	Throughout, our goal is to obtain an $(\eps,\calB,m)$-coreset for fixed $\calB$ and $m$. 
	If the context is clear,  we also use the shorthand $\eps$-coreset in replacement of $(\eps, \calB, m)$-coreset.
	%
	
	
	
	\section{Coresets for Clustering  with General Assignment Constraints}
	\label{sec:main_thm}
	
	%

	%
	
	Before stating our main theorem, we first introduce a special case, 
	the \kzmC\ problem.
	The \kzmC\ problem is the \kzC\ problem with a total capacity constraint $\|\sigma\|_1 = n - m$, and the goal is to find a 
	a center set $C^\star\in \calX^k$ and an assignment function $\sigma^\star: P\times C\rightarrow \R_{\geq 0}$ with $\|\sigma^\star\|_1 = n - m$ that solve the following problem:
	\[
	\min_{C\in \calX^k, \sigma: \|\sigma\|_1 = n - m} \cost_z^\sigma(P,C).
	\]
	The problem reduces to vanilla \kzC when $m = 0$.
	Similar to prior studies on this problem \cite{braverman2022power,Huang2022NearoptimalCF}, 
	we need to introduce a tri-criteria approximation algorithm for \kzmC 
	for constructing coresets. 
	
	\begin{definition}[\bf{$(\alpha,\beta,\gamma)$-Approximation for \kzmC}]
		\label{def:tri}
		Let $P\subset \R^d$ be a dataset and $\alpha,\beta,\gamma\geq 1$ be constants. 
		An $(\alpha,\beta,\gamma)$-approximation of $P$ for \kzmC is a center set $C^\star\subset \R^d$ of size at most $\beta k$ such that 
		$
		\cost^{(\gamma m)}_z(P,C)\leq \alpha \cdot \OPT_{k,z,m}(P),
		$
		\sloppy
		where $\cost_z^{(m)}(P,C) :=\min_{\sigma: \|\sigma\|_1 = n-m} \cost_z^\sigma(P,C)$ and $\OPT_{k,z,m}$ denotes the optimal value of \kzmC.
	\end{definition}
	
	\noindent
	Please refer to \cite[Appendix A]{Huang2022NearoptimalCF}
	for more discussions of such approximation algorithms.
	For instance, a $(2^{O(z)}, O(1), O(1))$-approximation can be constructed in near-linear time~\cite{bhaskara2019greedy}.
	
	\vspace{0.1cm}
	Now, we are ready to state our main theorem.

	\begin{theorem}[\bf{Coresets for clustering with general assignment constraints}]
		\label{thm:coreset}
		Let $(\calX,d)$ be a metric space, $k\geq 1, m\geq 0$ be  integers, and $z\geq 1$ be a constant.
		Let $\eps,\delta\in (0,1)$ and $\calB\subseteq\Delta_k$ be 
		a convex body specifying the assignment structure constraint. 
		%
		%
		There exists a randomized algorithm that given a dataset $P\subseteq \calX$ of size $n\geq 1$ and an $(2^{O(z)},O(1),O(1))$-approximation $C^\star\in \calX^k$ of $P$ for \kzmC, constructs an $(\eps,\calB,m)$-coreset for \kzC with general assignment constraints of size
		\begin{align}
			\label{eq:size1}
			O(m) + 2^{O(z\log z)}\cdot \tilde{O}(\lip(\calB)^2\cdot (\cc{\eps} + k + \eps^{-1}) \cdot k^2\eps^{-2z}) \cdot \log \delta^{-1},
		\end{align}
		in $O(nk)$ time with probability at least $1-\delta$, where $\lip(\calB)$ is the Lipschitz constant of $\calB$ defined in \Cref{def:mass}, $\cc{\eps}$ is the covering exponent of $\calX$ defined in Definition~\ref{def:covering_wo}, and $\tilde{O}$ hides a $\poly \log (\lip(\calB)\cdot \cc{\eps}\cdot k\eps^{-1})$ term.
		Moreover, when $\calB = \Delta_k$ (in this case, $\lip(\calB)=1$), the coreset size can be further improved to
		\begin{align}
			\label{eq:size2}
			O(m) + 2^{O(z\log z)}\cdot \tilde{O}((\cc{\eps} + \eps^{-1}) \cdot k^2\eps^{-2z}) \cdot \log \delta^{-1}.
		\end{align}
	\end{theorem}

	\color{black}
	
	\noindent
	Our theorem provides the first coreset construction for capacitated and fair \kzC with $m$ outliers, improves the previous coreset size for capacitated/fair/robust \kzC by at least a factor of $k\eps^{-z}$, and establishes the first coreset construction for fault-tolerant clustering in various metric spaces; see ``Notable Concrete Results'' in Section~\ref{subsec:contribution} for more details.
	
	For ease of analysis, we assume the given dataset $P$ is unweighted.
	This assumption can be removed by a standard process of scaling the point weights to large integers,\footnote{We suppose all weights are rational numbers such that we can round them to integers. If not, we can always replace it with a sufficiently close rational number such that the slight difference does not affect the clustering cost.} and treating each weight as a multiplicity of a point; details can be found in~\cite[Corollary 2.3]{cohen2021new} and~\cite[Section 6.1]{cohen2022towards}.
	By this theorem, the coreset size is decided by the Lipschitz constant $\lip(\calB)$ and the covering exponent 
	$\cc{\eps}$.
	If both $\cc{\eps}$ and $\lip(\calB)$ are independent of $n$, e.g., upper bounded by $\poly(k,\eps^{-1})$, our coreset size is at most $\poly(k,\eps^{-1})$.
	When there are no structure constraints ($\calB = \Delta_k$ and $\lip(\calB) = 1$), we highlight that the coreset size in Equation \eqref{eq:size2} is better than that in Equation \eqref{eq:size1}, by reducing the factor of $(\cc{\epsilon} + k + \epsilon^{-1})$ to $(\cc{\epsilon} + \epsilon^{-1})$.
	
	The coreset construction algorithm of \Cref{thm:coreset} is shown in \Cref{sec:framework}, and the proof of \Cref{thm:coreset} can be found in \Cref{sec:proof_main}.
	Now we provide the formal definitions of $\lip(\calB)$ and $\cc{\eps}$ that appear in the statement of Theorem~\ref{thm:coreset}.
	We discuss the Lipschitz constant  $\lip(\calB)$ in Section~\ref{sec:Lipschitz}, and show how to bound $\cc{\eps}$ in Section~\ref{sec:covering_metric}.
	%

	\paragraph{Lipschitz constant $\lip(\calB)$.}
	We first define optimal assignment transportation (\AT).
	Since this notion may be of independent interest, we present it in a slightly abstract way and we explain how it connects to our problem after the definition.
	Then, the key notion $\lip(\calB)$ is defined (in Definition~\ref{def:mass}) as the Lipschitz constant of this \AT procedure.

	\begin{definition}[\bf{Optimal assignment transportation}]
		\label{def:transportation}
		Given a convex body $\calB\subseteq \Delta_k$, two $k$-dimensional vectors 
		$\avg, \avg' \in \calB$,
		and a function $\sigma : [n] \times [k] \to \R_{\geq 0}$ such that $\| \sigma \|_1 = 1$
		and $\sigma \sim (\calB, \avg)$,
		the optimal assignment transportation is defined as the minimum total mass transportation
		$ \|\sigma - \sigma'\|_1 $
		from $\sigma$ to $\sigma'$,
		over functions $\sigma' : [n] \times [k] \to \R_{\geq 0}$ such that
		$\sigma'\sim (\calB,\avg')$ (see Definition~\ref{def:capacity_constraint}, \ref{def:structure_constraint}) and $\forall p \in [n]$,
		$\|\sigma'(p, \cdot) \|_1 = \|\sigma(p, \cdot)\|_1$. Namely,
		%
		%
		\[
		\AT(\calB, \avg, \avg', \sigma) := \min_{\substack{
				\sigma'\sim (\calB,\avg') : \\
				\forall p \in [n], \|\sigma'(p, \cdot)\|_1 = \|\sigma(p, \cdot)\|_1 \\
		}}
		\| \sigma - \sigma' \|_1.
		\]
	\end{definition}
	
	\noindent
	%
	To see how our problem is related to Definition~\ref{def:transportation}, we view $\sigma$ as an assignment function,
	and $[n]$ and $[k]$ are interpreted as a data set $P$ of $n$ points and a center set $C$, respectively.
	Without loss of generality, we can assume $n = 1$ by normalization, and the requirement $\|\sigma\|_1 = n = 1$ in Definition~\ref{def:transportation} is satisfied. 
	In fact, we choose to use $[n]$ and $[k]$ since the representation/identity of a point in the metric does not affect transportation; in other words, $\AT$ is oblivious to the metric space.
	The requirement of $\avg, \avg' \in \calB$ is to ensure the feasibility of \AT.
	%
	%
	%
	Intuitively, this \AT aims to find a transportation plan, that transports the minimum mass to turn an initial capacity vector $\avg$
	into a target capacity vector $\avg'$.
	However, due to the presence of the assignment structure constraint $\calB$, not all transportation plans are allowed.
	In particular, we are in addition required to start from a given assignment $\sigma \in \calB$ that is consistent with $\avg$,
	and the way we reach $\avg'$ must be via another assignment $\sigma' \in \calB$ (which we optimize).

	\begin{definition}[\bf{Lipschitz constant of \AT on $\calB$}]
		\label{def:mass}
		Let $\calB\subseteq \Delta_k$ be a convex body.
		We define the Lipschitz constant of \AT on $\calB$ as 
		$
		\lip(\calB) := \max_{\substack{
				\avg, \avg'\in  \calB, \\
				\sigma\sim (\calB, \avg) : \|\sigma\|_1 = 1} }
		\frac{\AT(\calB, \avg, \avg', \sigma)}{\|\avg - \avg'\|_1}.
		$
	\end{definition}
	
	\noindent
	Note that $\lip(\calB)\geq 1$ since $\sum_{p\in [n]} (\sigma(p,\cdot) - \sigma'(p,\cdot))
	= \avg - \avg'$.
	%
	%
	%
	It is not hard to see that the Lipschitz constant is scale-invariant, i.e., $\lip(\calB) = \lip(c\cdot \calB)$ for any $\calB$ and $c > 0$.
	\paragraph{Covering Exponent $\cc{\eps}$.}
	We now introduce the notion of covering for all center points $c\in \calX$.
	Furthermore, we also define a quantity called \emph{covering number} to measure the size/complexity of the covering,
	and we use a more convenient \emph{covering exponent} to capture the worst-case size of the covering, which can serve as a parameter of the complexity for the metric space.

	\begin{definition}[\bf{Covering, covering number and covering exponent}]
		\label{def:covering_wo}
		%
		Let $P\subset \calX$ be an unweighted set of $n\geq 1$ points and  $a\in \calX$ be a point.
		%
		%
		Let $r_{\max} = \max_{p\in P} d(p,a)$ such that $P\subseteq \Ball(a,r_{\max})$.
		We say a collection $\calC\subseteq \Ball(a, 48z r_{\max}\eps^{-1})$ of points is an $\alpha$-covering of $P$ if for every point $c\in \Ball(a, 48z r_{\max}\eps^{-1})$, there exists some $c'\in \calC$ such that
		$
		\max_{p\in P} |d(p,c) - d(p,c')|\leq \frac{\alpha r_{\max}}{12z}.
		$

		Define $\cover(P,\alpha)$ to be the minimum cardinality $|\calC|$ of any $\alpha$-covering $\calC$ of $P$.
		Define the $\alpha$-covering number of $\calX$ to be 
		$
		\cover(n, \alpha):= \max_{P \subseteq \calX: |P| = n} \cover(P,\alpha),
		$
		i.e., the maximum $\cover(P,\alpha)$ over all possible unweighted sets $P$ of size $n$.

		Moreover, we define the $\alpha$-covering exponent of 
		the metric space $(\calX, d)$, denoted by $\cc{\alpha}$, to be the least integer $\gamma\geq 1$ such that $\cover(n,\alpha) \leq O(n^{\gamma})$ holds for any $n\geq 2$.
	\end{definition}
	
	\noindent
	Roughly speaking, an $\alpha$-covering $\calC$ is an discretization of the continuous space of
	all possible centers within $\Ball(a, 48z r\eps^{-1})$ of a large radius w.r.t. the $\ell_{\infty}$-distance differences from points $p\in P$ to $c\in \Ball(a, 48z r\eps^{-1})$.
	For those center $c\notin \Ball(a, 48z r\eps^{-1})$, we can verify that the distances between every point $p\in P$ and $c$ are very close, i.e., $d^z(p,c)\in (1\pm \eps)\cdot d^z(q,c)$ for any two points $p,q\in P$.
	This observation enables us to safely ``ignore'' the complexity of these remote centers in coreset construction, and hence we only need to consider centers within $\Ball(a, 48z r\eps^{-1})$.

	Since $|\calC|\leq |\calX|$, we have that the covering exponent $\cc{\alpha}\leq \log |\calX|$.
	%
	%
	This notion is closely related to other combinatorial dimension notions such as the VC-dimension and the (fat) shattering dimension of the set system formed by all metric balls which have been extensively considered in previous works, e.g.~\cite{langberg2010universal,feldman2011unified,FSS20,huang2018epsilon,BJKW21}.
	We will show the relations between the covering exponent in our setting 
	and two well-studied dimension notions, 
	the shattering dimension and the doubling dimension, in Section~\ref{sec:covering_metric}.
	%
	%
	\subsection{Overview of the Proof of Theorem \ref{thm:coreset}}
	\label{sec:tech_overview}
	
	We provide an overview of the proof of Theorem \ref{thm:coreset}.
	
	\subsubsection{Improved Hierarchical Uniform Sampling Framework}
	\label{sec:technical_algorithm}
	
	We take Euclidean \kMedian as an example, whose idea can be extended to \kzC via the generalized triangle inequality (\Cref{lm:triangle}).  
	The following analysis is for capacity constraints, which can be extended to general assignment constraints in Section~\ref{sec:technical_handling}.
	For simplicity, we use $\cost$ to represent $\cost_1$.
	
	\sloppy
	\paragraph{Review: Hierarchical Uniform Sampling Framework~\cite{braverman2022power,Huang2022NearoptimalCF}}
	We briefly review the coreset algorithm in \cite{braverman2022power} now. 
	We first compute an $O(1)$-approximation $C^\star=\{c_1^\star,\cdots,c_k^\star\}$ for the \kMedian problem and partition dataset $P$ into $k$ clusters $P_1,\cdots,P_k$. 
	Then adopting the idea of \cite{braverman2022power} (\Cref{thm:meta}), we decompose every $P_i$ into a collection $\calR_i$ of $\tilde{O}(k\epsilon^{-1})$ disjoint \emph{rings} and a collection $\calG_i$ of $\tilde{O}(k\epsilon^{-1})$ disjoint \emph{groups} centered at $c_i^\star$, and reduce the problem to constructing coresets for $\calR_i$ and $\calG_i$ (\Cref{thm:meta}). 
	To this end, Braverman et al.~\cite{braverman2022power} takes a uniform sample $S_R$ on every ring $R\in\calR_i$ with a uniform size $\Gamma=\tilde{O}(k\epsilon^{-4})$ of samples\footnote{In their original paper, \cite{braverman2022power} claims a bound $\Gamma=\tilde{O}(k\epsilon^{-5})$ and in a follow-up work \cite{Huang2022NearoptimalCF} it is shown that $\tilde{O}(k\epsilon^{-4})$ is already a sufficient choice} such that for every center set $C\subset \R^d$ and every capacity constraint $h$, 
	\begin{align}
		\label{eq:single_error}
		|\cost(R,C,h)-\cost(S_R,C,h)|
		\leq  \eps (\cost(R,C,h) + \cost(R,c_i^*));
	\end{align}
	and construct a two-point coreset $S_G$ on every group $G\in \calG_i$.
	Later on, Huang et al.~\cite{Huang2022NearoptimalCF} showed this framework also works for \kMedian with $m$ outliers (the only difference is that we need to compute an $O(1)$-approximation $C^\star=\{c_1^\star,\cdots,c_k^\star\}$ for the \kMedian problem with $m$ outliers in the first step).
	In previous work \cite{braverman2022power,Huang2022NearoptimalCF}, it is unknown whether Inequality~\eqref{eq:single_error} still work
	if we allow $m$ outliers.
	
	\paragraph{New Idea: Adaptive Sample Sizes for Rings}
	The main improvement of our algorithm (\Cref{alg:main}) compared to~\cite{braverman2022power,Huang2022NearoptimalCF} is a significant decrease in the number of samples for the rings.
	The key observation is that we only need the total error of rings to be upper bounded, i.e., 
	\begin{align}
		\label{eq:total_error}
		\sum_{R\in \calR_i} |\cost(R,C,h)-\cost(S_R,C,h)|
		\leq \epsilon \left (\sum_{R\in \calR_i} \cost(R,C,h)\right ) + \epsilon\cost(P,c_i^*),
	\end{align}
	which is intuitively much easier to be satisfied than Inequality~\eqref{eq:single_error}.
	Concretely, we regard $\calR_i$ as a whole instead of independent rings and focus on ensuring Inequality~\eqref{eq:total_error}. 
	This point of view enables us to adaptive select the sample size $\Gamma_R$ for every ring $R\in \calR_i$ (Line 6 of \Cref{alg:main}), say $\Gamma_R=\tilde{O}(k\epsilon^{-4})\cdot \lambda_R$ where $\lambda_R=\cost(R,c_i^*)/\cost(P_i,c_i^*)$ is the relative contribution of $R$ to $P_i$. 
	Then our coreset size is dominated by $\sum_{i\in [k], R\in \calR_i} \Gamma_R\leq \tilde{O}(k^2\epsilon^{-4})$, saving a factor of $k \eps^{-1}$ compared to that of \cite{braverman2022power,Huang2022NearoptimalCF}. 
	
	\paragraph{Handling Rings}
	The most technical part is to show that Inequality~\eqref{eq:total_error} always holds by our construction, and we introduce the ideas now.
	We remark that the estimation error $|\cost(R,C,h)-\cost(S_R,C,h)|$ of $R = \ring(c^\star_i, r, 2r)$ is likely to be decided by those centers $c\in C$ that are both ``not too far'' from $R$ and ``not too close'' to $c^\star_i$.
	This intuition motivates us to consider the number of ``effective centers'' to $R$, denoted as the level $t_R(C):=|C \cap \ring(c_i^\star,\frac{\epsilon r}{48},\frac{48r}{\epsilon})|$ of $R$ w.r.t. $C$ (\Cref{def:ring_level}). 
	The idea of excluding centers outside $B(c^\star_i, \frac{48r}{\epsilon})$ has been applied in~\cite{braverman2022power} to handle rings, and of excluding centers within $B(c^\star_i, \frac{\epsilon r}{48})$ is new.
	An immediate geometric observation is the following bound of the total level for any center set $C\subset \R^d$ (\Cref{lem:total_level}): 
	\begin{align}
		\label{eq:total_lv}
		\begin{aligned}
			\sum_{R\in \calR_i} t_R(C)
			\leq O(k\log \epsilon^{-1}),
		\end{aligned}
	\end{align}
	due to the ring structure of $\calR_i$.
	This property is somewhat surprising since the positions of $C$ seem to be arbitrary, and is a key for our improvement.
	We remark that the analysis of \cite{braverman2022power} simply bounds $t_R(C)$ by $k$, and hence, leads to a bound $\sum_{R\in\calR_i} t_R(C)\leq \tilde{O}(k^2\epsilon^{-1})$, which is a factor of $k\eps^{-1}$ larger than Inequality~\eqref{eq:total_lv}. 
	This factor also matches our improvement in the coreset size. 
	Our key lemma (\Cref{lm:ring}) shows that the estimation error of $R$ is ``proportional to'' $t_R(C)$ such that the total error is ``proportional to'' $\sum_{R\in \calR_i} t_R(C)$, which is small by the above observation.
	
	Overall, our geometric observation of $t_R(C)$ (Inequality~\eqref{eq:total_lv}) enables us to tolerant a larger estimation error for every ring that is captured by $\lambda_R$, and the total estimation error is under control due to a combined consideration of $t_R(C)$ and $\lambda_R$.

	\paragraph{Handling Groups}
	We expand upon the analysis presented in \cite{Huang2022NearoptimalCF} for the \kMedian problem with $m$ outliers, specifically addressing the inclusion of capacity constraints. 
	A crucial geometric observation made in \cite{Huang2022NearoptimalCF} is that, in the context of any center set $C\subset \mathbb{R}^d$, there are at most two ``special uncolored groups'' which intersect the outliers in a partial manner. 
	However, this property does not hold when considering capacity constraints with outliers.
	Fortunately, we can overcome this limitation by dividing the groups into $k$ equivalent classes based on their corresponding remote centers (\Cref{lm:equivalent}). 
	Consequently, the number of special uncolored groups is bounded by $O(k)$, thereby achieving the desired error bound for the groups (\Cref{lm:group}).
	
	\paragraph{Extension to General Metric Spaces}
	We apply a high-level idea in \cite{braverman2022power} to discretize the hyper-parameter space $\calX^k\times (|R|\cdot\calB)$ into a small number of representative pairs (\Cref{lm:relation}),
	and show that the error is preserved for every representative pair $(C,\avg)$ with very high probability (using e.g., concentration bounds).
	Our discretization of the center space $\calX^k$ directly relates to the \emph{covering exponent} $\cc{\eps}$ of the metric space, which is a complexity measure of general metric spaces, instead of a more geometric discretization based on the notion of metric balls which are more specific to Euclidean spaces. 
	Notably, this feature enables us to handle capacitated and fair clustering on any metric that admits a small shattering dimension or doubling dimension, while the analysis of \cite{braverman2022power} is specific to metric space with bounded doubling dimension due to the requirement of the existence of a small $\epsilon$-net.

	\subsubsection{Handling Assignment Structure Constraints}
	\label{sec:technical_handling}
	
	Next, we provide some explanations on why the algorithm can handle assignment structure constraints.
	We discuss rings and show how the additional term $\lip(\calB)$ appears in the coreset size.
	More detailed analysis for rings and groups can be found in Sections \ref{sec:proof_ring} and \ref{sec:proof_group}, respectively.

	\paragraph{Complexity Measure of Assignment Structure Constraint: $\lip(\calB)$}
	We start with a more detailed discussion on the new notion $\lip(\calB)$ due to its conceptual and technical importance.
	The definition of $\lip(\calB)$ (Definitions~\ref{def:transportation} and~\ref{def:mass}) may be interpreted in several ways.
	We start with an explanation from a technical perspective of coreset construction.
	A natural way of building a coreset which we also use,
	is to draw independent samples $S$ from data points $P$ (and re-weight).
	To analyze $S$, let us fix some capacity constraint $\avg$ and some center $C$.
	Let $\sigma^*$ be an assignment of $\avg$, i.e., $\cost_z(P, C, \calB, \avg) = \cost^{\sigma^*}_z(P, C)$.
	Then, one can convert $\sigma^*$ to $\sigma : S \times C \to \R_{\geq 0}$ for the sample $S\subseteq P$, 
	by setting $\sigma(p,\cdot) := w(p) \cdot \sigma^*(p,\cdot)$ and $w(p) = \frac{n}{|S|}$ for $p \in S$
	(where we can guarantee that $w(S)= n$).
	Even though this assignment $\sigma$ may slightly violate the capacity constraint $\avg$,
	the violation, denoted as $\|\avg - \avg'\|_1$ where $\avg'$ is the capacity induced by $\sigma$,
	is typically very small (by concentration inequalities), and it may be charged to $\cost_z(P, C, \calB, \avg)$.
	However, we still need to transport $\sigma$ to $\sigma'$ so that it is consistent with $\avg'$.
	More precisely, we need to find $\sigma' : S \times C \to \R_{\geq 0}$ for the sample $S$ that satisfies $\sigma' \sim (\calB, \avg)$,
	such that the total transportation $ \| \sigma - \sigma' \|_1$\footnote{Here, we interpret $\sigma, \sigma'$ as vectors on $\R^{S \times C}$.} (which relates to the cost change), is minimized.
	This minimum transportation plan is exactly defined by the {\em optimal assignment transportation} (\AT) (Definition \ref{def:transportation}).
	Since we eventually wish to bound the \AT cost against the mentioned $\|\avg - \avg'\|_1$,
	our $\lip(\calB)$ is defined as the universal upper bound of the \AT cost 
	relative to $\|\avg - \avg'\|_1$ over all $\avg, \avg', \sigma$, which can also be viewed as the Lipschitz constant of \AT.
	

	Another perspective of interpreting the notion of \AT is via the well-known {\em optimal transportation}.
	Specifically, the minimum $L_1$ transportation cost for turning $\avg$ to $\avg'$ without any constraint is exactly $\|\avg-\avg'\|_1$, which equals to the minimum transportation cost from any assignment $\sigma\sim \avg$ to some $\sigma'\sim \avg'$.
	Hence, compared with optimal transportation,
	our notion adds additional requirements that the transportation plan must be inside of $\calB$, and from a given starting assignment $\sigma$ with $\sigma\sim \avg$.
	We care about the worst-case relative ratio $\lip(\calB)$,
	which measures how many times more expensive the constrained optimal transportation
	cost \AT than the optimal transportation cost $\|\avg - \avg'\|_1$.
	Even though notions of constrained optimal transportation were also considered in the literature, see e.g., ~\cite{tan2013optimal,korman2013insights,korman2015optimal,korman2015dual} and a survey~\cite[Chapter 10.3]{peyre2019computational} for more,
	we are not aware of any previous work that studies exactly the same problem.

	\subsection{The Coreset Construction Algorithm}
	\label{sec:framework}
	
	We introduce the algorithm used in Theorem~\ref{thm:coreset}.
	The algorithm improves the coreset algorithm proposed
	very recently by~\cite{braverman2022power,Huang2022NearoptimalCF}. 

	We need the following decomposition that is defined with respect to a dataset $Q\subseteq \calX$  
	and a given center $c\in \calX$.

	\begin{theorem}[\bf{Decomposition into rings and groups, \cite{braverman2022power}}]
		\label{thm:meta}
		Let $Q\subseteq \calX$ be a weighted dataset and $c\in \calX$ be a center point. 
		There exists an $O(|Q|)$ time algorithm $\mathrm{Decom}(Q,c)$ that outputs $(\calR^\star, \calG^\star)$ as a partition of $Q$, where $\calR^\star$ and $\calG^\star$ are two disjoint collections of sets such that $Q=(\cup_{R\in \calR^\star} R)\cup (\cup_{G\in \calG^\star} G)$. 
		Moreover, $\calR^\star$ is a collection of disjoint rings satisfying
		\begin{enumerate}
			\item $\forall R\in \calR^\star$, $R$ is a ring of form $R=R_i(Q,c)$ for some integer $i\in \Z\cup \{-\infty\}$, where $R_i(Q,c):=Q\cap \ring(c,2^{i-1},2^{i})$ for $i\in \Z$ and $R_{-\infty}(Q,c):=Q\cap \{c\}$;
			\item $\forall R\in \calR^\star$, $\cost_z(R,c)\geq (\frac{\eps}{6z})^z\cdot \frac{\cost_z(Q,c)}{k\cdot \log (48z\eps^{-1})}$;
			\item $|\calR^\star|\leq 2^{O(z\log z)}\cdot \tilde{O}(k\eps^{-z})$;
		\end{enumerate}
		and $\calG^\star$ is a collection of disjoint groups satisfying
		\begin{enumerate}
			\item $\forall G\in \calG^\star$, $G$ is the union of a few consecutive rings of $(Q,c)$ and all these rings are disjoint. 
			Formally, $\forall G\in \calG^\star$, there exists $l_G,r_G\in \Z^{*},l_G\leq r_G$ such that $G=\cup_{i=l_G}^{r_G} R_i(Q,c)$ and the intervals $\{[l_G,r_G],G\in\calG^\star\}$ are disjoint;
			\item $\forall G\in \calG^\star$, $\cost_z(G,c)\leq (\frac{\eps}{6z})^z\cdot \frac{\cost_z(Q,c)}{k\cdot \log (48z\eps^{-1})}$;
			\item $|\calG^\star|\leq 2^{O(z\log z)}\cdot \tilde{O}(k\eps^{-z})$.
		\end{enumerate}
	\end{theorem}
	
	\noindent
	Roughly speaking, we decompose $Q$ into rings w.r.t. $c$.
	The collection $\calR^\star$ contains those rings with ``heavy'' costs, say $\cost_z(R,c) > (\frac{\eps}{6z})^z\cdot \frac{\cost_z(Q,c)}{k\cdot \log (48z\eps^{-1})}$ for $R\in \calR^\star$.\footnote{In~\cite{braverman2022power}, they call these rings heavy rings or marked rings.}
	They also gather the remaining ``light'' rings and form the collection of groups $\calG^\star$ (see~\cite[Lemma 3.4]{braverman2022power}), and ensure that the cardinality $|\calG^\star|$ is upper bounded by $2^{O(z\log z)}\cdot \tilde{O}(k\eps^{-z})$. 

	For each group $G$, we provide the following data structure.

	\sloppy
	\begin{definition}[\bf{Two-point coreset, Line 5 of Algorithm 1 in \cite{braverman2022power}}] 
		\label{def:twopoints}
		For a weighted dataset/group $G\subset \mathbb{R}^d$ and a center point $c\in \mathbb{R}^d$, let $\pfar^G$ and $\pclose^G$ denote the furthest and closest point to $c$ in $G$. For every $p\in G$, compute the unique $\lambda_p\in [0,1]$ such that $d^z(p,c)=\lambda_p\cdot d^z(\pclose^G,c)+(1-\lambda_p)\cdot d^z(\pfar^G,c)$. 
		Let $D=\{\pfar^G,\pclose^G\}$, $w_D(\pclose^G)=\sum_{p\in G} \lambda_p$, and $w_D(\pfar^G)=\sum_{p\in G} (1-\lambda_p)$. $D$ is called the \emph{two-point} coreset of $G$ with respect to $c$.
	\end{definition}

	\noindent
	By definition, we know that $w(D) = |G|$ and $\cost_z(D, c) = \cost_z(G,c)$, which are useful for upper bounding the error induced by such two-point coresets.

	\begin{algorithm}
		\caption{$\HUS(P,k,z,\Gamma)$}
		\label{alg:main}
		\begin{algorithmic}[1]
			\Require An unweighted dataset $P\subseteq \calX$ of size $n\geq 1$, an integer $k\geq 1$, an integer $0\leq m\leq n$, constant $z\geq 1$, a sampling size $\Gamma \geq 1$, and an $\left(2^{O(z)},O(1), O(1)\right)$-approximate solution $C^\star = \left\{c^\star_1,\ldots,c^\star_k\right\}\subseteq \calX$ of $P$ for \kzmC 
			\Ensure 
			\State Let $L^*\leftarrow \arg \min_{L\subseteq P: |L| = m} \cost_z(P\setminus L, C^\star)$ denote the set of $m$ outliers of $P$ w.r.t. $C^\star$
			\State  Decompose $P\setminus L^\star$ into $k$ clusters $P_1,\ldots,P_k$ such that each $P_i$ contains all points in $P$ whose closest center in $C^\star$ is $c^\star_i$ (breaking ties arbitrarily). 
			\State  For each $i\in [k]$, apply the decomposition of Theorem~\ref{thm:meta} to $(P_i, c_i^\star)$ and obtain $(\calR_i,\calG_i)\leftarrow \mathrm{Decom}(P_i,c^\star_i)$, where $\calR_i$ is a collection of disjoint rings and $\calG_i$ is a collection of disjoint groups.
			\State For each $i\in [k]$ and each ring $R\in \calR_i$, set $\lambda_R \leftarrow \frac{\cost_z(R,c^\star_i)}{\cost_z(P_i, c^\star_i)}$ and take a uniform sample $S_R$ of size $\Gamma_R \leftarrow \lceil \Gamma\cdot \lambda_R \rceil$ from $R$, and set $w_{S_R}(x) = \frac{|R|}{\Gamma_R}$ for each point $x\in S_R$. 
			\State For each $i\in [k]$, for each group $G\in \calG_i$ and center $c_i^\star$, construct a two-point coreset $D_G$ of $G$ by Definition~\ref{def:twopoints}.
			\State Return $S\leftarrow L^\star\cup (\bigcup_{R} S_R)\cup (\bigcup_{G} D_G)$.
		\end{algorithmic}
	\end{algorithm}
	
	\paragraph{Hierarchical Uniform Sampling Coreset Framework}
	We are now ready to introduce the hierarchical uniform sampling framework $\HUS(P,k,z,\Gamma)$ (Algorithm~\ref{alg:main}).
	To simplify our analysis, we slightly abuse the notation and consider $C^\star$ as an $\left(2^{O(z)},1,1 \right)$-approximation instead of a tri-approximation, e.g., a $(2^{O(z)}, O(1), O(1))$-approximation by~\cite{bhaskara2019greedy}. 
	This simplification allows the algorithm to decompose the inliers $P\setminus L^\star$ into simply $k$ clusters rather than $O(k)$ clusters, which only results in a factor of $O(1)$ difference in the coreset size.
	%
	%
	We first compute an outlier set $L^\star$ w.r.t. $C^\star$ in Line 1 (as in~\cite[Algorithm 1]{Huang2022NearoptimalCF}).
	Note that we apply \Cref{thm:meta} for each $P_i$ ($i\in [k]$) in Line 3, and obtain collections $\calR_1,\ldots, \calR_k$ and $\calG_1,\ldots,\calG_k$.
	Let $\calR = \bigcup_{i\in [k]} \calR_i$ be the collection of all rings in different clusters and $\calG = \bigcup_{i\in [k]} \calG_i$ be the collection of all groups.
	By Theorem~\ref{thm:meta}, we have the following observation.
	
	\begin{observation}[\bf{Bound for $|\calR|$ and $|\calG|$}]
		\label{ob:ring_num}
		$|\calR|, |\calG| \leq 2^{O(z\log z)}\cdot \tilde{O}(k^2\eps^{-z})$.
	\end{observation}
	
	\noindent
	The primary improvement of Algorithm~\ref{alg:main} lies in Line 4, where we selectively choose the coreset size $\Gamma_R$ that is proportional to the relative contribution $\lambda_R$ of each ring $R$.  
	Moreover, we demonstrate that our algorithm, denoted as $\HUS(P,k,z,\Gamma)$, can generate an output $(S,w)$ that is a coreset for \kzC with general assignment constraints in general metric spaces, provided that we carefully choose the sample number $\Gamma$.
	
	We remark that Algorithm \ref{alg:main} presents a hierarchical uniform sampling framework instead of purely sampling.
	However, even in vanilla clustering, achieving state-of-the-art coreset size purely through sampling, without any pre-processing, has only been recently established \cite{bansal2024sensitivity}.
	It is interesting to study whether this purely sampling framework can be extended to constrained clustering. 
	
	\subsection{Proof of Theorem~\ref{thm:coreset}: Performance Analysis of \Cref{alg:main}}
	\label{sec:proof_main}
	
	We use the following well-known generalized triangle inequalities for \kzC; see e.g.,~\cite{munteanu2019coresets,braverman2022power} for more variants.
	The key is to bound the induced errors for rings and groups for every $k$-center set $C$ respectively (Lemmas \ref{lm:ring} and \ref{lm:group}).
	For groups, we provide a unified error bound that does not depend on the choice of $C$.
	In contrast, our error bounds for rings $R$ depend on the relative location of $C$ to $R$, captured by a new notion called the level $t_R(C)$ of rings $R$ w.r.t. $C$ (Definition \ref{def:ring_level}).
	Bounding the total error induced by rings relies on a new geometric observation that bounds the total levels $\sum_R t_R(C)$ for all $C$ (Lemma \ref{lem:total_level}).

	The following lemma is useful for bounding the total induced error of rings.
	
	\begin{lemma}[\bf{H\"older's Inequality}] 
		\label{lm:Holder}
		Assume $p,q>1$ such that $\frac{1}{p}+\frac{1}{q}=1$ then for every integer $n\geq 1$ and two sequences of numbers $a_1,a_2,\cdots,a_n,b_1,b_2,\cdots,b_n>0$,
		$$
		\sum_{i=1}^n a_ib_i\leq (\sum_{i=1}^n a_i^p)^{\frac{1}{p}}\cdot (\sum_{i=1}^n b_i^q)^{\frac{1}{q}}
		$$
	\end{lemma}

	%

	\noindent
	Now we are ready to prove Theorem~\ref{thm:coreset}.
	Like~\cite{braverman2022power}, we prove the theorem by analyzing rings and groups separately; summarized by Lemmas~\ref{lm:ring} and~\ref{lm:group} respectively.

	Given a ring $R \subset \ring(a, r, 2r)$ and a center set $C\in \calX^k$, we denote $\Cfar^R :=\left\{c\in C: d(c,a) \geq 48z \eps^{-1} r\right\}$ to be the centers that are remote to $R$ and $\Cclose^R := \left\{c\in C: d(c,a) \leq \eps r/48z\right\}$ to be the centers that are close to $a$.
	We first introduce the following important notion.
	
	\begin{definition}[\bf{Level of rings}]
		\label{def:ring_level}
		Given a ring $R\subset \ring(a,r,2r)$ for some $a\in \R^d$ and radius $r > 0$ and a center set $C\in \calX^k$, we denote the level of ring $R$ w.r.t. $C$ to be 
		$
		t_R(C) := |C\setminus \Cfar^R\setminus \Cclose^R|.
		$
	\end{definition}
	
	\noindent
	Note that $0\leq t_R(C)\leq k$.
	We will see that the level $t_R(C)$ heavily affects the induced error of samples $S_R$ w.r.t. $C$.
	By Theorem~\ref{thm:meta}, we directly have the following observation.

	\begin{lemma}[\bf{Bounding total levels}]
		\label{lem:total_level}
		Let $\eps\in (0,\frac{1}{4})$.
		Given a center set $C\subseteq \calX^k$, for every $i\in [k]$, we have that
		$
		\sum_{R\in \mathcal{R}_i} t_R(C)\leq 10zk\log \eps^{-1}.
		$
	\end{lemma}
	
	\begin{proof}
		We say center $c\in C$ is \emph{interesting} to ring $R$ if $c\in C\setminus \Cfar^R\setminus \Cclose^R$.
		It suffices to prove that every $c\in C$ can be interesting to at most $10z\log \eps^{-1}$ rings in $\calR_i$. 
		For the sake of contradiction, suppose $c$ is interesting to more than $8z\log \eps^{-1}+1$ rings. 
		Among these rings, let $R_1=P_i\cap \mathrm{ring}(c_i^\star,r_1,2r_1)$ and $R_2=P_i\cap \mathrm{ring}(c_i^\star,r_2,2r_2)$ denote rings with the largest and smallest radii respectively. 
		Recall that all rings are disjoint, thus we have $r_1/r_2>2^{8z\log \eps^{-1}}=\eps^{-8z}$. 
		However, as $c$ is interesting to both $R_1$ and $R_2$, by Definition~\ref{def:ring_level}, we know that 
		$$
		\frac{\eps r_1}{48 z}\leq d(c,c_i^\star)\leq \frac{48z r_2}{\eps }
		$$
		which implies $r_1/r_2\leq 2304 z^2 \eps^{-2}<\eps^{-8z}$ since $\eps<\frac{1}{4}$.
		
		So we conclude with a contradiction and have proved Lemma~\ref{lem:total_level}.
	\end{proof} 
	
	\noindent
	We have the following lemma that relates the induced error of rings with their levels.
	
	\begin{lemma}[\bf{Error analysis for rings}]
		\label{lm:ring}
		For each $i\in [k]$ and ring $R\in \calR_i$, suppose
		\[
		\Gamma_R \geq 2^{O(z\log z)}\cdot \lambda_R\cdot \lip(\calB)^2\cdot  (\cc{\eps} + k + \eps^{-1}) \cdot k\eps^{-2z} \cdot \log (\lip(\calB)\cdot \cc{\eps} \delta^{-1}) \log^7 (k \eps^{-1}),
		\]
		and when $\calB = \Delta_k$, suppose
		\[
		\Gamma_R \geq 2^{O(z\log z)}\cdot \lambda_R\cdot  (\cc{\eps} + \eps^{-1}) \cdot k\eps^{-2z} \cdot \log (\cc{\eps} \delta^{-1}) \log^7 (k \eps^{-1}),
		\]
		where $\Gamma_R$ is the sample size of $S_R$ as in Line 4 of Algorithm~\ref{alg:main}.
		With probability at least $1-\frac{\delta}{|\calR|}$, for every $k$-center set $C\in \calX^k$ and capacity constraint $h\in |R|\cdot \conv(\calB^o)$,
		\begin{align*}
			& \quad |\cost_z(R,C,\calB, h)-\cost_z(S_R,C,\calB,h)| \\
			\leq & \quad \eps \left(\cost_z(R,C,\calB,h) + \cost_z(R,c_i^\star)\right) +\big(\frac{t_R(C)}{10z k\lambda_R \log \eps^{-1}}\big)^{\frac{1}{2}}\cdot \eps \cost_z(R,c_i^\star).
		\end{align*}
	\end{lemma}
	
	\begin{proof}
		The proof can be found in Section~\ref{sec:proof_ring}.
	\end{proof}
	
	\noindent
	For groups, we have the following lemma.
	
	\begin{lemma}[\bf{Error analysis for groups}]
		\label{lm:group}
		For each $i\in [k]$, let $G[i] = \bigcup_{G\in \calG_i} G$ be the union of all groups $G\in \calG_i$ and $D[i] = \bigcup_{G\in \calG_i} D_G$  be the union of all two-point coresets $D_G$ with $G\in \calG_i$.
		For every $k$-center set $C\in \calX^k$ and capacity constraint $\avg\in |G[i]|\cdot \conv(\calB^o)$,
		\[
		\left|\cost_z(G[i], C, \calB, \avg) - \cost_z(D[i], C, \calB, \avg) \right| \leq O(\eps)\cdot \left(\cost_z(G[i], C, \calB, \avg) + \cost_z(P_i, c^\star_i) \right). 
		\]
	\end{lemma}

	\begin{proof}
		The proof can be found in Section~\ref{sec:proof_group}.
	\end{proof}

	\noindent
	Note that for groups, the induced error of two-point coresets $D[i]$ is deterministically upper bounded, which is not surprising since there is no randomness in the construction of $D[i]$.
	This property is quite powerful since we do not need to consider the complexity of center sets in different metric spaces when analyzing the performance of two-point coresets.
	Now we are ready to prove \Cref{thm:coreset}.
	%

	\begin{proof}[Proof of Theorem~\ref{thm:coreset}]
		By Lemma~\ref{lm:ring}, we can select 
		\[
		\Gamma = 2^{O(z\log z)}\cdot  \tilde{O}(\lip(\calB)^2\cdot  (\cc{\eps} + k + \eps^{-1} )\cdot k\eps^{-2z}) \cdot \log \delta^{-1}
		\]
		for general assignment structure constraint $\calB$ and select
		\[
		\Gamma = 2^{O(z\log z)}\cdot  \tilde{O}( (\cc{\eps} + \eps^{-1} )\cdot k\eps^{-2z}) \cdot \log \delta^{-1}
		\]
		when $\calB = \Delta_k$, and apply $\HUS(P,k,z,\Gamma)$ that outputs $(S,w)$.
		We verify that $S$ is the desired $O(\eps)$-coreset.

		For the coreset size $|S|$, we first note that $\sum_{R\in \calR_i} \Gamma_R \leq \Gamma$ for every $i\in [k]$ by the definition of $\lambda_R$.
		Also by Observation~\ref{ob:ring_num}, $\sum_{G\in \calG} \Gamma_G\leq O(k^2 \eps^{-z})$.
		Hence, the size $|S|$ is dominated by
		$|L^\star| + \Gamma \cdot k$, which matches the coreset size in Theorem~\ref{thm:coreset}.
		For correctness, we first have the following claim by Lemma~\ref{lm:ring}.
		
		\begin{claim}
			\label{claim:total_error}
			With probability at least $1-\delta$, for every $i\in [k]$, for every center set $C\in \calX^k$ and capacity constraints $\{h_R\}_{ R\in \mathcal{R}_i}$ satisfying $\forall R\in \mathcal{R}_i, h_R\in |R|\cdot \conv(\calB^o)$, we have
			\begin{align*}
				\sum_{R\in\calR_i} |\cost_z(R,C,\calB,h_R)-\cost_z(S_R,C,\calB,h_R)|
				\leq \eps\sum_{R\in \mathcal{R}_i}\cost_z(R,C,\calB,h_R)+ 2\eps \cost_z(P_i,c_i^\star).
			\end{align*}
		\end{claim}
		
		\begin{proof}
			Assume Lemma~\ref{lm:ring} holds for all rings $R\in \calR$, whose success probability is at least $1-\delta$ by the union bound.
			Fix a center set $C\in \calX^k$ and capacity constraints $\{h_R\}_{R\in \calR}$. 
			By Lemma~\ref{lm:ring} and $\lambda_R = \frac{\cost_z(R,c^\star_i)}{\cost_z(P_i, c^\star_i)}$, we have that for every ring $R\in \calR_i$,
			
			\begin{align*}
				\begin{aligned}
					& \quad|\cost_z(R,C,\calB,h_R)-\cost_z(S_R,C,\calB,h_R)|\\
					\leq &\quad\eps \left(\cost_z(R,C,\calB,h_R) + \cost_z(R,c^\star_i)\right)+\big(\frac{t_R(C)}{10z k\lambda_R\cdot \log \eps^{-1}}\big)^{\frac{1}{2}}\cdot  \eps \cost_z(R,c_i^\star)\\
					= &\quad\eps \left(\cost_z(R,C,\calB,h_R) + \cost_z(R,c^\star_i)\right) + (\frac{t_R(C)}{10zk\log \eps^{-1}})^{\frac{1}{2}}\lambda_R^{\frac{1}{2}}\cdot \eps \cost_z(P_i,c_i^\star) 
				\end{aligned}
			\end{align*}
			Summing over all $R\in \mathcal{R}_{i,0}$ we have
			\begin{align*}
				\begin{aligned}
					&\quad\sum_{R\in\calR_i} |\cost_z(R,C,\calB,h_R)-\cost_z(S_R,C,\calB,h_R)|\\
					\leq&\quad \eps \sum_{R\in \calR_i} \left(\cost_z(R,C,\calB,h_R) + \cost_z(R,c^\star_i)\right)  + (\frac{t_R(C)}{10zk\log \eps^{-1}})^{\frac{1}{2}}\lambda_R^{\frac{1}{2}}\cdot \eps \cost_z(P_i,c_i^\star)\\
					\leq &\quad \eps \sum_{R\in \calR_i}\cost_z(R,C,\calB,h_R)+ 2\eps \cost_z(P_i,c_i^\star)
				\end{aligned}
			\end{align*}
			where for the last inequality, we are using Holder's inequality (Lemma~\ref{lm:Holder}) and the fact $$\sum_{R\in\calR_i} t_R(C)\leq 10zk \log \eps^{-1} \text{ and } \sum_{R\in \calR_i}\lambda_R\leq 1$$ to obtain 
			\begin{align*}
				& \quad\sum_{R\in \calR_i} (\frac{t_R(C)}{10zk\log \eps^{-1}})^{\frac{1}{2}}\lambda_R^{\frac{1}{2}}\\
				\leq &\quad 
				\big(\sum_{R\in \calR_i} \frac{t_R(C)}{10zk\log \eps^{-1}}\big)^{\frac{1}{2}}\cdot \big(\sum_{R\in \calR_i} \lambda_R\big)^{\frac{1}{2}}\\
				\leq &\quad 1.
			\end{align*}
			Thus, we prove Claim~\ref{claim:total_error}.
		\end{proof}

		\noindent
		Fix a $k$-center set $C\in \calX^k$, and a capacity constraint $\avg\in (n-m)\cdot \calB$. 
		Suppose a collection $h^L\cup \left\{\avg^R\in |R|\cdot \conv(\calB^o): R\in \calR\right\}\cup \left\{\avg^{(i)}\in |G[i]|\cdot \conv(\calB^o): G\in \calG\right\}$ of capacity constraints satisfy that
		\[
		h^L + \sum_{R\in \calR} \avg^R + \sum_{i\in [k]} \avg^{(i)} = \avg,
		\]
		and
		\begin{align}
			\label{eq1:proof_main}
			\cost_z(P,C,\calB,\avg) = \cost_z(L^\star, C, \calB, \avg^L) + \sum_{R\in \calR} \cost_z(R,C,\calB,\avg^R) + \sum_{i\in [k]} \cost_z(G[i],C,\calB,\avg^{(i)}).
		\end{align}
		We have
		\begin{align*}
			&  \quad \cost_z(S,C,\calB,\avg) & \\
			\leq & \quad\cost_z(L^\star, C, \calB, \avg^L) + \sum_{R\in \calR} \cost_z(S_R,C,\calB,\avg^R) + \sum_{i\in [k]} \cost_z(D[i],C,\calB,\avg^{(i)}) & (\text{by optimality}) \\
			\leq  & \quad\cost_z(L^\star, C, \calB, \avg^L) &\\
			& \quad+ (1+O(\eps))\cdot\sum_{R\in \calR}  \cost_z(R,C,\calB,\avg^R) + O(\eps)\cdot \sum_{i\in [k]} \cost_z(P_i,c^\star_i) & (\text{Claim~\ref{claim:total_error}})\\
			& \quad +  \sum_{i\in [k]} (1+O(\eps))\cdot\cost_z(G[i],C,\calB,\avg^{(i)}) + O(\eps)\cdot \cost_z(P_i,c^\star_i) & (\text{Lemma~\ref{lm:group}}) \\
			\leq & \quad (1+O(\eps))\cdot \cost_z(P,C,\calB,\avg) + O(\eps)\cdot \cost_z(P,C^\star) & (\text{Ineq.~\eqref{eq1:proof_main}}) \\
			\leq & \quad (1+O(\eps))\cdot \cost_z(P,C,\calB,\avg). & (\text{Defn. of $C^\star$})
		\end{align*}
		Similarly, we also have that $\cost_z(P,C,\calB,\avg) \leq (1+O(\eps))\cdot \cost_z(S,C,\calB,\avg)$.
		Thus, $(S,w)$ is indeed an $O(\eps)$-coreset.
		
		For the running time, Line 1 costs $O(nk)$ time. 
		Line 2 costs $\sum_{i\in [k]} O(|P_i|) = O(n)$ time by Theorem~\ref{thm:meta}.
		Line 3 costs $O(n)$ time.
		Line 4 costs $\sum_{G\in \calG_i} O(|G|) = O(n)$ time by Definition~\ref{def:twopoints}.
		Overall, the total time is $O(nk)$.
	\end{proof}

	\subsection{Proof of Lemma~\ref{lm:ring}: Error Analysis for Rings}
	\label{sec:proof_ring}
	
	Now we analyze the errors induced by rings.
	We first provide an overview of the proof.
	Motivated by \cite{braverman2022power}, we first observe that error induced by centers in $\Cfar^R\setminus \Cclose^R$ can be well controlled and we only need to bound the error induced by centers in $C\setminus \Cfar^R\setminus \Cclose^R$ (Lemma \ref{lm:approximation_cost}).
	We also define another notion of covering $\cover(R,t,\beta,\calB)$ (Definition~\ref{def:covering}), and relate it to $\cc{\beta}$ (Lemma \ref{lm:relation}).
	Then we prove the main technical lemma that bounds the estimation error $|\cost_z(R,C,\calB,\avg) - \cost_z(S_R,C,\calB,\avg)|$ for every ring $R$: Lemma \ref{lm:ring_weak}, which is a weak version of Lemma~\ref{lm:ring}.
	Compared to Lemma~\ref{lm:ring}, the sample size in Lemma \ref{lm:ring_weak} contains an additional term $\log n_R$, which can be removed by applying the iterative size reduction approach of~\cite{BJKW21}.

	To prove Lemma \ref{lm:ring_weak}, we consider two cases: $t_R(C) = 0$ and $t_R(C)\in [k]$, and discuss the induced errors respectively (Lemmas \ref{lm:ring_center} and \ref{lm:ring_pair}).
	The case of $t_R(C) = 0$ is easier and can be shown by McDiarmid's Inequality (Theorem \ref{thm:mcdiarmid}).
	The case of $t_R(C) \in [k]$ is much more involved.
	We extend the flow idea in~\cite{cohen-addad2019on} to handle assignment constraints.
	Our idea is to show that $\cost_z(S_R,C,\calB,\avg)$ concentrates on its expectation $\Exp_S\left[\cost_z(S_R,C,\calB,\avg)\right]$ (\Cref{lm:concentration}) and to show the expectation is very close to $\cost_z(R,C,\calB,\avg)$ (\Cref{lm:close}).
	The proof of the former about the expectation follows easily from concentration inequalities, but that of the latter one is much more difficult and constitutes a major part of our analysis.
	Let $\sigma^\star$ be an optimal assignment of $\avg$ on $R$, i.e., $\cost_z(R,C,\calB,\avg) = \cost_z^{\sigma^\star}(R,C)$.
	We convert $\sigma^\star$ to $\sigma: S_R\times C\rightarrow \R_{\geq 0}$, by setting $\sigma(p,\cdot) := w_S(p) \cdot \sigma^*(p,\cdot)$ for $p \in S$ (as mentioned in Section \ref{sec:technical_handling}), 
	and can show that $\cost_z^{\sigma}(S_R,C)\approx \cost_z^{\sigma^\star}(R,C)$ (\eqref{eq2:proof_claim_sample} in the proof of \Cref{claim:sample_upper_bound}).
	We are done if $\sigma\sim \avg$, but unfortunately, this generally does not hold.
	Hence, we turn to show the existence of an assignment $\sigma'\sim (\calB, \avg)$ on $S$ such that $|\cost_z^{\sigma}(S_R,C) - \cost_z^{\sigma'}(S_R,C)|$ is small enough.
	This existence of such $\sigma'$ is shown in Claim~\ref{claim:sample_upper_bound}, and here we sketch the main technical ideas.
	We reduce the problem of bounding $|\cost_z^{\sigma}(S_R,C) - \cost_z^{\sigma'}(S_R,C)|$ to bounding the mass movement $\|\sigma - \sigma'\|_1$ from $\sigma$ to $\sigma'$, based on a novel idea that we can safely ignore the total assignment cost of $R$ and $S$ to remote centers (denoted by $\Cfar$), and the difference between the assignment cost induced by $\sigma$ and $\sigma'$ to $C\setminus \Cfar$ is proportional to $\|\sigma - \sigma'\|_1$, due to the generalized triangle inequality (\Cref{lm:triangle}).
	Now it suffices to require that $\|\sigma - \sigma'\|_1\leq \tilde{O}_z(\eps^{z+1} |R|)$ for bounding $|\cost_z^{\sigma}(S_R,C) - \cost_z^{\sigma'}(S_R,C)|$.
	By the definition of $\lip(\calB)$, we only need to make sure that $\|\avg - \avg'\|_1$ ($\avg'$ is induced by $\sigma$) is as small up to $\tilde{O}_z(\eps^{z+1} |R|)/\lip(\calB))$ (\Cref{claim:psi_tiny}), which is again guaranteed by McDiarmid's Inequality (\Cref{thm:mcdiarmid}).
	The additional factor $1/\lip(\calB)$ results in the term $\lip(\calB)^2$ in our coreset size.

	For preparation, we introduce the following generalized triangle inequality.
	
	\begin{lemma}[\bf{Generalized triangle inequality~\cite[Lemma 2.1]{braverman2022power}}]
		\label{lm:triangle}
		Let $a,b,c\in \calX$ and $z\geq 1$.
		For every $t\in (0,1]$, the following inequalities hold:
		\[
		d^z(a,b) \leq (1+t)^{z-1} d^z(a,c) + \Bigl(1+\frac{1}{t}\Bigr)^{z-1} d^z(b,c),
		\]
		and
		\[
		\left| d^z(a,c) - d^z(b,c) \right| \leq t\cdot d^z(a,c) + \Bigl(\frac{3z}{t}\Bigr)^{z-1}\cdot d^z(a,b).
		\]
	\end{lemma}
	
	
	\noindent
	We also use the following concentration inequality for analysis.

	\begin{theorem}[\bf{McDiarmid's Inequality~\cite[Theorem 3.11]{Handel2014ProbabilityIH}}]
		\label{thm:mcdiarmid}
		Let $E$ be a ground set and $n \geq 1$ be an integer.
		Let $g: E^n\rightarrow \R$ be a function satisfying that
		for any sequence $(x_1,x_2,\ldots,x_n)\in E^n$, there exists a universal constant $\delta_i > 0$ for each $i\in [n]$ such that
		\[
		\delta_i \geq \sup_{y\in E} g(x_1,\ldots, x_{i-1}, y, x_{i+1},\ldots, x_n) - \inf_{y\in E} g(x_1,\ldots, x_{i-1}, y, x_{i+1},\ldots, x_n).
		\]
		Then for independent random variables $X_1,\ldots, X_n$, we have for every $t>0$,
		\[
		\Pr\bigl[g(X_1,\ldots, X_n) - \Exp_{S_R}\left[g(X_1,\ldots, X_n)\right]\geq t\bigr] \leq e^{-\frac{2t^2}{\sum_{i\in [n]} \delta_i^2}}.
		\]
	\end{theorem}
	
	Suppose $R\subseteq \ring(c^\star_i, r, 2r)$ for some $r>0$.
	For preparation, we have the following observation that shows that $\cost_z(P,C,\calB,\avg)$ has some Lipschitz property. The proof is the same as that in~\cite[Lemma 3.6]{Huang2022NearoptimalCF}.

	\begin{observation}[\bf{Lipschitz property of $\cost_z(P,C,\calB,\avg)$ on $P$}]
		\label{ob:Lipschitz}
		Let $P, Q\subseteq \calX$ be two weighted sets with $w_P(P) = w_{Q}(Q)$ and $c\in \calX$ be a center point.
		For any $k$-center set $C\in \calX^k$, any assignment constraint $(\calB,\avg)$ and any $\eps\in (0,1]$, we have
		\[
		\left|\cost_z(P,C,\calB,\avg) - \cost_z(Q,C,\calB,\avg) \right| \leq \eps\cdot \cost_z(P,C,\calB,\avg) + \left(\frac{6z}{\eps}\right)^{z-1} \cdot \left(\cost_z(P,c) + \cost_z(Q,c)\right).
		\]
	\end{observation}
	
	\noindent
	This observation is useful for providing an upper bound for $\left|\cost_z(P,C,\calB,\avg) - \cost_z(Q,C,\calB,\avg) \right|$.
	For instance, if $z=1$ and $\cost_1(P,c)+\cost_1(Q,c)\ll \cost_1(P,C,\calB,\avg)$, we may have 
	\[
	\left|\cost_1(P,C,\calB,\avg) - \cost_1(Q,C,\calB,\avg) \right|\leq O(\eps)\cdot \cost_1(P,C,\calB,\avg).
	\]

	For any center set $C\in \calX^k$, consider a mapping $\nu: C \rightarrow \R^d$ defined as follows: $\nu(c) = c_i^\star$ for every $c\in \Cfar^R\cup \Cclose^R$, and $\nu(c)=c$ for the remaining centers $c\in C\setminus \Cfar^R\setminus \Cclose^R$.
	By definition, we know that $\nu(C)$ contains at most $t_R(C) + 1$ distinct centers.
	We first have the following lemma that enables us to only focus on the concentration for center sets $C\subset B(c^\star_i, \frac{48z r}{\eps})$.

	\begin{lemma}[Approximation of $\cost_z$]
		\label{lm:approximation_cost}
		For every weighted set $Q\subseteq R$ with total weight $n_R$, we have
		\[
		\cost_z(Q,C,\calB, h) \in (1\pm \frac{\eps}{4}) \cdot (\cost_z(Q,\nu(C),\calB, h) + \phi(C,h)),
		\]
		where 
		\[
		\phi(C,h) := \sum_{c\in \Cfar^R} h_c\cdot d^z(c,c^\star_i),
		\]
		which is independent of the choice of $Q$.\footnote{$\phi(C,\avg)$ plays the same role as $\Delta(C)$ defined in~\cite[Lemma 4.4]{braverman2022power}, that captures the clustering cost of points to remote centers in $\Cfar$.}
	\end{lemma}
	
	\begin{proof}
		The lemma is implied by the proof of Lemma 4.4 of~\cite{braverman2022power}.
		For completeness, we provide proof here.

		By the construction of $\nu(C)$, we can check that for every $p\in Q$ and $c\in C$,
		\begin{enumerate}
			\item if $c\in \Cfar^R$, $d^z(p,c) \in (1\pm \frac{\eps}{4})\cdot (d^z(p,\nu(c)) + d^z(c,\nu(c)))$ by the definition of $\Cfar^R$;
			\item if $c\in \Cclose^R$, $d^z(p,c) \in (1\pm \frac{\eps}{4})\cdot d^z(p,\nu(c))$ by the definition of $\Cclose^R$;
			\item if $c\in C\setminus\Cfar^R\setminus\Cclose^R$; $d^z(p,c) = d^z(p,\nu(c))$.
		\end{enumerate}
		Let $\sigma\sim (\calB,h)$ be an arbitrary assignment function.
		Combining the above properties, we have
		\[
		\cost^{\sigma}(Q,C) \in (1\pm \frac{\eps}{4}) \cdot (\cost^{\sigma}(Q,\nu(C)) + \phi(C,h)),
		\]
		which implies the lemma.
	\end{proof}
	
	\noindent
	We also need to define another notion of covering for ring $R$, which aims to cover both the metric space $\calX$ and the hyper-parameter space of the feasible capacity constraints induced by $\calB$ (Definition~\ref{def:covering}).
	A similar idea appeared in~\cite[Lemma 4.4]{braverman2022power} but it concerns capacity constraints only in Euclidean spaces.
	Recall that $\conv(\calB^o) = \conv(\calB \cup \{0\})$.
	Let $\Phi$ denote the collection of all pairs $(C,h)\in \calX^k\times (n_R\cdot \conv(\calB^o))$.
	We partition $\Phi$ into $k+1$ sub-collections $\Phi_t$ for $t \in \left\{0,1,\ldots,k\right\}$ where $\Phi_t$ is the collection of $(C,h)\in \Phi$ such that $R$ is a $t$-level ring w.r.t. $C$ ($t_R(C) = t$).

	\begin{definition}[\bf{Coverings and covering numbers with assignment structure constraints}]
		\label{def:covering}
		Let $\calB \subseteq \Delta_k$ be an assignment structure constraint.
		Let $R\subset \ring(c_i^\star, r, 2r)$ be a ring of $n_R\geq 1$ points.
		Let $t\in [k]$ be an integer.
		We say a collection $\calF\subset \Phi_t$ is a $(t,\alpha)$-covering w.r.t. $(R,\calB)$ if for every $(C,\avg)\in \Phi_t$, there exists $(C',\avg')\in \calF$ such that for every weighted set $Q\subseteq R$ with $w_Q(Q) = n_R$,
		\[
		\cost_z(Q,C,\calB, \avg) \in \big(1\pm (\beta+\eps)\big)\cdot \left(\cost_z(Q,C',\calB, \avg') + \phi(C,\avg) \right) \pm \beta n_R r^z.
		\]
		Define $\cover(R,t,\beta,\calB)$ to be the minimum cardinality $|\calF|$ of any $(t,\beta)$-covering $\calF$ w.r.t. $(R,\calB)$.
	\end{definition}
	
	\noindent
	As the covering number $\cover(R,t,\beta,\calB)$ becomes larger, the family $\Phi_t$ is likely to induce more types of $\cost_z(Q,C,\calB, \avg)$'s. 
	Although the definition of $\cover(R,t,\beta,\calB)$ is based on the clustering cost, which looks quite different from Definition~\ref{def:covering_wo}, we have the following lemma that relates the two notions of covering numbers.

	\begin{lemma}[\bf{Relating two types of covering numbers}]
		\label{lm:relation}
		Let $\calB$ be an assignment structure constraint.
		For every $\beta > 0$ and $t\in [k]$, we have
		\begin{align*}
			\log \cover(R,t,\beta,\calB) \leq & O\left(t\cdot \log \cover(n_R,\beta)) + zk\cdot \log (\lip(\calB) \cdot z\eps^{-1} \beta^{-1})\right) \\
			\leq & O\left(\cc{\beta}\cdot t\log n_R + zk\cdot \log (\lip(\calB) \cdot z\eps^{-1} \beta^{-1})\right).
		\end{align*}
		Moreover, when $\calB = \Delta_k$, we have
		\[
		\log \cover(R,t,\beta,\calB) \leq O\left(\cc{\beta}\cdot t\log n_R + zt\cdot \log (zk\eps^{-1} \beta^{-1})\right).
		\]
	\end{lemma}
	
	\begin{proof}
		The proof can be found in Section~\ref{sec:proof_relation}.
	\end{proof}
	
	\noindent
	Now we are ready to prove Lemma~\ref{lm:ring}.
	We first prove the following weak version.

	\begin{lemma}[\bf{A weak version of Lemma~\ref{lm:ring}}]
		\label{lm:ring_weak}
		For each $i\in [k]$ and ring $R\in \calR_i$ of size $n_R\geq 1$, suppose for every $t\in [k+1]$,
		\begin{align}    
			\label{eq:lm:ring_weak}
			\Gamma_R \geq 2^{O(z\log z)}\cdot \lambda_R\cdot \lip(\calB)^2\cdot  \eps^{-2z} \cdot (k\eps^{-1}\log \delta^{-1} + \frac{k}{t}\log ( \cover(R,t,\eps,\calB)\cdot 2^k \delta^{-1})) \cdot \log^3 (k\eps^{-1}),
		\end{align}
		and specifically, when $\calB = \Delta_k$,
		\begin{align}    
			\label{eq2:lm:ring_weak}
			\Gamma_R \geq 2^{O(z\log z)}\cdot \lambda_R\cdot \eps^{-2z} \cdot (k\eps^{-1}\log \delta^{-1} + \frac{k}{t}\log ( \cover(R,t,\eps,\calB)\cdot 2^t \delta^{-1})) \cdot \log^3 (k\eps^{-1}).
		\end{align}
		With probability at least $1-\frac{\delta}{|\calR|}$, for every $k$-center set $C\in \calX^k$ and capacity constraint $h\in n_R\cdot \conv(\calB^o)$, we have that
		\begin{align*}
			& \quad |\cost_z(R,C,\calB, h)-\cost_z(S_R,C,\calB,h)| \\
			\leq & \quad \eps \left(\cost_z(R,C,\calB,h) + \cost_z(R,c_i^\star)\right) +\big(\frac{t_R(C)}{10z k\lambda_R \log \eps^{-1}}\big)^{\frac{1}{2}}\cdot \eps \cost_z(R,c_i^\star).
		\end{align*}
	\end{lemma}
	
	\noindent
	Combining with Lemma~\ref{lm:relation}, we know that the required sample number in Lemma~\ref{lm:ring_weak} is at most
	\[
	2^{O(z\log z)}\cdot \lip(\calB)^2\cdot (\cc{\eps} + k + \eps^{-1}) \cdot k\eps^{-2z}\cdot \log (n_R\cdot \lip(\calB)\cdot \delta^{-1}) \log^4 (k \eps^{-1})
	\]
	for general assignment structure constraint $\calB$, and is at most
	\[
	2^{O(z\log z)}\cdot (\cc{\eps} + \eps^{-1}) \cdot k\eps^{-2z}\cdot \log (n_R\cdot \lip(\calB)\cdot\delta^{-1}) \log^4 (k \eps^{-1})
	\]
	when $\calB = \Delta_k$.
	Compared to Theorem~\ref{thm:coreset}, there is an additional term $\log n_R$ in the coreset size, which can be as large as $O(\log n)$.
	We will show how to remove this term later.

	We first give the following lemma that solves the case of $(C,h)\in \Phi_0$ for Lemma~\ref{lm:ring_weak}.

	\begin{lemma}[\bf{Lemma~\ref{lm:ring_weak} for $\Phi_0$}]
		\label{lm:ring_center}
		With probability at least $1-\frac{\delta}{2|\calR|}$, for all $(C,h)\in \Phi_0$, the following inequality holds:
		\[
		|\cost_z(R,C,\calB,h) - \cost_z(S_R,C,\calB,h)|\leq \eps \left(\cost_z(R,C,\calB,h) + \cost_z(R,c_i^\star)\right).
		\]
	\end{lemma}
	
	\begin{proof}
		We claim that
		\begin{align}
			\label{eq1:lm:ring_center}
			|\cost_z(R,\nu(C),\calB,h) - \cost_z(S_R,\nu(C),\calB,h)|\leq \frac{\eps}{4} \left(\cost_z(R,\nu(C),\calB,h) + \cost_z(R,c_i^\star)\right),
		\end{align}
		which implies that
		\begin{align*}
			\cost_z(S_R,C,\calB,h) \in & \quad (1\pm \frac{\eps}{4}) \cdot (\cost_z(S_R,\nu(C),\calB, h) + \phi(C,h)) & (\text{Lemma~\ref{lm:approximation_cost}}) \\
			\in & \quad (1\pm \frac{\eps}{2})\cdot (\cost_z(R,\nu(C),\calB,h) + \phi(C,h) ) \pm \frac{\eps}{2}\cdot \cost_z(R,c_i^\star) & (\text{Ineq.~\eqref{eq1:lm:ring_center}}) \\
			\in & \quad (1\pm \eps)\cdot \cost_z(R,C,\calB, h) \pm \eps\cdot \cost_z(R,c_i^\star). & (\text{Lemma~\ref{lm:approximation_cost}}) 
		\end{align*}
		Hence, it suffices to prove Inequality~\eqref{eq1:lm:ring_center}.
		Since $(C,h)\in \Phi_0$ implies that $t_R(C) = 0$, we have that $\nu(C) = {c_i^\star}$, which implies that
		$\cost_z(Q,\nu(C), \calB, h) = \cost_z^{(n_R - \|h\|_1)}(Q, c_i^\star)$.
		Then it is equivalent to prove that for all $0\leq m_R\leq n_R$,
		\[
		|\cost_z^{(m_R)}(R, c_i^\star) - \cost_z^{(m_R)}(S_R, c_i^\star)| \leq \frac{\eps}{4} \left(\cost_z^{(m_R)}(R, c_i^\star) + \cost_z(R, c_i^\star)\right).
		\]
		Since $\cost_z^{(m_R)}(R,c_i^\star)\geq 0$, it suffices to prove that for all $0\leq m_R\leq n_R$,
		\[
		|\cost_z^{(m_R)}(R, c_i^\star) - \cost_z^{(m_R)}(S_R, c_i^\star)| \leq \frac{\eps}{4} \cost_z(R, c_i^\star).
		\]
		Note that $\cost_z(R,c_i^\star)\geq n_R r^z$.
		Also note that for every $p\in R$, we have $d^z(p, c_i^\star)\leq 2^z r^z$, which implies that for any $0\leq m\leq m'\leq n_R$ with $m' - m\leq \frac{\eps n_R}{2^{z+4}}$,
		\[
		\cost_z^{(m)}(R, c_i^\star) \leq \cost_z^{(m')}(R,c_i^\star) + \frac{\eps}{16} \cost_z(R,c_i^\star),
		\]
		and
		\[
		\cost_z^{(m)}(S_R, c_i^\star) \leq \cost_z^{(m')}(S_R,c_i^\star) + \frac{\eps}{16} \cost_z(R,c_i^\star).
		\]
		Then it suffices to prove that for $m_R = 0, \frac{\eps n_R}{2^{z+4}}, \frac{2\eps n_R}{2^{z+4}}, \ldots, n_R$, the following inequality holds:
		\begin{align}
			\label{eq2:lm:ring_center}
			|\cost_z^{(m_R)}(R, c_i^\star) - \cost_z^{(m_R)}(S_R, c_i^\star)| \leq \frac{\eps}{8} \cdot n_R r^z,
		\end{align}
		i.e., we only need to consider $O(2^z \eps^{-1})$ different values of $m_R$.

		By Theorem~\ref{thm:meta} and the definition of $\lambda_R$, we know that $\lambda_R \geq (\frac{\eps}{6z})^z\cdot \frac{1}{k\cdot \log (48z\eps^{-1})}$.
		Consequently, we have
		\begin{align}
			\label{eq3:lm:ring_center}
			\begin{aligned}
				\Gamma_R \geq & \quad 2^{O(z\log z)}\cdot \lambda_R\cdot  \eps^{-2z} \cdot (k\eps^{-1} \log \delta^{-1} + \log ( \cover(R,k,\eps,\calB)\cdot 2^k \delta^{-1})) \cdot \log^3 (k\eps^{-1}) & (\text{Ineq.~\eqref{eq:lm:ring_weak}}) \\
				\geq & \quad 2^{O(z \log z)} \cdot \eps^{-z-1} \cdot \log (|\calR|\cdot \delta^{-1} \eps^{-1}) & \\ 
				\geq & \quad 2^{O(z \log z)} \cdot \eps^{-z-1} \cdot \log (|\calR|\cdot \delta^{-1} \eps^{-1}). & (z\geq 1)
			\end{aligned}
		\end{align}
		Suppose $S_R = S\cup \left\{q\right\}$ and $S'_R = S\cup \left\{q'\right\}$ where $S\subseteq R$, $|S| = \Gamma_R - 1$ and $q\neq q'\in R$.
		We know that $|\cost_z^{(m)}(S_R, c_i^\star) - \cost_z^{(m)}(S'_R, c_i^\star)|\leq \frac{2n_R}{\Gamma_R}\cdot 2^z r^z$.
		By construction, $S_R$ consists of $\Gamma_R$ i.i.d. uniform samples.
		Hence, we can apply Theorem~\ref{thm:mcdiarmid} (McDiarmid's Inequality) and obtain that for every $t>0$,
		\[
		\Pr\left[\left|\cost_z^{(m_R)}(S_R, c_i^\star) - \Exp_{S_R}\left[\cost_z^{(m_R)}(S_R, c_i^\star)\right] \right|\geq t\right]\leq e^{-\frac{2t^2}{\Gamma_R\cdot (\frac{2n_R}{\Gamma_R}\cdot 2^z r^z)^2}}.
		\]
		Since $\Exp_{S_R}\left[\cost_z^{(m_R)}(S_R, c_i^\star)\right] = \cost_z^{(m_R)}(R, c_i^\star)$, we conclude that Inequality~\eqref{eq2:lm:ring_center} holds with probability at least $1-\frac{\delta \eps}{2^{z+10} |\calR|}$, which can be verified by letting $t = \frac{\eps}{8}\cdot n_R r^z$ and the bound of $\Gamma_R$ in Inequality~\eqref{eq3:lm:ring_center}.
	\end{proof}
	
	\noindent
	Next, we consider the case of $\Phi_t$ for $t\in [k]$.
	The proof idea is to first show the concentration property of $S_R$ w.r.t. a fixed pair $(C,h)\in \Phi_t$, and then discretize the parameter space $\Phi_t$ and use the union bound to handle all pairs $(C, \avg)\in \Phi_t$.
	We first propose the following key lemma for a specific pair $(C, \avg)\in \Phi_t$, which is a generalization of~\cite[Lemma 4.3]{braverman2022power} based on~\cite{cohen-addad2019on,BandyapadhyayFS21} by considering general $z\geq 1$, introducing assignment structure constraints $\calB$, and carefully analyzing the induced error of rings w.r.t. different levels $t\in [k]$.

	\begin{lemma}[\bf{Error analysis for rings w.r.t. a pair $(C, \avg)\in \Phi_t$}]
		\label{lm:ring_pair}
		Fix $t\in [k]$ and $(C,h)\in \Phi_t$.
		Let $\alpha = (\frac{t}{10z k\lambda_R \log \eps^{-1}})^{\frac{1}{2}} \eps$.
		With probability at least $1-\frac{\delta}{ 2k\cover(R,t,\eps,\calB)\cdot |\calR|}$, the following holds:
		\[
		\left|\cost_z(R, C, \calB, \avg) - \cost_z(S_R, C, \calB, \avg) \right| \leq \eps \cost_z(R, C, \calB, \avg) + \alpha \cost_z(R,c_i^\star).
		\]
	\end{lemma}
	
	\begin{proof}
		By a similar argument as in Lemma~\ref{lm:ring_center}, it suffices to prove that
		\[
		|\cost_z(R,\nu(C),\calB,h) - \cost_z(S_R,\nu(C),\calB,h)|\leq \alpha \cost_z(R,c_i^\star).
		\]
		We define a function $g$ that takes a weighted set $Q\subseteq R$ with $w_Q(Q) = n_R$ as input, and outputs $g(Q) = \cost_z(Q,\nu(C), \calB, h)$.
		Thus, it suffices to prove that with probability at least $1-\frac{\delta}{2k\cover(R,t,\eps,\calB)\cdot |\calR|}$, 
		\begin{align}
			\label{eq1:proof_ring_pair}
			\left|g(R) - g(S_R) \right| \leq O(\alpha n_R r^z). 
		\end{align}
		Actually, the construction and the following analysis of $g$ is the key for extending the \kMedian results in~\cite{cohen-addad2019on} to general $z\geq 1$.
		Our idea is to prove the following two lemmas, where the first one shows the concentration property of $g(S_R)$ and the second one shows the closeness of $\Exp_{S_R}\left[g(S_R)\right]$ and $g(R)$.
		The main difficulty is handling the additional assignment structure constraint $\calB$.

		\begin{lemma}[\bf{Concentration of $g(S_R)$}]
			\label{lm:concentration}
			With probability at least $1-\frac{\delta}{2k\cover(R,t,\eps,\calB)\cdot |\calR|}$, 
			\[
			\left| g(S_R) - \Exp_{S_R}\left[g(S_R)\right] \right| \leq \alpha n_R r^z.
			\]
		\end{lemma}

		\begin{lemma}[\bf{Closeness of $\Exp_{S_R}\left[g(S_R)\right]$ and $g(R)$}]
			\label{lm:close}
			The following holds:
			\[
			g(R)\leq \Exp_{S_R}\left[g(S_R)\right] \leq g(R) + \alpha n_R r^z.
			\]
		\end{lemma}
		
		\noindent
		Inequality~\eqref{eq1:proof_ring_pair} is a direct corollary of the above lemmas since with probability at least $1-\frac{\delta}{2k\cover(R,t,\eps,\calB)\cdot |\calR|}$, we have that
		\begin{align*}
			|g(R) - g(S_R)| & \leq  |g(R) - \Exp_{S_R}\left[g(S_R)\right]| + |\Exp_{S_R}\left[g(S_R)\right] - g(S_R)| \leq   2\alpha n_R r^z. 
		\end{align*}
		Hence, it remains to prove these two lemmas.
		Note that $\alpha = (\frac{t}{10z k\lambda_R \log \eps^{-1}})^{\frac{1}{2}} \eps \leq 2^{O(z \log z)} k^{1/2} \eps^{1-z/2}$ since $\lambda_R \geq (\frac{\eps}{6z})^z\cdot \frac{1}{k\cdot \log (48z\eps^{-1})}$.
		Then by Inequality~\eqref{eq:lm:ring_weak}, we note that
		\begin{align}
			\label{eq1:lm:ring_pair}
			\begin{aligned}
				\Gamma_R \geq 2^{O(z\log z)}\cdot \lip(\calB)^2\cdot  \eps^{-2z+2} \alpha^{-2} \cdot \log ( \cover(R,t,\eps,\calB)\cdot 2^k \delta^{-1}) \cdot \log (k \alpha^{-1}\eps^{-1}),
			\end{aligned}
		\end{align}
		and when $\calB = \Delta_k$, 
		\begin{align}
			\label{eq2:lm:ring_pair}
			\begin{aligned}
				\Gamma_R \geq 2^{O(z\log z)}\cdot  \eps^{-2z+2} \alpha^{-2} \cdot \log ( \cover(R,t,\eps,\calB)\cdot 2^t \delta^{-1}) \cdot \log (k \alpha^{-1}\eps^{-1}).
			\end{aligned}
		\end{align}
		For ease of analysis, we slightly abuse the notation by using $C$ to replace $\nu(C)$ in the following.
		Then we know that $C\subset B(c_i^\star, \frac{48z r}{\eps})$ and there are $k-t$ centers $c\in C$ located at $c_i^\star$.

		\begin{proof}[Proof of Lemma~\ref{lm:concentration}]
			%
			We first show the following Lipschitz property of $g$. 
			
			\begin{claim}[\bf{Lipschitz property of $g$}]
				\label{claim:sample_Lipschitz}
				For every two realizations $S_R$ and $S'_R$ of size $\Gamma_R$ that differ by one sample, we have
				\[
				\left|g(S_R) - g(S'_R) \right|\leq \frac{n_R}{\Gamma_R}\cdot (62z)^z \eps^{-z+1} r^z.
				\]
			\end{claim}
			
			\begin{proof}
				Suppose $S_R = S\cup \left\{q\right\}$ and $S'_R = S\cup \left\{q'\right\}$ where $S\subseteq R$, $|S| = \Gamma_R - 1$ and $q\neq q'\in R$.
				Let $\sigma: S_R\times C\rightarrow \R_{\geq 0}$ with $\sigma\sim (\calB, \avg)$ be an optimal assignment function such that $g^\sigma(S_R) = g(S_R)$.
				We define $\sigma': S'_R\times C\rightarrow \R_{\geq 0}$ with $\sigma'\sim (\calB, \avg)$ to be: 
				\[
				\sigma'(p, \cdot ) = \sigma(p,\cdot ), \forall p\in S, \text{ and } \sigma'(q',\cdot ) = \sigma(q,\cdot).
				\]
				We have
				\begin{align*}
					&  g(S'_R) & \\ 
					\leq &  g^{\sigma'}(S'_R) & (\text{by optimality}) \\
					= & g^{\sigma}(S_R) + \sum_{c\in C: c\neq c_i^\star} \sigma(q,c)\cdot (d^z(q',c) - d^z(q,c))  & (\text{Defns. of $\sigma'$ and $g$}) \\
					\leq & g^{\sigma}(S_R) + \sum_{c\in C: c\neq c_i^\star} \sigma(q,c)\cdot \left(\eps\cdot d^z(q,c) + (\frac{3z}{\eps})^{z-1} d^z(q,q')\right) & (\text{Lemma~\ref{lm:triangle}}) \\
					\leq & g^{\sigma}(S_R) + \sum_{c\in C: c\neq c_i^\star} \sigma(q,c)\cdot \left(\eps\cdot (d(c^\star_i,c)+ d(q,c^\star_i))^z + (\frac{3z}{\eps})^{z-1} (4r)^z\right) & (\text{triangle ineq.}) \\
					\leq & g^{\sigma}(S_R) + \sum_{c\in C: c\neq c_i^\star} \sigma(q,c)\cdot \left(\eps\cdot (\frac{48zr}{\eps}+ 2r)^z + (\frac{3z}{\eps})^{z-1} (4r)^z\right) & (\text{Defn. of $\barCfar$}) \\
					\leq & g^{\sigma}(S_R) + \sum_{c\in C: c\neq c_i^\star} \sigma(q,c)\cdot (62z)^z \eps^{-z+1} r^z & \\
					\leq & g^{\sigma}(S_R) + \frac{n_R}{\Gamma_R}\cdot (62z)^z \eps^{-z+1} r^z. & (w_{S_R}(q) = \frac{n_R}{\Gamma_R})
				\end{align*}
				By symmetry, we complete the proof.
			\end{proof}

			\noindent
			By construction, $S_R$ consists of $\Gamma_R$ i.i.d. uniform samples.
			Hence, we can apply Theorem~\ref{thm:mcdiarmid} (McDiarmid's Inequality) and obtain that for every $t>0$,
			\[
			\Pr\left[\left|g(S_R) - \Exp_{S_R}\left[g(S_R)\right] \right|\geq t\right]\leq e^{-\frac{2t^2}{\Gamma_R\cdot (\frac{n_R}{\Gamma_R}\cdot (62z)^z \eps^{-z+1} r^z)^2}}.
			\]
			Lemma~\ref{lm:concentration} can be verified by letting $t = \alpha n_R r^z$ and Inequalities~\eqref{eq1:lm:ring_pair} and~\eqref{eq2:lm:ring_pair}.
		\end{proof}
		
		\begin{proof}[Proof of Lemma~\ref{lm:close}]
			%
			We first show the easy direction, say $g(R) \leq \Exp_{S_R}\left[g(S_R)\right]$, which is guaranteed by the following convexity of $g$.
			For each possible realization of $S_R$, let $\mu_{S_R}$ denote the realized probability of $S_R$ and $\sigma_{S_R}: S_R\times C\rightarrow \R_{\geq 0}$ with $\sigma_{S_R}\sim (\calB, \avg)$ be an optimal assignment function such that $g^{\sigma_{S_R}}(S_R) = g(S_R)$.
			For the convenience of the argument, we extend the domain of $\sigma_{S_R}$ to be $\sigma_{S_R}: R\times C\rightarrow \R_{\geq 0}$ where $\sigma_{S_R}(p,\cdot) = 0$ if $p\in R\setminus X$.
			By definition, we know that 
			\begin{align}
				\label{eq:proof_lm_close}
				\Exp_{S_R}\left[g(S_R)\right] = \Exp_{S_R}\left[g^{\sigma_{S_R}}(S_R)\right] = \sum_{S_R} \mu_{S_R}\cdot g(S_R).
			\end{align}
			Now consider an assignment function $\sigma: R\times C\rightarrow \R_{\geq 0}$ obtained by adding up $\sigma_{S_R}$, i.e., $\sigma = \Exp_{S_R}\left[\sigma_{S_R}\right] = \sum_{S_R} \mu_{S_R}\cdot \sigma_{S_R}$.
			We first show that $\sigma \sim (\calB, \avg)$ is a feasible assignment function on $R$.
			On one hand, since $\sigma_{S_R}\sim \calB$, we know that for every $p\in S_R$, $\sigma_{S_R}(p,\cdot )\in w_{S_R}(p)\cdot \calB$.
			Thus, we have for every $p\in R$, 
			\[
			\sigma(p,\cdot) = \sum_{S_R} \mu_{S_R}\cdot \sigma_{S_R}(p,\cdot) \in \sum_{S_R} \mu_{S_R} \cdot w_{S_R}(p) \cdot \calB = \calB,
			\]
			since $\Exp_{S_R}\left[w_{S_R}(p)\right] = 1$.
			On the other hand, since $\sigma \sim \avg$, it is obvious that $\sigma \sim \avg$ since $\sum_{S_R} \mu_{S_R} = 1$.
			Thus, we obtain the following inequality
			\begin{align*}
				g(R) &\leq  g^{\sigma}(R) & (\text{by optimality})\\
				& =  \sum_{S_R} \mu_{S_R} \cdot g(S_R) & (\text{by linearity}) \\
				& = \Exp_{S_R}\left[g(S_R)\right]. & (\text{Eq.~\eqref{eq:proof_lm_close}})
			\end{align*}
			Next, we show the difficult direction that $\Exp_{S_R}\left[g(S_R)\right] \leq (1+ \eps)\cdot g(R) + \alpha n_R r^z$.
			We first have the following claim showing that the maximum difference between $g(S_R)$ and $g(R)$ is bounded.
			\begin{claim}[\bf{Uniform upper bound of $g(S_R)$}]          
				\label{claim:close}
				The following holds
				\[
				g(S_R) \leq g(R) + (62z)^z \eps^{-z+1} n_R r^z.
				\]
			\end{claim}
			
			\begin{proof}
				The proof is almost identical to that of Claim~\ref{claim:sample_Lipschitz}.
				The only difference is that $R$ and $S_R$ may differ by $n$ points instead of one point with weight $\frac{n_R}{\Gamma_R}$, which results in a total difference on weights
				\[
				\sum_{p\in R} | 1 - w_{S_R}(p) | \leq 2n_R,
				\]
				where we let $w_{S_R}(p) = 0$ for $p\notin S_R$.
			\end{proof}

			\noindent
			Let $\sigma^\star: R\times C\rightarrow \R_{\geq 0}$ with $\sigma^\star \sim (\calB, \avg)$ be an optimal assignment function such that $g^{\sigma^\star}(R) = g(R)$.
			We first construct $\pi: S_R\times C\rightarrow \R_{\geq 0}$ as follows: for every $p\in S_R$,
			\[
			\pi(p,\cdot) = \frac{n_R}{\Gamma_R}\cdot \sigma^{\star}(p,\cdot).
			\]
			Note that $\pi\sim \calB$ always holds, but $\pi\sim \avg$ may not hold.
			We then construct another assignment function $\pi': S_R\times C\rightarrow \R_{\geq 0}$ with $\|\pi'(p,\cdot)\|_1 = \|\pi(p,\cdot)\|_1$ for every $p\in S_R$ and $\pi'\sim (\calB, \avg)$ such that the total mass movement $\sum_{p\in S_R} \|\pi(p,\cdot ) - \pi'(p,\cdot)\|_1$ from $\pi$ to $\pi'$ is minimized.
			We have the following claim.
			\begin{claim}[\bf{Upper bound of $g(S_R)$ w.h.p.}]
				\label{claim:sample_upper_bound}
				With probability at least $1-\frac{\eps^{z-1}\alpha}{2\cdot (62z)^z}$, we have
				\[
				g(S_R) \leq g^{\pi'}(S_R) \leq  g(R) + 0.5\alpha n_R r^z.
				\]
			\end{claim}
			\begin{proof}
				$g(S_R) \leq g^{\pi'}(S_R)$ holds by optimality, and it remains to prove $g^{\pi'}(S_R) \leq  g^{\sigma^\star}(R) + 0.5\alpha n_R r^z$.
				Since $g^{\sigma^\star}(R) = g(R)$ by definition,
				note that
				\begin{align*}
					& g^{\pi'}(S_R) - g^{\sigma^\star}(R) & \\
					= &  \sum_{c\in C: c = c_i^\star} \left(\sum_{p\in S_R} \pi'(p,c)\cdot d^z(c, c^\star_i) - \sum_{p\in R} \sigma^\star(p,c)\cdot d^z(c, c^\star_i) \right) & \\
					&  + \sum_{c\in C: c\neq c_i^\star} \left(\sum_{p\in S_R} \pi'(p,c)\cdot d^z(p, c) - \sum_{p\in R} \sigma^\star(p,c)\cdot d^z(p,c)\right) & (\text{Defn. of $g$}) \\
					= &  \sum_{c\in C: c\neq c_i^\star} \left(\sum_{p\in S_R} \pi'(p,c)\cdot d^z(p, c) - \sum_{p\in R} \sigma^\star(p,c)\cdot d^z(p,c)\right). & (\pi', \sigma^\star \sim \avg)
				\end{align*}
				Hence, it is equivalent to proving that with probability at least $1-\frac{\eps^{z-1}\alpha}{2\cdot (62z)^z}$,
				\begin{align}
					\label{eq1:proof_claim_sample}
					\sum_{c\in C: c\neq c_i^\star} \left(\sum_{p\in S_R} \pi'(p,c)\cdot d^z(p, c) - \sum_{p\in R} \sigma^\star(p,c)\cdot d^z(p,c)\right) \leq 0.5\alpha n_R r^z.
				\end{align}
				The idea is to use $\pi$ as an intermediary and we define another helper function $\phi$ that takes $S_R$ as input, and outputs
				\[
				\phi(S_R) := \sum_{c\in C: c\neq c_i^\star} \sum_{p\in S_R} \pi(p,c)\cdot d^z(p,c) + \sum_{c\in C: c\neq c_i^\star} (\avg_c - \sum_{p\in S_R} \pi(p,c))\cdot d^z(c,c^\star_i).
				\]
				To prove Inequality~\eqref{eq1:proof_claim_sample}, it suffices to prove that each of the following two inequalities holds with probability at least $1-\frac{\eps^{z-1} \alpha}{4\cdot (62z)^z}$.
				\begin{align}
					\label{eq2:proof_claim_sample}
					\phi(S_R) - \sum_{c\in C: c\neq c_i^\star} \sum_{p\in R} \sigma^\star(p,c)\cdot d^z(p,c) \leq 0.25\alpha n_R r^z,
				\end{align}
				and
				\begin{align}
					\label{eq3:proof_claim_sample}
					\sum_{c\in C: c\neq c_i^\star} \sum_{p\in S_R} \pi'(p,c) d^z(p,c) - \phi(S_R) \leq 0.25\alpha n_R r^z.
				\end{align}

				\paragraph{Proof of Inequality~\eqref{eq2:proof_claim_sample}}
				By linearity and the definition of $\pi$, we know that $\Exp_{S_R}\left[\phi(S_R)\right] = \sum_{c\in C: c\neq c_i^\star} \sum_{p\in R} \sigma^\star(p,c)\cdot d^z(p,c)$.
				Similar to $g$, we show the following Lipschitz of $\phi$.
				
				\begin{claim}[\bf{Lipschitz property of $\phi$}]
					\label{claim:phi_Lipschitz}
					For every two realizations $S_R$ and $S'_R$ of size $\Gamma_R$ that differ by one sample, we have
					\[
					\left|\phi(S_R) - \phi(S'_R) \right|\leq \frac{2n_R}{\Gamma_R}\cdot (62z)^z \eps^{-z+1} r^z.
					\]
				\end{claim}
				
				\begin{proof}
					Suppose the only different points are $p\in S_R$ and $q\in S'_R$.
					By definition, we know that
					\begin{align*}
						&  \left|\phi(S_R) - \phi(S'_R) \right| & \\
						= &  \frac{n_R}{\Gamma_R}\cdot \left|\sum_{c\in C: c\neq c_i^\star} \sigma^\star(p,c)\cdot (d^z(p,c) - d^z(c,c^\star_i)) - \sigma^\star(q,c)\cdot (d^z(q,c) - d^z(c,c^\star_i)) \right| & (\text{Defns. of $\phi$ and $\pi'$})\\
						\leq &  \frac{2 n_R}{\Gamma_R} \cdot (62z)^z \eps^{-z+1} r^z, & (\text{Proof of \Cref{claim:sample_Lipschitz}})
					\end{align*}
					which completes the proof.
				\end{proof}

				\noindent
				Now similar to Lemma~\ref{lm:concentration}, we can apply Theorem~\ref{thm:mcdiarmid} (McDiarmid's Inequality) and obtain that for every $t>0$,
				\[
				\Pr\left[\left|\phi(S_R) - \Exp_{S_R}\left[\phi(S_R)\right] \right|\geq t\right]= \Pr\left[|\phi(S_R)|\geq t\right] \leq e^{-\frac{2t^2}{\Gamma_R\cdot (\frac{2n_R}{\Gamma_R}\cdot (62z)^z \eps^{-z+1} r^z)^2}}.
				\]
				Inequality~\eqref{eq2:proof_claim_sample} can be verified by letting $t = 0.25\alpha n_R r^z$ and Inequalities~\eqref{eq1:lm:ring_pair} and~\eqref{eq2:lm:ring_pair}.
				
				\paragraph{Proof of Inequality~\eqref{eq3:proof_claim_sample}}
				We first consider the general assignment structure constraint $\calB$.
				For analysis, we construct a function $\psi$ that takes $S_R$ as input, and outputs
				\[
				\psi(S_R) := \sum_{c\in C} \left|\sum_{p\in S_R} \pi(p,c) - \avg_c \right|,         
				\]
				i.e., the total capacity difference between $\pi$ and $\sigma^\star$.
				We have the following claim showing that $\psi(S_R)$ is likely to be small.

				\begin{claim}[\bf{$\psi(S_R)$ is tiny w.h.p.}]
					\label{claim:psi_tiny}
					With probability at least $1-\frac{\eps^{z-1} \alpha}{4\cdot (62z)^z}$, we have
					\[
					\psi(S_R) \leq \frac{\eps^{z-1}\alpha n_R}{4\cdot \lip(\calB)\cdot (62z)^z}.
					\]
				\end{claim}
				
				\begin{proof}
					For every string $s\in \left\{+1, -1\right\}^k$, we define a function $\psi_s$ that takes $S_R$ as input, and outputs
					\[
					\psi_s(S_R) := \sum_{c\in C} s_c\cdot \left(\sum_{p\in S_R} \pi(p,c) - \avg_c \right).
					\]
					Note that $\psi(S_R) = \max_{s\in \left\{+1, -1\right\}^k} \psi_s(S_R)$.
					It is equivalent to proving that with probability at least $1-\frac{\eps^{z-1} \alpha}{4\cdot (62z)^z}$, for all $s\in \left\{+1, -1\right\}^k$,
					\begin{align}
						\label{eq:proof_claim:psi_tiny}
						\psi_s(S_R) \leq \frac{\eps^{z-1} \alpha n_R}{4\cdot \lip(\calB)\cdot (62z)^z}.
					\end{align}
					By the union bound, it suffices to prove that for any $s\in \left\{+1, -1\right\}^k$, Inequality~\eqref{eq:proof_claim:psi_tiny} holds with probability at least $1-\frac{\eps^{z-1}\alpha}{2^{k+2}\cdot (62z)^z}$.

					Fix a string $s\in \left\{+1, -1\right\}^k$ and note that $\Exp_{S_R}\left[\psi_s(S_R)\right] = 0$.
					Also noting that for any two realizations $S_R$ and $S'_R$ of size $\Gamma_R$ that differ by one sample, we have
					\[
					|\psi_s(S_R) - \psi_s(S'_R)| \leq \frac{2n_R}{\Gamma_R},
					\]
					i.e., $\psi_s(S_R)$ is $\frac{2n_R}{\Gamma_R}$-Lipschitz.
					We again apply  Theorem~\ref{thm:mcdiarmid} (McDiarmid's Inequality) and obtain that for every $t>0$,
					\[
					\Pr\bigl[\left|\psi_s(S_R) - \Exp_{S_R}\left[\psi_s(S_R)\right] \right|\geq t \bigr]= \Pr\bigl[|\psi_s(S_R)|\geq t\bigr] \leq e^{-\frac{2t^2}{\Gamma_R\cdot (2n_R/\Gamma_R)^2}}.
					\]
					Inequality~\eqref{eq:proof_claim:psi_tiny} can be verified by letting $t = \frac{\eps^{z-1} \alpha n_R}{4\cdot \lip(\calB)\cdot (62z)^z}$ and Inequality~\eqref{eq1:lm:ring_pair}, which ensures the success probability to be at least $1-\frac{\eps^{z-1}\alpha}{2^{k+2}\cdot (62z)^z}$.
				\end{proof}

				\noindent
				Let $\avg'\in n_R\cdot \calB$ be the capacity constraint with $\avg'_c = \sum_{p\in S_R} \pi(p,c)$.
				We have that $\pi \sim (\calB, \avg')$ and $\psi(S_R) = \|\avg' - \avg\|_1$.
				Let $\tau = \sum_{p\in S_R} \|\pi(p,\cdot) - \pi'(p,\cdot)\|_1$.
				Recall that $\pi(p,\cdot), \pi'(p,\cdot)\in \frac{n_R}{\Gamma_R}\cdot \calB$ for every $p\in S_R$.
				We have that with probability at least $1-\frac{\eps^{z-1}\alpha}{4\cdot (62z)^z}$,
				\begin{align}
					\label{eq:proof_eq3}
					\begin{aligned}
						\tau & \leq \,\,\, \lip(\frac{n_R}{\Gamma_R}\cdot \calB) \cdot \|\avg' - \avg\|_1 & \qquad(\text{Defns. of $\lip(\calB)$ and $\pi'$}) \\
						& = \,\,\, \lip(\calB)\cdot \psi(S_R) & \qquad (\text{scale-invariant of $\lip$})\\
						& \leq \,\,\, \frac{\eps^{z-1} \alpha n_R}{4\cdot (62z)^z}, &
						\qquad (\text{Claim~\ref{claim:psi_tiny}})
					\end{aligned}
				\end{align}
				which implies that
				\begin{align*}
					&  \sum_{c\in C: c\neq c_i^\star} \sum_{p\in S_R} \pi'(p,c)\cdot d^z(p,c) - \phi(S_R) & \\
					= & \,\,\, \sum_{c\in C: c\neq c_i^\star} \sum_{p\in S_R} (\pi'(p,c) - \pi(p,c))\cdot d^z(p,c) - \sum_{c\in C: c\neq c_i^\star} (\avg_c - \sum_{c\in C: c\neq c_i^\star} \pi(p,c))\cdot d^z(c,c^\star_i) & (\text{Defn. of $\phi$})\\
					= & \,\,\, \sum_{c\in C: c\neq c_i^\star} \sum_{p\in S_R} (\pi'(p,c) - \pi(p,c))\cdot (d^z(p,c)- d^z(c,c^\star_i)) & (\pi'\sim \avg)\\
					\leq & \,\,\, \Bigl| \sum_{c\in C: c\neq c_i^\star} \sum_{p\in S_R} (\pi'(p,c)-\pi(p,c))\Bigr|\cdot \max_{p\in R, c\in C: c\neq c_i^\star} (d^z(p,c) - d^z(c,c^\star_i)) & \\ 
					\leq & \,\,\, \tau\cdot \max_{p\in R, c\in C: c\neq c_i^\star} (d^z(p,c) - d^z(c,c^\star_i)) & (\text{Defn. of $\tau$}) \\
					\leq & \,\,\, \tau \cdot (62z)^z \eps^{-z+1} r^z & (\text{Defn. of $\barCfar$}) \\
					\leq & \,\,\, 0.25\alpha n_R r^z, & (\text{Ineq.~\eqref{eq:proof_eq3}})
				\end{align*}
				i.e., Inequality~\eqref{eq3:proof_claim_sample} holds for general assignment structure constraint $\calB$.

				When $\calB = \Delta_k$, we can safely regard $(C,h)$ as $(\widehat{C}, \widehat{h})$ where 
				\begin{itemize}
					\item when $t\in [k-1]$, $\widehat{C} = \left\{c\in C: c\neq c_i^\star\right\}\cup {c_i^\star}$, and when $t= k$, $\widehat{C} = C$. 
					\item $\widehat{h}_c = h_c$ for $c\neq c_i^\star$, and $\widehat{h}_{c_i^\star} = \sum_{c\in C: c = c_i^\star} h_c$.
				\end{itemize}
				The only difference is that we only need to consider at most $2^{t+1}$ different functions $\psi_s$ in the proof of Claim~\ref{claim:psi_tiny}.
				This difference enables us to set $\Gamma_R$ as in Inequality~\eqref{eq2:lm:ring_pair}, which contains a factor of $2^t$ instead of $2^k$ in Inequality~\eqref{eq1:lm:ring_pair}.
				The remaining proofs are the same, and Inequality~\eqref{eq3:proof_claim_sample} holds when $\calB = \Delta_k$.
				
				Thus, we complete the proof of Claim~\ref{claim:sample_upper_bound}.
			\end{proof}

			\noindent
			Now we are ready to prove $\Exp_{S_R}\left[g(S_R)\right] \leq g(R) + \alpha n_R r^z$.
			By Claims~\ref{claim:close} and~\ref{claim:sample_upper_bound},
			\begin{align*}
				&  \Exp_{S_R}\left[g(S_R)\right] \\
				\leq &  \frac{\eps^{z-1}\alpha}{2\cdot (62z)^z}\cdot (g(R) + (62z)^z \eps^{-z+1} n_R r^z) + (1-\frac{\eps^{z-1}\alpha}{2\cdot (62z)^z})(g(R) + 0.5\alpha n_R r^z)  \\
				\leq &  g(R) + \alpha n_R r^z. 
			\end{align*}
			Thus, we complete the proof of Lemma~\ref{lm:close}.
		\end{proof}

		\noindent
		Overall, we complete the proof of Lemma~\ref{lm:ring_pair}.
	\end{proof}

	\noindent
	Combining with Definition~\ref{def:covering}, we are ready to prove Lemma~\ref{lm:ring_weak}.

	\begin{proof}[Proof of Lemma~\ref{lm:ring_weak}]
		Lemma~\ref{lm:ring_center} shows the correctness of Lemma~\ref{lm:ring_weak} for all $(C,h)\in \Phi_0$.
		Hence, it suffices to prove that for every $t\in [k]$, with probability at least $1-\frac{\delta}{2k\cdot |\calR|}$, the following holds for all $(C,h)\in \Phi_t$:
		\begin{align}
			\label{eq0:proof_lm:ring}
			\begin{aligned}
				& \quad |\cost_z(R,C,\calB, h)-\cost_z(S_R,C,\calB,h)| \\
				\leq & \quad \eps \left(\cost_z(R,C,\calB,h) + \cost_z(R,c_i^\star)\right) +\big(\frac{t}{10z k\lambda_R \log \eps^{-1}}\big)^{\frac{1}{2}}\cdot \eps \cost_z(R,c_i^\star).
			\end{aligned}
		\end{align}
		Lemma~\ref{lm:ring_weak} is a direct corollary by the union bound.
		
		Fix $t\in [k]$ and let $\calF_t\subset \Phi_t$ be a $(t,\eps)$-covering of $(R,\calB)$ of size at most $\cover(R,t, \eps,\calB)$.
		Let $\alpha = (\frac{t}{10z k\lambda_R \log \eps^{-1}})^{\frac{1}{2}} \eps$.
		By Lemma~\ref{lm:ring_pair}, with probability at least $1-\frac{\delta}{2k\cdot |\calR|}$, for every $(C,\avg)\in \calF$, the following holds:
		\begin{align}
			\label{eq1:proof_lm:ring}
			\left|\cost_z(R, C, \calB, \avg) - \cost_z(S_R, C, \calB, \avg) \right| \leq \eps \cost_z(R, C, \calB, \avg) + \alpha \cost_z(R, c^\star_i). 
		\end{align}
		For every $(C,\avg)\in \Phi_t$, there must exist $(C',\avg')\in \calF_t$ such that
		\[
		\cost_z(R,C,\calB, h) \in (1\pm 2\eps)\cdot (\cost_z(R,C',\calB, \avg') + \phi(C,h)) \pm \eps n_R r^z,
		\]
		and
		\[
		\cost_z(S_R,C,\calB, h) \in (1\pm 2\eps)\cdot (\cost_z(S_R,C',\calB, \avg') + \phi(C,h)) \pm \eps n_R r^z.
		\]
		Combining the above two inequalities with Inequality~\eqref{eq1:proof_lm:ring}, we conclude that
		\[
		\left|\cost_z(R, C, \calB, \avg) - \cost_z(S_R, C, \calB, \avg) \right| \leq O(\eps) \left(\cost_z(R, C, \calB, \avg) + \cost_z(R, c^\star_i) \right) + O(\alpha)\cdot \cost_z(R,c_i^\star),
		\]
		which completes the proof of Inequality~\eqref{eq0:proof_lm:ring}.
	\end{proof}

	\noindent
	Finally, we show how to prove \Cref{lm:ring} by shaving off the term $ \log n_R$ in~\Cref{lm:ring_weak}.
	The main idea is to use the iterative size reduction approach of~\cite{BJKW21} that has been widely employed in recent works \cite{cohen2021new,cohen2022towards,braverman2022power} to obtain coresets of size independent of $n_R$. The technique is somewhat standard and we omit the details.
	%
	
	\begin{proof}[Proof of Lemma~\ref{lm:ring}]
		We can interpret $S_R$ as an $L=O(\log^* |R|)$-steps uniform sampling. Specifically, set  $T=2^{O(z\log z)}\cdot \lambda_R\cdot \lip(\calB)^2\cdot  \eps^{-2z} \cdot \log^3 (k\eps^{-1}\delta^{-1})$, 
		and for $i=1,\cdots,L$, let $S_R^{(i)}$ be a uniform sample of size $\tilde{O}((\Gamma_R + k) T\cdot  (\log ^{(i)}(|R|))^3)$ from $S_R^{(i-1)}$ where $\log^{(i)}$ denotes the $i$-th iterated logarithm and $S_R^{(0)}=R$. We remark that $S_R$ has the same distribution as $S_R^{(L)}$. 
		
		Identical to the proof of Theorem 3.1 in \cite{BJKW21}, we can obtain that with probability at least $1-\frac{\delta}{|\calR|\cdot|S_R^{(i-1)}|}$, for every $(C,h)$,
		\begin{eqnarray}
			\label{eqn:iterative}
			\left|\cost_z(S_R^{(i)}, C, \calB, \avg) - \cost_z(S_R^{(i-1)}, C, \calB, \avg) \right| \leq O\bigg(\frac{\eps\cdot \left(\cost_z(R, C, \calB, \avg) + \cost_z(R, c^\star_i) \right)}{(\log^{(i)} |R|)^{\frac{1}{2}}}\bigg).
		\end{eqnarray}
		Summing (\ref{eqn:iterative}) over $i=1,2,..,L$ and applying the union bound, we finish the proof.
	\end{proof}
	
	\subsection{Proof of Lemma~\ref{lm:group}: Error Analysis for Groups}
	\label{sec:proof_group}
	
	Finally, we analyze the error induced by groups. 
	Throughout this section, we fix $i\in [k]$ and a $k$-center set $C\in \calX^k$.
	Technically, we prove that the inclusion of assignment structure constraints $\calB$ does not pose additional challenges compared to \cite{braverman2022power}. 
	This is due to the subdivision of uncolored groups into at most $k$ equivalent classes according to $\Cfar^G$ (Lemma \ref{lm:equivalent}) and the projection of all centers to $c_i^\star$ for examination (\Cref{obs:outliers}). 
	It is important to note that the introduction of $\calB$ does not impact the cost function when all centers are positioned at the same location ($c_i^\star$). 
	These new geometric observations conclude Lemma~\ref{lm:group}.

	We recall the following group decomposition as in~\cite{braverman2022power}.

	\begin{definition}[\bf{Colored and uncolored groups~\cite{braverman2022power}}]
		\label{def:color}
		Fix $i\in [k]$ and a $k$-center set $C\in \calX^k$.
		The collection of groups $\calG_i$ can be decomposed into \emph{colored} groups and \emph{uncolored} groups w.r.t. $C$ such that
		\begin{enumerate}
			\item There are at most $O(k\log (z\eps^{-1}))$ colored groups; \item For every uncolored group $G\in \calG_i$, center set $C$ can be decomposed into two parts $C = \Cclose^G \cup \Cfar^G$ such that
			\begin{itemize}
				\item For any $c\in \Cclose^G$ and $p\in G$, $d(c,c^\star_i) < \frac{\eps}{9z}\cdot d(p,c^\star_i)$;
				\item For any $c\in \Cfar^G$ and $p\in G$, $d(c,c^\star_i) > \frac{24z}{\eps}\cdot d(p,c^\star_i)$.
			\end{itemize}
		\end{enumerate}
	\end{definition}
	
	\noindent
	Intuitively, for every uncolored group $G\in \calG_i$, every center $c\in C$ is either very ``close'' or very ``far'' from $G$.
	%
	
	%
	For every colored group $G\in \calG_i$, we use the following lemma to upper bound the induced error, which is a corollary of Theorem~\ref{thm:meta} and Observation~\ref{ob:Lipschitz}.
	
	\begin{lemma}[\bf{Error analysis of colored groups}]
		\label{lm:color}
		Let $G\in \calG_i$ be a colored group and $\avg\in |G|\cdot \calB$ be a feasible capacity constraint of $G$.
		The following holds:
		\[
		\left| \cost_z(G,C,\calB,\avg) - \cost_z(D_G,C,\calB,\avg) \right| \leq \eps \cdot \cost_z(G,C,\calB,\avg) + \frac{\eps}{k\log (z\eps^{-1})}\cdot \cost_z(P_i,c^\star_i).
		\]
	\end{lemma}

	\begin{proof}
		By Theorem~\ref{thm:meta}, we know that $\cost_z(G,c^\star_i)\leq (\frac{\epsilon}{6z})^z\cdot \frac{\cost_z(P_i,c^\star_i)}{k\cdot \log (48z/\epsilon)}$ since $G$ is a group.
		Then by Observation~\ref{ob:Lipschitz}, we have
		\begin{align*}
			&\quad \left| \cost_z(G,C,\calB,\avg) -  \cost_z(D_G,C,\calB,\avg) \right| & \\
			\leq &\quad  \eps \cdot \cost_z(G,C,\calB,\avg) + (\frac{6z}{\eps})^{z-1}\cdot (\cost(G,c^\star_i) + \cost(D_G, c^\star_i)) & (\text{Ob.~\ref{ob:Lipschitz}}) \\
			= &\quad  \eps \cdot \cost_z(G,C,\calB,\avg) + 2\cdot (\frac{6z}{\eps})^{z-1}\cdot \cost(G,c^\star_i) & (\text{Defn.~\ref{def:twopoints}}) \\
			\leq &\quad  \eps \cdot \cost_z(G,C,\calB,\avg) + \frac{\eps}{k\log (z\eps^{-1})}\cdot \cost_z(P_i,c^\star_i), & (\text{Thm.~\ref{thm:meta}})
		\end{align*}
		which completes the proof.
	\end{proof}
	
	\noindent
	Since there are at most $O(k\log (z\eps^{-1}))$ colored groups, the above lemma provides an upper bound of the error induced by all colored groups $G\in \calG_i$.
	It remains to upper bound the error of uncolored groups. To this end, we have the following lemma that further classifies the uncolored groups according to their $\Cfar^G$. 
	
	\begin{lemma}[\bf{Equivalent classes of uncolored groups w.r.t. $\Cfar^G$}]
		\label{lm:equivalent}
		There exists a partition $\calU_1,\dots,\calU_k$ of $\calG_i$ such that for every $j\in[k]$, it holds that 
		$\forall G,G'\in\calU_j,\Cfar^G=\Cfar^{G'}$.
	\end{lemma}
	\begin{proof}
		To prove the desired result, it is sufficient to demonstrate that there are at most $k$ distinct $\Cfar^G$ for $G \in \calG_i$. In other words, we need to show that $|\{\Cfar^G\}_{G \in \calG_i}|\le k$.
		
		For any two groups $G,G'\in\calG_i$, we can assume, without loss of generality, that
		\begin{equation*}
			\forall p\in G, p'\in G', \quad d(p,c_i^\star)\le d(p',c_i^\star).
		\end{equation*}
		
		According to the definition of $\Cfar^G$ and $\Cfar^{G'}$ (see \Cref{def:color}), this implies that $\Cfar^{G'}\subseteq \Cfar^G$. Consequently, if we let $C_1,\dots, C_l$ (with $|C_1|\le |C_2|\le\dots\le |C_l|$) represent different center sets in $\{\Cfar^G\}_{G\in \calG_i}$, it follows that $C_1\subsetneq C_2\subsetneq\dots\subsetneq C_l\subseteq C$. Therefore, we have $|C_1|<|C_2|<\dots <|C_l|\le k$, which implies $|\{\Cfar^G\}_{G\in \calG_i}|\le k$. This completes the proof.
	\end{proof}
	
	\noindent
	Our plan is to merge uncolored groups in each $\calU_j$ and analyze their error as a whole.
	Specifically, for every $j\in[k]$, we define $U_j:=\bigcup_{G\in\calU_j} G$ as the union of all groups in $\calU_j$, and define $D_{U_j}=\bigcup_{G\in \calU_j}D_G$ as the union of all two-point coresets of groups in $\calU_j$.
	The following lemma provides an upper bound on the error for each $U_j$.
	
	\begin{lemma}[\bf{Error analysis of $D_{U_j}$'s}]
		\label{lem:uncolored_groups}
		Fix $j\in[k]$ and let $h\in |U_j|\cdot \conv(\calB^o)$ be a feasible capacity constraint of $U_j$, the following holds:
		\begin{equation*}
			\left|\cost_z(U_j,C,\calB,\avg)-  \cost_z(D_{U_j},C,\calB,\avg) \right| \le O(\eps)\cdot \cost_z(U_j,C,\calB,\avg) +  O\left(\frac{\eps}{k\log (z\eps^{-1})}\right)\cdot \cost_z(P_i,c^\star_i).
		\end{equation*}
	\end{lemma}
	
	
	\noindent
	\Cref{lem:uncolored_groups} provides a guarantee for each $U_j$ similar to that of colored groups in \Cref{lm:color}. 
	Since there are at most $k$ such groups, we can use \Cref{lem:uncolored_groups} to upper bound the error induced by all uncolored groups in $\calG_i$. We defer the proof of \Cref{lem:uncolored_groups} later.
	\begin{proof}[Proof of \Cref{lm:group}]
		Fix a center set $C\in \calX^k$ and a capacity constraint $h\in |G[i]|\cdot \conv(\calB^o)$. Let $\calC_i:=\{G\in\calG_i:G\text{ is colored}\}$ be the collection of all colored groups in $\calG_i$.  Suppose a collection $\{h^G\in |G|\cdot \conv(\calB^o):G\in \calC_i\}\cup \{h^{(j)}\in |U_j|\cdot \conv(\calB^o):j\in [k] \}$ of capacity constraints satisfy that 
		\begin{equation*}
			\sum_{G\in \calC_i}h^G + \sum_{j=1}^k h^{(j)} = h
		\end{equation*}
		and
		\begin{equation}
			\label{eq:G[i]}
			\cost_z(G[i],C,\calB,\avg)=\sum_{G\in \calC_i}\cost_z(G,C,\calB,\avg^G) + \sum_{j=1}^k \cost_z(U_j,C,\calB,\avg^{(j)}).
		\end{equation}
		We have 
		\begin{align*}
			&\quad \cost_z(D[i],C,\calB,\avg)&\\
			&\le \sum_{G\in \calC_i}\cost_z(D_G,C,\calB,\avg^G) + \sum_{j=1}^k \cost_z(D_{U_j},C,\calB,\avg^{(j)}) &(\text{by optimality})\\
			&\le (1+O(\eps))\sum_{G\in \calC_i}\cost_z(G,C,\calB,\avg^G) + k\log(z\eps^{-1})\cdot O\left(\frac{\eps\cdot \cost_z(P_i,c_i^\star)}{k\log(z\eps^{-1})} \right) &(\text{\Cref{lm:color}})\\
			&\quad + (1+O(\eps))\sum_{j=1}^k \cost_z(U_j,C,\calB,\avg^{(j)}) + k\cdot O\left(\frac{\eps\cdot \cost_z(P_i,c_i^\star)}{k\log(z\eps^{-1})} \right)&(\text{\Cref{lem:uncolored_groups}})\\
			&\le (1+O(\eps))\cost_z(G[i],C,\calB,h) + O(\eps)\cdot \cost_z(P_i,c_i^\star)&(\text{Eq. \eqref{eq:G[i]}})
		\end{align*}
		Similarly, we can obtain $\cost_z(G[i],C,\calB,h)\le (1+O(\eps))\cdot \cost_z(D[i],C,\calB,\avg) + O(\eps)\cdot \cost_z(P_i,c_i^\star)$ and thus complete the proof.
	\end{proof}

	\noindent
	It remains to prove \Cref{lem:uncolored_groups}. 
	In the following discussion, we fix $j\in [k]$. For the sake of simplicity, we slightly abuse the notation by denoting $U_j$ as $U$, $D_{U_j}$ as $D$, and the common $\Cfar^G$ among all groups $G\in\calU_j$ as $\Cfar$. Our strategy is similar to that in the ring case, which first shifts the focus from $C$ to $\nu(C)$ based on \Cref{lm:approximation_cost}. Note that for uncolored groups, $\nu(C)\equiv c_i^\star$ since all centers in $C$ are either in $\Cclose$ or $\Cfar$, we have the following inequality:
	\begin{equation*}
		\cost_z(U,C,\calB,h)\in (1+\frac{\eps}{4})\cdot \left(\cost_z(U, c_i^\star, \calB, h) + \phi(C,h) \right),
	\end{equation*}
	where $\phi(C,h)=\sum_{c\in \Cfar} \|h(\cdot,c)\|_1\cdot d(c,c_i^\star)$.
	Notice that 
	\begin{equation*}
		\cost_z(U,c_i^\star,\calB,h)=\min_{\sigma \sim (\calB, \avg)}\sum_{p\in U}\|\sigma(p,\cdot)\|_1\cdot d(p,c_i^\star),
	\end{equation*}
	where the constraint $h$ no longer constrains the capacities of the centers but still indicates the number (fraction) of points in $U$ that should be discarded as outliers. We formalize this result in the following lemma, which is a direct corollary of \Cref{lm:approximation_cost} combined with above observations.
	
	\begin{lemma}[\bf{Reduce the capacity constraint of $h$ on centers}]
		\label{obs:outliers}
		For a feasible capacity constraint $h\in |U|\cdot \conv(\calB^o)$ of $U$, let $m:=w_U(U)-\|h\|_1$ denote the number of outliers, then it holds that
		\begin{equation*}
			\cost_z(U,C,\calB,h)\in (1+\frac{\eps}{4})\cdot \left(\cost_z^{(m)}(U, c_i^\star) + \phi(C,h) \right),
		\end{equation*}
		Here $\cost_z^{(m)}(U,c_i^\star)$ is defined as
		\begin{equation}
			\label{eq:robust}
			\cost_z^{(m)}(U,c_i^\star):=\min_{w:U\to \R_{\ge 0}:\|w\|_1=m,w\le w_U}\sum_{p\in U}\left(w_U(p)-w(p) \right)d^z(p,c_i^\star),
		\end{equation}
		where $w\le w_U$ denotes that $\forall p\in U,w(p)\le w_U(p)$. 
	\end{lemma}
	
	\noindent
	Clearly, \Cref{obs:outliers} works for $D$ as well, i.e.,
	\begin{equation*}
		\cost_z(D,C,\calB,\avg) \in (1+\frac{\eps}{4})\cdot \left(\cost_z^{(m)}(D, c_i^\star) + \phi(C,h) \right).    
	\end{equation*}
	We remark that $\cost_z^{(m)}(U,c_i^\star)$ defined in \eqref{eq:robust} is the same as the objective of \kzmC, which has been investigated in prior studies, such as \cite{Huang2022NearoptimalCF}. The following lemma upper bounds the error between $\cost_z^{(m)}(U,c_i^\star)$ and $\cost_z^{(m)}(D,c_i^\star)$.
	
	\begin{lemma}[\bf{Error analysis for $\cost_z^{(m)}(U,c_i^\star)$}]
		\label{lem:error each uncolored group}
		For real number $0\le m \le w_U(U)$, it holds that
		\begin{equation*}
			\left|\cost_z^{(m)}(U,c_i^\star) - \cost_z^{(m)}(D,c_i^\star) \right|\le \eps\cdot \cost_z^{(m)}(U,c_i^\star)+  O\left(\frac{\eps}{k\log (z\eps^{-1})}\right)\cdot \cost_z(P_i,c^\star_i).
		\end{equation*}
	\end{lemma}
	\begin{proof}  
		The lemma is implied by the proof of \cite[Lemma 3.8]{Huang2022NearoptimalCF}. For completeness, we present the proof here.
		
		We separately prove the following two directions.
		\begin{equation}
			\label{eq:direction1}
			\cost_z^{(m)}(D,c_i^\star)\le (1+\eps)\cost_z^{(m)}(U,c_i^\star) + O\left(\frac{\eps}{k\log (z\eps^{-1})}\right)\cdot \cost_z(P_i,c^\star_i),
		\end{equation}
		\begin{equation}
			\label{eq:direction2}
			\cost_z^{(m)}(U,c_i^\star)\le (1+\eps)\cost_z^{(m)}(D,c_i^\star) + O\left(\frac{\eps}{k\log (z\eps^{-1})}\right)\cdot \cost_z(P_i,c^\star_i).
		\end{equation}
		\paragraph{Proof of \eqref{eq:direction1}} 
		Let $w^\star:U\to \R_{\ge 0}$  be the solution of the optimization problem \eqref{eq:robust}. Namely, it holds that $\|w^\star\|_1=m$, $w^\star(p)\le w_U(p),\forall p\in U$ and $\cost_z^{(m)}(U,c_i^\star)=\sum_{p\in U}(w_U(p)-w^\star(p))d^z(p,c_i^\star)$. Recall that for every $G\in \calU_j$ and $p\in G$, there exists a unique $\lambda_p\in [0,1]$ such that $d^z(p,c_i^\star)=\lambda_p\cdot d^z(\pclose^G,c_i^\star) + (1-\lambda_p)\cdot d^z(\pfar^G,c_i^\star)$, then we construct $w':D \to\R_{\ge 0}$ as follows: for every $G\in \calU_j$ and $c\in C$,
		\begin{equation*}
			w'(\pclose^G) = \sum_{p\in G}\lambda_p\cdot w^\star(p),\quad\text{and}\quad w(\pfar^G) = \sum_{p\in G}(1-\lambda_p)\cdot w^{\star}(p).
		\end{equation*}
		Note that $\|w'\|_1=\|w^\star\| = m$, and $w^\star(p)\le w_U(p),\forall p\in U$ implies that 
		\begin{equation*}
			w'(\pclose^G)\le \sum_{p\in G}\lambda_p\cdot w_U(p) = w_D(\pclose^G),\quad\text{and}\quad w'(\pfar^G)\le \sum_{p\in G}(1-\lambda_p)\cdot w_U(p) = w_D(\pfar^G).
		\end{equation*}
		Hence, $w'$ is a feasible solution of the optimization problem $\cost_z^{(m)}(D,c_i^\star)$. We have
		\begin{align*}
			&\cost_z^{(m)}(D,c_i^\star)\\
			\le\quad & \sum_{p\in D}\left(w_D(p)-w'(p) \right) d^z(p,c_i^\star)\\
			=\quad&\sum_{G\in\calU_j}\left(\left(w_D(\pclose^G) - w'(\pclose^G) \right) d^z(\pclose^G,c_i^\star) + \left(w_D(\pfar^G) - w'(\pfar^G) \right) d^z(\pfar^G,c_i^\star) \right)\\
			=\quad&\sum_{G\in\calU_j}\left(\sum_{p\in G} (w_U(p) - w^\star(p)) \lambda_p\cdot d^z(\pclose^G,c_i^\star) + \sum_{p\in G} (w_U(p) - w^\star(p)) (1-\lambda_p)\cdot d^z(\pfar^G,c_i^\star) \right)\\
			=\quad&\sum_{G\in \calU_j}\sum_{p\in G}(w_U(p)-w^\star(p))\cdot d^z(p,c_i^\star)\\
			=\quad&\cost_z^{(m)}(U,c_i^\star),
		\end{align*}    
		which completes the proof of \eqref{eq:direction1}.
		\paragraph{Proof of \eqref{eq:direction2}}
		
		Let $\{m_G\}_{G\in\calU_j}$ be a sequence of positive real numbers such that $m=\sum_{G\in \calU_j}m_G$ and $\cost_z^{(m)}(D,c_i^\star) = \sum_{G\in\calU_j}\cost_z^{(m_G)}(D_G,c_i^\star)$. We then de a case study for every group $G\in \calU_j$.
		\begin{itemize}
			\item If $m_G=0$, then $\cost_z^{(m_G)}(G,c_i^\star) = \cost_z(G,c_i^\star)$ and $\cost_z^{(m_G)}(D_G,c_i^\star) = \cost_z(D_G,c_i^\star)$. By definition of two-point coresets (\Cref{def:twopoints}), we have $\cost_z(D_G,c_i^\star)=\cost_z(G,c_i^\star)$ and hence $\cost_z^{(m_G)}(G,c_i^\star) = \cost_z^{(m_G)}(D_G,c_i^\star)$.
			\item If $m_G=w_{D}(D_G)=w_U(G)$, then $\cost_z^{(m_G)}(G,c_i^\star)=\cost_z^{(m_G)}(D_G,c_i^\star)=0$.
			\item If $0<m_G<w_D(D_G)$, we call such group a \emph{special} group. By \cite[Lemma 3.17]{Huang2022NearoptimalCF}, there are at most $O(1)$ special groups. By a similar argument as in the proof of \Cref{lm:color}, we can obtain the following inequality.
			\begin{equation*}
				\cost_z^{(m_G)}(G,c_i^\star) \le (1+\eps)\cost_z^{(m_G)}(D_G,c_i^\star) + \frac{\eps \cdot \cost_z(P_i,c_i^\star)}{k\log(z\eps^{-1})}.
			\end{equation*}
		\end{itemize}
		Putting everything together, we have 
		\begin{align*}
			&\cost_z^{(m)}(U,c_i^\star)\\
			\le\quad&\sum_{G\in \calU_j}\cost_z^{(m_G)}(G,c_i^\star)\\
			\le\quad&(1+\eps)\sum_{G\in \calU_j}\cost_z^{(m_G)}(D_G,c_i^\star) + O(1)\cdot \frac{\eps \cdot \cost_z(P_i,c_i^\star)}{k\log(z\eps^{-1})}\\
			=\quad&(1+\eps)\cost_z^{(m)}(D,c_i^\star) + O\left(\frac{\eps}{k\log (z\eps^{-1})}\right)\cdot \cost_z(P_i,c^\star_i) \\
		\end{align*}
		which completes the proof.
	\end{proof}
	\begin{proof}[Proof of \Cref{lem:uncolored_groups}]
		\Cref{lem:uncolored_groups} is a direct corollary of \Cref{obs:outliers} and \ref{lem:error each uncolored group}.
	\end{proof}
	\section{Bounding the Lipschitz Constant $\lip(\calB)$}
	\label{sec:Lipschitz}
	
	In this section, we give upper and lower bounds for the Lipschitz constant $\lip(\calB)$, for various notable cases of $\calB$. 
	In particular, we provide an upper bound $\lip(\calB)$ for an important class of assignment structure constraints $\calB$, called matroid basis polytopes (Theorems~\ref{thm:matroid} and~\ref{thm:laminar_matroid}).
	We also show that $\lip(\calB)$ may be unbounded with knapsack constraints (Theorem~\ref{thm:knapsack_2}). We review the definitions and some useful properties
	as follows.

	Firstly, we introduce a general assignment structure constraint captured by the so-called matroid basis polytope. For more information about matroid and matroid basis polytope, see the classic reference \cite{LexSchrijver2003CombinatorialO}. 
	%
	\begin{definition}[\bf{Matroid}]
		\label{def:matroid}
		Given a ground set $E$, a family $\calM$ of subsets of $E$ is a {\em matroid} if 
		\begin{itemize}
			\item $\emptyset \in \calM$;
			\item If $I\in \calM$ and $I'\subset I$, then $I'\in \calM$;
			\item If $I, I'\in \calM$ and $|I| < |I'|$, then there must exist an element $a\in I'\setminus I$ such that $I\cup \left\{a\right\}\in \calM$.
		\end{itemize}
		Each $I\in \calM$ is called an independent set.
		The maximum size of an independent set is called the rank of $\calM$, denoted by $\rank(\calM)$.
		Each set $I\in \calM$ of size equal to the rank is called a basis of $\calM$.
	\end{definition}
	
	\noindent
	Matroid is a very general combinatorial structure
	that generalizes many set systems including uniform matroid, partition matroid, laminar matroid, regular matroid, graphic matroid, transversal matroid and so on.
	Now we define the matroid basis polytope.
	
	\begin{definition}[\bf{Matroid basis polytope}]
		\label{def:polytope}
		Let $E$ be a ground set and let $\calM$ be a matroid on $E$.
		For each basis $I\in \calM$, let $e_I:= \sum_{i\in I} e_i$ denote the indicator vector of $I$, where $e_i\in \R^{|E|}$ is the standard $i$-th unit vector.
		The matroid basis polytope $P_{\calM}$ is the convex hull of the set
		\[
		P_{\calM}:=\left\{e_I: \text{$I$ is a basis of $\calM$}\right\}.
		\]
	\end{definition}
	
	\noindent
	The following definition provides another description of $P_{\calM}$ by rank functions.

	\begin{definition}[\bf{Rank function and $P_{\calM}$}]
		\label{def:rank}
		Let $E$ be a ground set of size $n\geq 1$ and let $\calM$ be a matroid on $E$.
		The rank function $\rank: 2^E\rightarrow \mathbb{Z}_{\geq 0}$ of $\calM$ is defined as follows: for every $A\subseteq E$, $\rank(A) = \max_{I\in \calM} |I\cap A|$.
		Moreover, the matroid basis polytope can be equivalently defined by
		the following linear program (see e.g.,\cite{LexSchrijver2003CombinatorialO}):
		\[
		P_{\calM} = \Bigl\{x\in \R_{\geq 0}^n: \sum_{i\in A} x_i \leq \rank(A), \forall A\in \calA;\  \|x\|_1 = \rank(\calM)\Bigr\}.
		\]
		Let $\avg\in P_{\calM}$ be a point inside the matroid basis polytope. 
		We say a subset $A\subseteq E$ is \emph{tight} on $\avg$ if $\sum_{i\in A} \avg_i = \rank(A)$.
	\end{definition}
	
	\noindent
	%
	We also consider a specific type of matroid, called laminar matroid.

	\begin{definition}[\bf{Laminar matroid}]
		\label{def:laminar}
		Given a ground set $E$, a family of $\calA$ of subsets of $E$ is called \emph{laminar} if for every two subsets $A,B\in \calA$ with $A\cap B\neq \emptyset$, either $A\subseteq B$ or $B\subseteq A$.
		Assume $E\in \calA$ and we define the \emph{depth} of a laminar $\calA$ to be the largest integer $\ell\geq 1$ such that there exists a sequence of subsets $A_1,\ldots, A_{\ell}\in \calA$ with $A_1\subsetneq A_2\subsetneq \cdots \subsetneq A_{\ell} = E$.

		We say a family $\calM$ of subsets of $E$ is a laminar matroid if there exists a laminar $\calA$ and a capacity function $u: \calA \rightarrow \Z_{\geq 0}$ such that $\calM = \left\{ I\subseteq E: |I\cap A|\leq u(A), \forall A\in \calA \right\}$.
	\end{definition}

	\noindent
	%
	By Definition~\ref{def:rank}, we know that $\rank(A) \leq u(A)$ holds for every $A\in \calA$. 
	Specifically, we can see that a {\em uniform matroid} is a laminar matroid of depth 1 
	(i.e., there is a single cardinality constraint over the entire set $E$)
	and a {\em partition matroid} is a laminar matroid of depth 2
	(i.e., there is a partition of $E$ and each partition has a cardinality constraint).

	\paragraph{Technical Overview}
	For $\calB = \Delta_k$ which is the unconstrained case, we have $\lip(\calB) = 1$.
	However, the geometric structure of $\calB$ can indeed result in a significantly larger, even unbound, $\lip(\calB)$.
	In particular, we show that the value $\lip(\calB)$ is unbounded,
	even for a very simple $\calB$ defined by two knapsack constraints
	(see Theorem~\ref{thm:knapsack_2}).
	On the other hand, we do manage to show $\lip(\calB) \leq k-1$ for matroid basis polytopes (\Cref{thm:matroid}).
	To analyze $\lip(\calB)$ for the matroid case, we first show in \Cref{lm:upper_fB} that it suffices to restrict our attention to the case where the initial assignment $\sigma$ corresponds to vertices of the polytope (or basis) and so does terminal assignment $\sigma'$ (which we need to find). 
	Our argument is combinatorial and heavily relies on constructing augmenting paths (\Cref{def:augmenting}). 
	%
	%
	Roughly speaking, an augmenting path is defined by a sequence of basis $(I_1, \ldots, I_m)$ and a sequence of elements $(a_0,a_1,\ldots,a_m)$
	where each $I_{i}$ is an initial basis and the following exchange property holds:
	$I_i\cup \left\{a_{i-1}\right\}\setminus \left\{a_{i}\right\} \in \calM$ for all $i$.
	The exchange property allows us to perform a sequence of basis exchanges
	to transport the mass from $a_0$ to $a_m$.
	%
	Hence, the key is to find such an augmenting path in which $a_0$ is an element
	with $\avg_{a_0}<\avg'_{a_0}$ and $a_m$ is one with $\avg_{a_m}>\avg'_{a_m}$,
	for some $\avg, \avg' \in \calB$.
	Then performing the exchange operation reduces the value of $\|\avg-\avg'\|_1$ by $2\tau$ with total transportation cost $2m\tau$ for some small $\tau>0$.
	Thus, $\lip(\calB)$ is bounded by the path length $m$.
	We show that $m\leq k-1$ when $\calB$ is a matroid basis polytope (\Cref{claim:path_general}).
	To show the existence of such an augmenting path, we leverage several combinatorial properties 
	of a matroid, as well as the submodularity of the rank function, and the non-crossing property of tight subsets.
	We further improve the bound to $m\leq \ell+1$ for laminar matroid basis polytopes of depth at most $\ell\geq 1$ (\Cref{claim:path_laminar}), which results in bounds $\lip(\calB)\leq 2$ for uniform matroids and $\lip(\calB)\leq 3$ for partition matroids.

	\subsection{Lipschitz Constant for (Laminar) Matroid Basis Polytopes}
	\label{sec:Lipschitz_matroid}

	For ease of analysis, we first propose another Lipschitz constant on $\calB$ for general polytopes (\Cref{def:mass_II}). 
	%
	%
	For preparation, we need the following notion of coupling of distributions (see e.g., \cite{mitzenmacher2017probability}).

	\begin{definition}[\bf{Coupling of distributions}]
		\label{def:coupling}
		Given two distributions $(\mu,\mu')\in \Delta_m$ ($m\geq 1$), $\kappa: [m]\times [m]\rightarrow \R_{\geq 0}$ is called a coupling of $(\mu,\mu')$, denoted by $\kappa\vdash (\mu,\mu')$, if 
		\begin{itemize}
			\item For any $i\in [m]$, $\sum_{j\in [m]}\kappa(i,j) = \mu_i$;
			\item For any $j\in [m]$, $\sum_{i\in [m]}\kappa(i,j) = \mu'_j$.
		\end{itemize}
	\end{definition}
	
	\noindent
	Now we define the Lipschitz constant for a restricted assignment transportation problem, 
	in which the initial and terminal assignments are restricted to the vertices of polytope $\calB$.
	
	\begin{definition}[\bf{Lipschitz constant for restricted \AT}]
		\label{def:mass_II}
		Let $\calB\subseteq c\cdot \Delta_k$ for some $c > 0$ be a polytope.
		Let $V$ denote the collection of vertices of $\calB$.
		Let $\avg, \avg'\in \calB$ be two points inside $\calB$.
		Let $\mu \in \Delta_{|V|}$ denote a distribution on $V$ satisfying that
		\[
		\sum_{v\in V} \mu_v\cdot v = \avg.
		\]
		We define
		\[
		\widetilde{\AT}(\calB, \avg, \avg', \mu):= \min_{\substack{\mu'\in \Delta_{|V|}: \sum_{v\in V} \mu'_v\cdot v = \avg' \\ \kappa\vdash (\mu,\mu')}} \sum_{v,v'\in V}\kappa(v,v')\cdot \|v-v'\|_1
		\]
		to be the optimal assignment transportation cost constrained to vertices of $\calB$.
		Define the Lipschitz constant for the restricted $\AT$ on $\calB$ to be
		\[
		\widetilde{\lip}(\calB) := \max_{\substack{\avg, \avg'\in \calB \\ \mu \in \Delta_{|V|}: \sum_{v\in V} \mu_v\cdot v = \avg }} \frac{\widetilde{\AT}(\calB, \avg, \avg', \mu)}{\|\avg-\avg'\|_1}.
		\]
	\end{definition}
	
	\noindent
	By definition, we know that $\widetilde{\lip}(\calB) = \widetilde{\lip}(c\cdot \calB)$ for any $c>0$, i.e., $\widetilde{\lip}$ is scale-invariant on $\calB$.
	We will analyze $\widetilde{\lip}(P_{\calM})$ for matroid basis polytopes in the later sections. 
	Now, we show the following lemma that connects two Lipschitz constants.

	\begin{lemma}[\bf{Relation between $\lip(\calB)$ and $\widetilde{\lip}(\calB)$}]
		\label{lm:upper_fB}
		Let $\calB\subseteq \Delta_k$ be a polytope.
		We have $\lip(\calB) \leq \widetilde{\lip}(\calB)$.
	\end{lemma}
	
	\begin{proof}
		Let $V$ denote the collection of vertices of $\calB$.
		\eat{
			For convenience, we define 
			\[
			\AT(\calB,\avg,\avg', \sigma) := \frac{\AT(\calB,\avg,\avg', \sigma)}{\|\avg - \avg'\|_1} = \min_{\substack{
					\sigma'\sim (\calB,\avg') : \\
					\forall p \in [n], \|\sigma'(p, \cdot)\|_1 = \|\sigma(p, \cdot)\|_1 \\
			}} \sum_{p\in [n]} \|\sigma(p,\cdot) - \sigma'(p,\cdot)\|_1.
			\]
			Then 
			\[
			\lip(\calB) := \max_{\substack{\avg, \avg'\in \calB \\ \sigma\sim (\calB, \avg): \|\sigma\|_1=1}} \AT(\calB,\avg,\avg', \sigma).
			\]
		}
		%
		%
		We define a collection $\calH$ of assignment functions $\sigma$ such that for every $\sigma\in \calH$: for every $p\in [n]$, there exists a vertex $v\in V$ such that $\frac{\sigma(p,\cdot)}{\|\sigma(p,\cdot)\|_1} = v$, i.e., the assignment vector of each $p\in [n]$ is equal to a scale of some vertex of $\calB$.
		We have the following claim.

		\begin{claim}[\bf{An equivalent formulation of $\lip(\calB)$}]
			\label{claim:calG}
			$\lip(\calB) = \max_{\substack{\avg, \avg'\in \calB \\ \sigma\in \calH}} \frac{\AT(\calB,\avg,\avg', \sigma)}{\|\avg-\avg'\|_1}$.
		\end{claim}
		
		\begin{proof}
			Fix $\avg, \avg'\in \calB$ and an assignment function $\sigma\sim (\calB, \avg)$.
			It suffices to show the existence of another assignment function $\pi \in \calH$ such that
			\begin{align}
				\label{eq:proof_claim:calG}
				\AT(\calB, \avg, \avg', \sigma) \leq \AT(\calB, \avg, \avg', \pi).
			\end{align}
			For every $p\in [n]$, since $\frac{\sigma(p,\cdot)}{\|\sigma(p,\cdot)\|_1 }\in \calB$, we can rewrite 
			\[
			\sigma(p,\cdot) = \|\sigma(p,\cdot)\|_1 \cdot \sum_{v\in V} \alpha_v^p\cdot v
			\]
			for some distribution $\alpha^p\in \Delta_{|V|}$.
			Now we consider a weighted set $V$ together with weights $w_V(v) = \sum_{p\in [n]} \|\sigma(p,\cdot)\|_1\cdot \alpha_v^p$, and we construct $\pi: V\times C\rightarrow \R_{\geq 0}$ as follows: for every $v\in V$, let $\pi(v,\cdot) = w_V(v) \cdot v$.
			Since 
			\begin{align*}
				\sum_{v\in V} \pi(v,\cdot) & =  \sum_{v\in V, p\in [n]} \|\sigma(p,\cdot)\|_1\cdot \alpha_v^p\cdot v & (\text{Defn. of $\pi$}) \\
				& =  \sum_{p\in [n]} \sigma(p,\cdot) & (\text{Defn. of $\alpha^p$}) \\
				& =  \avg, & (\sigma\sim \avg)
			\end{align*}
			we know that $\pi\sim (\calB, \avg)$, which implies that $\pi\in \calH$.

			Let $\pi': V\times C\rightarrow \R_{\geq 0}$ with $\pi'\sim (\calB, \avg')$ be an assignment function such that
			\[
			\AT(\calB, \avg, \avg', \pi) = \sum_{v\in V} \|\pi(v,\cdot) - \pi'(v,\cdot) \|_1 .
			\]
			We construct an assignment function $\sigma': [n]\times [k]\rightarrow \R_{\geq 0}$ as follows: for every $p\in [n]$, let
			\[
			\sigma'(p,\cdot) = \sum_{v\in V} \frac{\|\sigma(p,\cdot)\|_1\cdot\alpha^p_v}{w_V(v)} \cdot \pi'(v,\cdot).
			\]
			We first verify that $\sigma'\sim (\calB, \avg')$.
			Note that for every $p\in [n]$
			\begin{align*}
				\|\sigma'(p,\cdot )\|_1 & =  \sum_{v\in V} \frac{\|\sigma(p,\cdot)\|_1\cdot\alpha^p_v}{w_V(v)} \cdot \|\pi'(v,\cdot)\|_1 & (\text{Defn. of $\sigma'$}) \\
				& =  \|\sigma(p,\cdot)\|_1\cdot \sum_{v\in V} \alpha^p_v & (\text{Defn. of $\pi'$}) \\
				& =  \|\sigma(p,\cdot)\|_1, & (\alpha^p\in \Delta_{|V|})
			\end{align*}
			which implies that $\sigma'(p,\cdot)\in \|\sigma(p,\cdot)\|_1\cdot \calB$.
			Also, we have
			\begin{align*}
				\sum_{p\in [n]} \sigma'(p,\cdot) & =  \sum_{p\in [n], v\in V} \frac{\|\sigma(p,\cdot)\|_1\cdot\alpha^p_v}{w_V(v)} \cdot \pi'(v,\cdot) & (\text{Defn. of $\pi'$})\\
				& =  \sum_{v\in V} \pi'(v,\cdot) & (\text{Defn. of $w_V(v)$})\\
				& =  \avg', & (\pi'\sim \avg')
			\end{align*}
			which implies that $\sigma'\sim \avg'$.
			Thus, $\sigma'\sim (\calB, \avg')$ holds.
			Finally, we have
			\begin{align*}
				& \sum_{p\in [n]} \|\sigma(p,\cdot) - \sigma'(p,\cdot) \|_1 & \\
				= \,\,\,&  \sum_{p\in [n]} \|\sigma(p,\cdot)\|_1\cdot\|\sum_{v\in V} \alpha_v^p \cdot v - \frac{\alpha_v^p}{w_V(v)}\cdot \pi'(v,\cdot) \|_1 & (\text{Defns. of $\alpha^p$ and $\sigma'$})\\
				\leq \,\,\, &  \sum_{p\in [n]} \sum_{v\in V} \frac{\|\sigma(p,\cdot)\|_1\cdot \alpha_v^p}{w_V(v)}\cdot \| w_V(v)\cdot v -\pi'(v,\cdot) \|_1 & \\
				\leq \,\,\, &  \sum_{v\in V} \sum_{p\in [n]} \frac{\|\sigma(p,\cdot)\|_1\cdot \alpha_v^p}{w_V(v)}\cdot \|\pi(v,\cdot) - \pi'(v,\cdot)\|_1 & (\text{Defn. of $\pi$})\\
				= \,\,\, &  \sum_{v\in V} \|\pi(v,\cdot) - \pi'(v,\cdot) \|_1, & (\text{Defn. of $w_V(v)$})
			\end{align*}
			which completes the proof of Claim~\ref{claim:calG}.
		\end{proof}
		
		\noindent
		Fix $\avg, \avg'\in \calB$ and $\sigma: [n]\times [k]\rightarrow \R_{\geq 0}\in \calH$.
		We construct a distribution $\mu \in \Delta_{|V|}$ as follows: for every $v\in V$, let 
		\[
		\mu_v = \sum_{p\in [n]: \sigma(p,\cdot) = \|\sigma(p,\cdot)\|_1\cdot v} \|\sigma(p,\cdot)\|_1.
		\]
		By construction, we know that 
		\begin{align*}
			\sum_{v\in V} \mu_v \cdot v &= \sum_{v\in V} \sum_{p\in [n]: \sigma(p,\cdot) = \|\sigma(p,\cdot)\|_1\cdot v} \|\sigma(p,\cdot)\|_1\cdot v & (\text{Defn. of $\mu$})\\
			&= \sum_{p\in [n]} \sigma(p,\cdot) &\\
			&= \avg, & (\sigma\sim \avg).
		\end{align*}
		Let $\mu'\in \Delta_{|V|}$ with $\sum_{v\in V} \mu'_v\cdot v = \avg'$ and $\kappa\vdash (\mu,\mu')$ satisfy that
		\[
		\widetilde{\AT}(\calB, \avg, \avg', \mu) = \sum_{v,v'\in V}\kappa(v,v')\cdot \|v-v'\|_1.
		\]
		Next, we construct another assignment function $\sigma'\sim (\calB, \avg')$ as follows: for every $p\in [n]$ with $\sigma(p,\cdot) = \|\sigma(p,\cdot)\|_1\cdot v$, let
		\[
		\sigma'(p,\cdot) = \frac{\|\sigma(p,\cdot)\|_1}{\mu_v}\cdot \sum_{v'\in V} \kappa(v,v')\cdot v'.
		\]
		Note that for every $p\in [n]$, $\|\sigma'(p,\cdot)\|_1 = \frac{\|\sigma(p,\cdot)\|_1}{\mu_v}\cdot\sum_{v'\in V} \kappa(v,v') =  \|\sigma(p,\cdot)\|_1$, which implies that $\sigma'(p,\cdot)\in \|\sigma(p,\cdot)\|_1\cdot \calB$.
		Also, we have
		\begin{align*}
			\sum_{p\in [n]} \sigma'(p,\cdot) & =  \sum_{v\in V} \sum_{p\in [n]: \sigma(p,\cdot) = \|\sigma(p,\cdot)\|_1 \cdot v} \frac{\|\sigma(p,\cdot)\|_1}{\mu_v}\cdot \sum_{v'\in V} \kappa(v,v')\cdot v' & (\text{Defn. of $\sigma'$})\\
			& =  \sum_{v, v'\in V} \kappa(v,v')\cdot v' & \\
			& =   \sum_{v'\in V} \mu_{v'}\cdot v' & (\kappa\vdash (\mu,\mu'))\\
			& =   \avg', &
		\end{align*}
		which implies that $\sigma'\sim (\calB, \avg')$.
		Moreover,
		\begin{align}
			\label{eq:proof_lm:upper_fB}
			\begin{aligned}
				& \AT(\calB, \avg, \avg', \sigma) & \\
				\leq \,\,\,&  \sum_{p\in [n]} \|\sigma(p,\cdot) - \sigma'(p,\cdot) \|_1  & (\text{by optimality}) \\
				= \,\,\,&  \sum_{v\in V} \sum_{p\in [n]: \sigma(p,\cdot) = \|\sigma(p,\cdot)\|_1\cdot v} \|\sigma(p,\cdot) - \sigma'(p,\cdot)\|_1  & \\
				= \,\,\,&  \sum_{v\in V} \sum_{p\in [n]: \sigma(p,\cdot) = \|\sigma(p,\cdot)\|_1\cdot v} \frac{\|\sigma(p,\cdot)\|_1}{\mu_v}\cdot \|\mu_v\cdot v- \sum_{v'\in V} \kappa(v,v')\cdot v'\|_1 & (\text{Defns. of $\sigma$ and $\sigma'$}) \\
				= \,\,\, &  \sum_{v\in V} \|\mu_v\cdot v- \sum_{v'\in V} \kappa(v,v')\cdot v'\|_1  & \\
				\leq \,\,\, &  \sum_{v,v'\in V} \kappa(v,v')\cdot \|v - v'\|_1  & \\
				= \,\,\, &  \widetilde{\AT}(\calB, \avg, \avg', \mu), & (\text{Defns. of $\mu'$ and $\kappa$})
			\end{aligned}
		\end{align}
		which implies that
		\begin{align*}
			\lip(\calB) & =  \max_{\substack{\avg, \avg'\in \calB \\ \sigma\in \calH: \|\sigma\|_1=1}} \frac{\AT(\calB,\avg,\avg', \sigma)}{\|\avg - \avg'\|_1} & (\text{Claim~\ref{claim:calG}}) \\
			& \leq  \max_{\substack{\avg, \avg'\in \calB\\ \mu \in \Delta_{|V|}: \sum_{v\in V} \mu_v\cdot v = \avg} } \frac{\widetilde{\AT}(\calB,\avg,\avg', \mu)}{\|\avg - \avg'\|_1} & (\text{Ineq.~\eqref{eq:proof_lm:upper_fB}}) \\
			& =  \widetilde{\lip}(\calB). &
		\end{align*}
		Thus, we complete the proof.
	\end{proof}

	\subsubsection{Lipschitz Constant for Matroid Basis Polytopes}
	\label{sec:matroid}
	
	For preparation, we need the following lemmas that provide well-known properties of matroids; see~\cite{LexSchrijver2003CombinatorialO} for more details.
	The first is an easy consequence of the submodularity of the rank function
	and the second easily follows from the exchange property of matroid.

	\begin{lemma}[\bf{Properties of rank function}]
		\label{lm:rank}
		Let $E$ be a ground set and let $\calM$ be a matroid on $E$ with a rank function $\rank: 2^E\rightarrow \calZ_{\geq 0}$.
		Let $\avg\in P_{\calM}$ be a point inside the matroid basis polytope. 
		Recall that we say a subset $A\subseteq E$ is \emph{tight} on $\avg$ if $\sum_{i\in A} \avg_i = \rank(A)$.
		If two subsets $A,B\subseteq E$ are tight on $\avg$, then both $A\cup B$ and $A\cap B$ are tight on $\avg$.
	\end{lemma}
	
	\begin{lemma}[\bf{Circuit}]
		\label{lm:circuit}
		Let $E$ be a ground set of size $n\geq 1$ and let $\calM$ be a matroid on $E$.
		Let $I\in \calM$ be a basis and $a\in E\setminus I$ be an element.
		Let circuit $C(I,a)$ be the smallest tight set on $e_I$ that contains $a$.
		We have that $C(I,a)\subseteq I$ and for every element $b\in C(I,a)$, $I\cup \left\{a\right\}\setminus \left\{b\right\} \in \calM$.
	\end{lemma}
	
	\noindent
	We show that $\lip(P_{\calM})\leq |E|-1$ for matroid basis polytopes by the following theorem.
	This theorem is useful since $|E| = k$ for every assignment structure constraint $\calB$, and hence, $\lip(\calB)\leq k-1$ when $\calB$ is a scaled matroid basis polytope.

	\begin{theorem}[\bf{Lipschitz constant for matroid basis polytopes}]
		\label{thm:matroid}
		Let $E$ be a ground set and let $\calM$ be a matroid on $E$.
		We have $\lip(P_{\calM})\leq \widetilde{\lip}(P_{\calM}) \leq |E|-1$.
	\end{theorem}
	
	\begin{proof}
		Let $V = \left\{e_I: \text{$I$ is a basis of $\calM$}\right\}$ be the vertex set of $P_{\calM}$.
		Fix $\avg, \avg'\in P_{\calM}$ and a distribution $\mu\in \Delta_{|V|}$ with $\sum_{v\in V} \mu_v\cdot v = \avg$.
		We aim to show that 
		\[
		\widetilde{\AT}(\calB, \avg, \avg', \mu) = \min_{\substack{\mu'\in \Delta_{|V|}: \sum_{v\in V} \mu'_v\cdot v = \avg' \\ \kappa\vdash (\mu,\mu')}} \sum_{v,v'\in V}\kappa(v,v')\cdot \|v-v'\|_1
		\leq (|E|-1)\cdot \|\avg - \avg'\|_1.
		\]
		We first have the following claim.
		
		\begin{claim}[\bf{Path decomposition from $\avg$ to $\avg'$.}]
			\label{claim:matroid_two}
			There exists a sequence $\avg^{(0)} = \avg, \avg^{(1)}, \ldots, \avg^{(m)} = \avg' \in P_{\calM}$ satisfying
			\begin{enumerate}
				\item For every $i\in [m]$, $\|\avg^{(i-1)} - \avg^{(i)}\|_0 = 2$;
				\item $\sum_{i\in [m]} \|\avg^{(i-1)} - \avg^{(i)}\|_1 = \|\avg-\avg'\|_1$.
			\end{enumerate}
		\end{claim}
		
		\begin{proof}
			We first show how to construct $\avg^{(1)}$ such that $\|\avg^{(0)} - \avg^{(1)}\|_0 = 2$ and $\|\avg^{(0)} - \avg^{(1)}\|_1 + \|\avg^{(1)} - \avg'\|_1 = \|\avg-\avg'\|_1$.

			Let $H^+ = \left\{i\in E: \avg_i > \avg'_i \right\}$ and $H^- = \left\{i\in E: \avg_i < \avg'_i\right\}$.
			Consider for every element $j\in H^-$, the smallest tight set $S_{j}$ on $\avg$ that contains $j$ (inclusion-wise).
			By Lemma~\ref{lm:rank}, we know that such $S_{j}$ is unique.
			If $S_{j} \cap H^+ \neq \emptyset$ for some $j^\star\in H^-$, we can choose an arbitrary element $i^\star \in S_{j^\star} \cap H^+$. 
			Let $\tau = \min_{A\subseteq E: \sum_{i\in A} \avg_i < \rank(A)} \left\{\rank(A) - \sum_{i\in A} \avg_i \right\}$.
			We construct $\avg^{(1)} = \avg + \tau\cdot (e_{j^\star} - e_{i^\star})$.
			It is easy to see that $\avg^{(1)}\in P_{\calM}$, since for every $S\subseteq E$, 
			\begin{enumerate}
				\item If $S$ is not tight on $\avg$, we have $\sum_{i\in S} \avg^{(1)}_i \leq \sum_{i\in S} \avg^{(0)}_i + \tau\leq \rank(S)$ by the definition of $\tau$.
				\item If $S$ does not contain $j^\star$, then $\sum_{i\in S} \avg^{(1)}_i\leq \sum_{i\in S} \avg^{(0)}_i\leq \rank(S)$ by the construction of $\avg^{(1)}$.
				\item If $S$ is tight on $\avg$ and $j^\star\in S$, then $S$ must contain $S_{j^\star}$ by Lemma~\ref{lm:rank}, which makes $\sum_{i\in S} \avg^{(1)}_i = \sum_{i\in S} \avg^{(0)}_i\leq \rank(S)$.
			\end{enumerate}
			
			Now, suppose for every $j\in H^-$ we have $S_j \cap H^+ = \emptyset$.
			Consider the union $U$ of all such $S_j$s.
			We can easily see that $H^-\subseteq U$ and $U\cap H^+=\emptyset$.
			By Lemma~\ref{lm:rank}, $U$ is also tight on $\avg$, say $\sum_{i\in U} \avg_i=\rank(U)$.
			Since $\avg'\in P_{\calM}$, $\sum_{i\in U} \avg'_i\leq \rank(U)$.
			However, by the definition of $H^-$, we have 
			\[
			\sum_{i\in U} \avg_i < \sum_{i\in U} \avg'_i \leq \rank(U),
			\]
			which is a contradiction.
			Hence, this case is impossible and we can always find $j\in H^-$ with $S_{j} \cap H^+ \neq \emptyset$ and construct $\avg^{(1)} = \avg + \tau\cdot (e_{j^\star} - e_{i^\star})$.

			We then can repeat the above procedure and construct a sequence $\avg^{(1)}, \avg^{(2)}, \ldots $.
			By the previous argument, at each iteration $t$ we increase the number of tight sets on $\avg^{(t)}$ by at least one from $\avg^{(t-1)}$.
			Since there is a finite number of tight sets (at most $2^{|E|}$), the above process terminates in a finite number of times and arrives $\avg'$.
			Hence, we complete the proof of Claim~\ref{claim:matroid_two}. %
		\end{proof}
		
		\noindent
		Suppose $\widetilde{\AT}(\calB, \avg, \avg', \mu)\leq (|E|-1)\cdot \|\avg - \avg'\|_1$ holds for the case of $\|\avg - \avg'\|_0 = 2$.
		Let $\mu^{(0)} = \mu$.
		For every $i\in [m]$, we consecutively construct a distribution $\mu^{(i)}\in \Delta_{|V|}$ with $\sum_{v\in V}\mu^{(i)}_v\cdot v = \avg^{(i)}$ and a coupling $\kappa^{(i)} \vdash (\mu^{(i-1)}, \mu^{(i)})$ such that
		\begin{align}
			\label{eq1:proof_thm:matroid}
			\widetilde{\AT}(\calB, \avg^{(i-1)}, \avg^{(i)}, \mu^{(i-1)}) = \sum_{v,v'\in V} \kappa^{(i)}(v,v')\cdot \|v-v'\|_1\leq (|E|-1)\cdot \|\avg^{(i-1)} - \avg^{(i)}\|_1.
		\end{align}
		Then we have
		\begin{align*}
			\widetilde{\AT}(P_{\calM}, \avg, \avg', \mu) 
			= & \,\, \min_{\substack{\mu'\in \Delta_{|V|}: \sum_{v\in V} \mu'_v\cdot v = \avg'\\ \kappa\vdash (\mu,\mu')} } \sum_{v,v'\in V}\kappa(v,v')\cdot \|v-v'\|_1 & \\
			\leq & \sum_{i\in [m]} \sum_{v,v'\in V} \kappa^{(i)}(v,v')\cdot \|v-v'\|_1  & (\text{Defn. of $\kappa^{(i)}$}) \\
			\leq & \,\,  (|E|-1)\cdot \sum_{i\in [m]} \|\avg^{(i-1)} - \avg^{(i)}\|_1  & (\text{Ineq.~\eqref{eq1:proof_thm:matroid}}) \\
			= &  \,\, (|E|-1)\cdot \|\avg - \avg'\|_1. & (\text{Claim~\ref{claim:matroid_two}}) 
		\end{align*}
		Thus, we only prove for the case that 
		$\avg$ and $\avg'$ differ in exactly two entries $\|\avg - \avg'\|_0 = 2$.
		%
		%
		%
		We define the following augmenting path.

		\begin{definition}[\bf{Augmenting path from $\avg$ to $\avg'$}]
			\label{def:augmenting}
			Let $\avg$ and $\avg'$ be two vectors in $\calB$ that differ in exactly two entries,
			$\|\avg - \avg'\|_0 = 2$.
			W.l.o.g., we assume the corresponding two elements are 
			$s, t\in E$ with $\avg_s < \avg'_s$ and $\avg_t > \avg'_t$ respectively.
			Let $\mu \in \Delta_{|V|}$ be the distribution associated with $\avg$ 
			(i.e., $\sum_{v\in V} \mu_v\cdot v = \avg$).
			We say a sequence of indicator vector of basis $(v_1 = e_{I_1}, \ldots, v_m = e_{I_m})$ 
			together with a sequence of indices of elements $(a_0 = s, a_1, \ldots, a_{m-1}, a_m = t)$
			($a_i\in E$ for $0\leq i\leq m$)
			form a weak augmenting path of length $m\geq 1$ from $\avg$ to $\avg'$ if
			\begin{enumerate}
				\item For every $i\in [m]$, $\mu_{v_i} > 0$ (i.e., every $v_i$ appears in the support of $\mu$);
				\item For every $i\in [m]$, $I_i\cup \left\{a_{i-1}\right\}\setminus \left\{a_{i}\right\} \in \calM$ (weak exchange property).
			\end{enumerate}
			Here, we do not require that $I_i$s and $a_i$s be distinct.
			%
			
			Moreover, we say $(v_1 = e_{I_1}, \ldots, v_m = e_{I_m})$ and 
			$(a_0 = s, a_1, \ldots, a_{m-1}, a_m = t)$ form a strong augmenting path if it is a weak augmenting path with
			the following additional properties: 
			\begin{enumerate}
				\item All elements $a_i$ are distinct;
				\item For every basis $I\in \calM$, letting $A_I=\left\{i\in [m]: I_i = I\right\}$, we have
				\[
				\widehat{I} = I\cup\left\{a_{i-1}: i\in A_I\right\}\setminus \left\{a_{i}: i\in A_I\right\} \in \calM. \quad \text{ (strong exchange property)}
				\]
			\end{enumerate}
		\end{definition}

		\noindent
		%
		%
		Intuitively, given a weak augmenting path, we can perform a sequence of exchange of basis
		according to the weak exchange property. As a result, we get a sequence of new basis
		so that the mass in $a_0=s$ is transported to $a_m=t$.
		However, there is a subtle technical problem that one basis may be used several times
		and we need to guarantee it is still a basis after all exchange operations,
		which motivates the notion of strong exchange property.
		The following crucial lemma shows that
		a strong augmenting path from $h$ to $h'$
		exists and its length can be bounded.
		%
		%
		
		\begin{lemma}[\bf{Existence of a strong augmenting path of length $\leq |E|-1$}]
			\label{claim:path_general}
			Given $\avg$ and $\avg'$ in $\calB$ that differ in exactly two entries,
			for any $\mu\in \Delta_{|V|}$ with $\sum_{v\in V} \mu_v\cdot v = \avg$, 
			there exists a strong augmenting path of length at most $|E|-1$ from $\avg$ to $\avg'$.
		\end{lemma}

		\begin{proof}
			For preparation, we construct a (multi-edge) directed graph $\augG$, called
			the {\em exchange graph}. $\augG$ has vertex set $E$ and the following set of
			edges: for every basis $I\in \calM$ with $\mu_{e_I} > 0$, every element $a\in E\setminus I$ and every element $b\in C(I, a)$, add a directed edge $(a,b)$ with a certificate $C(I,a)$ to the graph $\augG$.

			\paragraph{Existence of A Weak Augmenting Path}
			We first prove the existence of a weak augmenting path from $\avg$ to $\avg'$, which is equivalent to proving that there exists a path on $\augG$ from vertex $s$ to vertex $t$.
			By contradiction assume that there does not exist such a path.
			Let $E^-$ be the collection of elements $i\in E$ such that there exists a path from $s$ to $i$, and let $E^+ = E\setminus E^-$.
			We have that $t\in E^+$.
			Since $\sum_{i\in E^-} \avg_i < \sum_{i\in E^-} \avg'_i\leq \rank(E^-)$, we know that $E^-$ is not a tight set on $\avg$.
			Consequently, there must exist a basis $I\in \calM$ with $\mu_{e_I} > 0$ such that $|I\cap E^-| < \rank(E^-)$.
			Let $a\in E^-\setminus I$ be an element such that $(I\cap E^-)\cup \left\{a\right\}\in \calM$.
			If $C(I,a)\subseteq E^-$, we have $C(I,a)\cup \left\{a\right\}\notin \calM$.
			However, $C(I,a)\cup \left\{a\right\} \subseteq (I\cap E^-)\cup \left\{a\right\} \in \calM$, which is a contradiction.
			Hence, we have that $C(I,a)\cap E^+\neq \emptyset$ and there exists an edge from $a\in E^-$ to some element in $E^+$, which contradicts the definition of $E^-$.

			\paragraph{Existence of A Strong Augmenting Path of Length $\leq |E|-1$}
			Among all weak augmenting path from $\avg$ to $\avg'$, we select a shortest one, say $(v_1 = e_{I_1}, \ldots, v_m = e_{I_m})$ and $(a_0 = s, a_1, \ldots, a_{m-1}, a_m = t)$. 
			We claim that $a_0\neq a_1\neq \ldots \neq a_m$.
			Assume that $a_i = a_j$ for some $0\leq i<j\leq m$, we can see that $(e_{I_1}, \ldots, e_{I_i}, e_{I_{j+1}}, \ldots, e_{I_m})$ and $(a_0 = s, a_1, \ldots, a_i, a_{j+1}, \ldots, a_{m-1}, a_m = t)$ form a shorter weak augmenting path, which is a contradiction.
			Since there are at most $|E|$ different elements, we have that $m\leq |E|-1$.

			Now we prove that $(v_1 = e_{I_1}, \ldots, v_m = e_{I_m})$ and $(a_0 = s, a_1, \ldots, a_{m-1}, a_m = t)$ already form a strong augmenting path.
			Note that for every $i\in [m]$, $(a_{i-1},a_i)$ is a direct edge on $\augG$ with a certificate $C(I_i, a)$.
			It is easy to verify this when all $v_1,\ldots, v_m$ are distinct (i.e., each edge $(a_{i-1}, a_i)$ is defined by a distinct basis).

			Now, we consider the more difficult case where $v_i$s are not all distinct.
			Observe $v_i\neq v_{i+1}$ (i.e., two adjacent edges correspond to two different bases), since $a_{i} \in I_i$ and $a_{i} \notin I_{i+1}$ which implies $I_i\neq I_{i+1}$.
			Fix a basis $I\in \calM$ with $A_I=\left\{i_l\in [m]: t\in [T], I_{i_l} = I\right\}$ and $|A_I|\geq 2$.
			W.l.o.g., assume $0\leq i_1 < i_2< \ldots <i_T\leq m$ and we have $i_{l+1} - i_l\geq 2$ for $l\in [T-1]$.
			We first note that for any $i<j\in A_I$, $a_{j}\notin C(I, a_{i-1})$.
			Suppose, for contradiction that $a_{j}\in C(I, a_{i-1})$.
			Based on the construction of graph $\augG$, $(a_{i-1}, a_{j})$ is also an edge with certificate $C(I, a_{j})$.
			Consequently, $(e_{I_1}, \ldots, e_{I_i}, e_{I_{j+1}}, \ldots, e_{I_m})$ and $(a_0 = s, a_1, \ldots, a_{i-1}, a_j, a_{j+1}, \ldots, a_{m-1}, a_m = t)$ form a shorter weak augmenting path, which is a contradiction.

			Hence, we have that $a_{j} \notin C(I, a_{i-1})$ holds for any $i<j\in A_I$.
			We first swap $a_{i_T-1}$ and $a_{i_T}$ and obtain a basis $I^{(T)}=I\cup \left\{a_{i_T-1}\right\} \setminus \left\{a_{i_T}\right\}$.
			The circuit $C(I, a_{i_{T-1}-1})$ is completely contained in $I^{(T)}$ since $a_{i_T}\notin C(I, a_{i_{T-1}-1})$, and hence we can continuously perform the swap between $a_{i_{T-1}-1}$ and $a_{i_{T-1}}$.
			By reduction, we have that $I^{(l)} = I\cup \left\{a_{i_l-1}, \ldots, a_{i_T-1}\right\} \setminus \left\{a_{i_l}, \ldots, a_{i_T}\right\}$ for every $l\in [T]$ is still a basis.
			Hence, $\widehat{I} = I^{(1)}$ is a basis, which completes the proof.
		\end{proof}
		
		\noindent
		We consider the following procedure:
		\begin{enumerate}
			\item Find a strong augmenting path of length $m\leq |E|-1$ from $\avg$ to $\avg'$, say $(v_1 = e_{I_1}, \ldots, v_m = e_{I_m})$ and $(a_0 = s, a_1, \ldots, a_{m-1}, a_m = t)$.
			\item Let $\tau = \min\left\{\avg'_1 - \avg_1, \min_{i\in [m]} \mu_{v_i}\right\}$.
			For every basis $I\in \calM$, let $A_I:=\left\{i\in [m]: I_i = I\right\}$ and let $\widehat{I} = I\cup\left\{a_i: i\in A_I\right\}\setminus \left\{a_{i-1}: i\in A_I\right\}$.
			\item Construct $\mu''\in \Delta_{|V|}$ as the resulting distribution of the following procedure: for every $I\in \calM$ with $A_I\neq \emptyset$, reduce $\mu_{e_I}$ by $\tau$ and increase $\mu_{e_{\widehat{I}}}$ by $\tau$.
			\item Let $\avg'' = \sum_{v\in V} \mu''_v\cdot v$.
			If $\avg'' = \avg'$, we are done.
			Otherwise, let $\avg\leftarrow \avg''$ and $\mu\leftarrow \mu''$, and iteratively run the above steps for tuple $(\avg, \avg', \mu)$.
		\end{enumerate}
		After running an iteration, we can easily verify the following
		\begin{enumerate}
			\item $\|\avg'' - \avg'\|_0\leq 2$, $\avg''_1 - \avg_1 = \tau$ and $\avg_2 - \avg''_2 = \tau$.
			\item $\|\avg'' - \avg'\|_1 = \|\avg - \avg'\|_1 - 2\tau$.
			\item The total assignment transportation from $\mu$ to $\mu''$, say $\min_{\kappa \vdash (\mu, \mu'')} \sum_{v,v'\in V} \kappa(v,v')\cdot \|v-v'\|_1$, is at most
			$\sum_{I\in \calM: |I| = \rank(\calM)} \tau\cdot \|e_{I} - e_{\widehat{I}}\|_1 = 2m\tau\leq 2(|E| - 1) \tau$.
		\end{enumerate}
		It means that we can reduce the value $\|\avg - \avg'\|_1$ by $2\tau$, by introducing at most $2(|E| - 1) \tau$ assignment transportation for $\widetilde{\AT}(\calB, \avg, \avg', \mu)$.
		Hence, the required assignment transportation from $\avg$ to $\avg'$ is at most $(|E| - 1) \cdot \|\avg - \avg'\|_1$, which implies that
		\[
		\widetilde{\AT}(P_{\calM}, \avg, \avg', \mu) \leq (|E| - 1) \cdot \|\avg - \avg'\|_1.
		\]
		Due to the arbitrary selection of $\avg, \avg'$ and $\mu$, we complete the proof of Theorem~\ref{thm:laminar_matroid}.
	\end{proof}
	
	\subsubsection{Lipschitz Constant for Laminar Matroid Basis Polytopes}
	\label{sec:laminar_matroid}
	
	Our main theorem for laminar matroid basis polytopes is as follows.

	\begin{theorem}
		\label{thm:laminar_matroid}
		Let $E$ be a ground set and let $\calM$ be a laminar matroid on $E$ of depth $\ell \geq 1$.
		We have $\lip(P_{\calM})\leq \widetilde{\lip}(P_{\calM}) \leq \ell + 1$.
	\end{theorem}
	
	\begin{proof}
		Let $V = \left\{e_I: \text{$I$ is a basis of $\calM$}\right\}$ be the vertex set of $P_{\calM}$.
		Let $\calA$ be the corresponding laminar and $u$ be the corresponding capacity function of $\calM$.
		W.l.o.g., we assume $u(A) = \rank(A)$ for every $A\in \calM$.
		Let $\avg,\avg'\in P_{\calM}$ and $\mu \in \Delta_{|V|}$ be a distribution on $V$ satisfying that
		\[
		\sum_{v\in V} \mu_v\cdot v = \avg.
		\]
		By the proof of Theorem~\ref{thm:matroid}, we only need to prove  for the case that $\|\avg - \avg'\|_0 = 2$.
		W.l.o.g., we assume $\avg_1 < \avg'_1$ and $\avg_2 > \avg_2$.
		Again, by the same argument as in Theorem~\ref{thm:matroid}, it suffices to prove the following claim.

		\begin{lemma}
			\label{claim:path_laminar}
			There exists a strong augmenting path of length at most $\ell+1$ from $\avg$ to $\avg'$.
		\end{lemma}

		\begin{proof}
			We again construct the exchange graph $\augG$ as in the proof of \Cref{claim:path_general}, and iteratively construct a strong augmenting path from $\avg$ to $\avg'$.
			Also recall that we assume $\avg_s < \avg'_s$ and $\avg_t > \avg'_t$ for some $s,t\in E$.
			Let $A^\star \in \calA$ be the minimal set that contains $\left\{s,t\right\}$.
			Firstly, we find a vector $v_1 = e_{I_1}$ with $\mu_{v_1} > 0$ and $a_0=s\notin I_1$.
			Let $A_1\in \calA$ be the minimal set with $a_0\in A_1$ and $| I_1\cap A_1| = u(A_1)$.
			Note that such $A_1$ must exist, since $I_1\in P_{\calM}$ and we have $|I_1\cap E| = \rank(\calM) = u(E)$.
			We discuss the following cases.
			
			\paragraph{Case 1: $A_1\supseteq A^\star$}
			Since $\avg_t > \avg'_t \geq 0$, there must exist a vertex $v_2 = e_{I_2}$ with $\mu_{v_2} > 0$ and $a_2 = t\in I_2$.
			Note that $|(I_2\setminus \left\{a_2\right\}) \cap A_1| < |I_1\cap A_1| = u(A_1)$.
			There must exist an element $a_1\in I_1\cap A_1$ such that $(I_2\cap A_1)\cup \left\{a_1\right\}\setminus \left\{a_2\right\}\in \calM$.
			Due to the laminar structure, we know that $I_2\cup \left\{a_1\right\}\setminus \left\{a_2\right\}\in \calM$ also holds.
			Also, since $(I_1\cap A)\setminus \cup \left\{a_0\right\}\setminus \left\{a_1\right\}\in \calM$, we conclude that $I_1\cup \left\{a_0\right\}\setminus \left\{a_1\right\}\in \calM$.
			Consequently, $(v_1,v_2)$ and $(a_0,a_1,a_2)$ form an augmenting path of length 2 from $\avg$ to $\avg'$.
			
			\paragraph{Case 2: $A_1\subsetneq A^\star$}
			We know that $\sum_{i\in A_1} \avg'_i - \avg_i = \avg'_s - \avg_s > 0$.
			Thus, we must have
			\begin{align*}
				\sum_{e_I\in V}\mu_{e_I}\cdot |I \cap A_1| & =  \sum_{i\in A_1} \avg_i & (\text{Defn. of $\mu$}) \\
				& <  \sum_{i\in A_1} \avg'_i & \\
				& \leq  u(A_1), & (\avg'\in P_{\calM})
			\end{align*}
			which implies the existence of a vertex $v_2 = e_{I_2}\in V$ with $\mu_{v_2} > 0$ and $|I_2\cap A_1| < u(A_1)$.
			Since $u(A) = |I_1\cap A_1| > |I_2\cap A_1|$, there must exist an element $a_1\in I_1\cap A_1$ such that $(I_2\cap A_1)\cup \left\{a_1\right\}\in \calM$.
			Also note that since $I_1\cup \left\{a_0\right\}$ only violates the capacity constraint on sets $A\supseteq A_1$ ($A\in \calA$) by the fact that $\calM$ is a laminar matroid, we have that $I_1\cup \left\{a_0\right\}\setminus \left\{a_1\right\}\in \calM$.

			Note that $I_2\cup \left\{a_1\right\}$ can only violates the capacity constraint on sets $A\supsetneq A_1$ ($A\in \calA$).
			We can again find $A_2\in \calA$ be the minimal set with $a_1\in A_2$ and $|I_2\cap A_2| = u(A_2)$, and recursively apply the above argument to tuple $(I_2, a_1, A_2)$ until arriving $A_{m-1} \supseteq A^\star$.
			Since the operation on $I_1$, say $\widetilde{I}_1 = I_1\cup \left\{a_0\right\}\setminus \left\{a_1\right\}\in \calM$, maintains the number $|\widetilde{I}_1\cap A_2| = |I_1\cap A_2|$, we have that if $A_2\subsetneq A^\star$, the following inequality holds
			\[
			\sum_{e_I\in V}\mu_{e_I}\cdot |I \cap A_1| + \mu_{e_{I_1}}(|\widetilde{I}_1\cap A_2| - |I_1\cap A_2|) < u(A_2).
			\]
			Hence, the above argument can be applied to tuple $(I_2, a_1, A_2)$.

			Overall, we can get an augmenting path $(v_1 = e_{I_1}, \ldots, v_m = e_{I_m})$ and $(a_0 = s, a_1, \ldots, a_{m-1}, a_m = t)$, together with a certificate set sequence $A_1\subsetneq A_2\subsetneq \ldots \subsetneq A_{m-1}$ where $A_{m-1} \supseteq A^\star$.
			Since the depth of $\calM$ is $\ell$, we have $m-1\leq \ell$.
			Thus, we complete the proof of Claim~\ref{claim:path_laminar}.
		\end{proof}
		
		\noindent
		Overall, we complete the proof of Theorem~\ref{thm:laminar_matroid}.
	\end{proof}
	
	\subsection{Lipschitz Constant $\lip(\calB)$ can be Unbounded}
	\label{sec:Lipschitz_unbounded}
	
	We show that $\lip(\calB)$ can be extremely large by considering the following assignment structure constraints $\calB$, called knapsack polytopes. 

	\begin{definition}[\bf{Knapsack polytope}]
		\label{def:knapsack}
		Let $A\in \R_{\geq 0}^{m\times k}$ be a non-negative matrix.
		We say an assignment structure constraint $\calB\in \Delta_k$ is $A$-knapsack polytope if 
		\[
		\calB = \left\{x\in \Delta_k: Ax\leq 1 \right\},
		\]
		where $Ax\leq 1$ is called knapsack constraints.
	\end{definition}.

	\noindent
	Our result is as follows.

	\begin{theorem}[\bf{Lipschitz constant may be unbounded for knapsack polytope}]
		\label{thm:knapsack_2}
		Let $m=2$ and $k=3$.
		For any $U > 0$, there exists a matrix $A\in \R_{\geq 0}^{m\times k}$ such that the Lipschitz constant $\lip(\calB)$ of the $A$-knapsack polytope $\calB$ satisfies $\lip(\calB)\geq U$.
	\end{theorem}
	
	\begin{proof}
		We construct $A$ by $A_1 = (\frac{10U+2}{5U+3}, \frac{4}{5U+3}, 0)$ and $A_2 = (\frac{10U+2}{5U+3}, 0, \frac{4}{5U+3})$.
		Let $\avg = (\frac{1}{2}, \frac{1}{4}, \frac{1}{4})$ and $\avg' = (\frac{1}{2}+\frac{1}{10U}, \frac{1}{4} - \frac{1}{20U}, \frac{1}{4} - \frac{1}{20U})$.
		We can check that $\avg, \avg'\in \calB$ and specifically, $\avg'$ is a vertex of $\calB$.
		Let $\sigma: [2]\times [3]\rightarrow \R_{\geq 0}$ be defined as follows:
		\[
		\sigma(1,\cdot) = (\frac{1}{4}, \frac{1}{4}, 0), \text{ and } \sigma(2,\cdot) = (\frac{1}{4}, 0, \frac{1}{4}).
		\]
		We know that $\sigma \sim (\calB, \avg)$.
		Since $\avg'$ is a vertex of $\calB$, there is only one $\sigma': [2]\times [3]\rightarrow \R_{\geq 0}$ with $\sigma'\sim (\calB, \avg')$ and $\|\sigma'(p,\cdot)\|_1 = \|\sigma(p,\cdot)\|_1$ for $p\in [2]$, say
		\[
		\sigma'(1,\cdot) = \frac{1}{2}\cdot \avg', \text{ and } \sigma'(2,\cdot) = \frac{1}{2}\cdot \avg'
		\]
		for some $\alpha\in [0,1]$.
		Thus, we have
		\[
		\sum_{p\in [2]} \|\sigma(p,\cdot) - \sigma'(p,\cdot)\|_1 = \frac{1}{2} + \frac{3}{20U}.
		\]
		However, $\|\avg - \avg'\|_1 = \frac{1}{5U}$, which implies that
		\[
		\lip(\calB) \geq \AT(\calB, \avg, \avg', \sigma) \cdot \|\avg - \avg'\|_1
		= (\frac{1}{2} + \frac{3}{20U})/(\frac{1}{5U})\cdot  > U.
		\]
		This completes the proof.
	\end{proof}

	\section{Bounding the Covering Exponent $\cc{\eps}$}
	\label{sec:covering_metric}
	
	As we show in \Cref{thm:coreset}, the covering exponent $\cc{\eps}$ is a key parameter for the complexity of the metric space.
	In this section, we give several upper bounds of $\cc{\eps}$, against several other well-known notions of ``dimension'' measure of 
	the metric space $(\calX,d)$,
	which enables us to obtain coresets for constrained clustering in various metric spaces (see \Cref{remark:covering_exponent} for a summary of the concrete metric families that we can handle).
	We first introduce the shattering dimension (Definition~\ref{def:shattering}), and show its relation to the covering exponent in Lemma~\ref{lm:covering_shattering}.
	This relation helps to translate the upper bounds for the shattering dimension (which is rich in the literature) to small-size coresets (Corollary~\ref{cor:shattering}).
	Then we consider the doubling dimension (Definition~\ref{def:doubling}), and show the covering exponent is bounded by the doubling dimension (Lemma~\ref{lm:covering_doubling}).

	\paragraph{Shattering Dimension}
	%
	
	
	
	As was also considered in recent papers~\cite{baker2020coresets,BJKW21,braverman2022power}, we employ the following notion of shattering dimension of (the ball range space of) metric space (see also e.g. \cite{SarielBook}). 
	
	%
	
	
	\begin{definition}[\bf{Ball range space}] \label{def:ball}
		A ball $\Ball(x,r)=\{y\in \mathcal{X}\mid d(x,y)\leq r\}$ is determined by the center $c\in\mathcal{X}$ and the radius $r>0$. Let $\Balls(\calX)=\{\Ball(x,r)\mid x\in \mathcal{X},r>0\}$ denote the collection of all balls in $\calX$. $(\calX,\Balls(\calX))$ is called the ball range space of $\calX$. For $P\subset \mathcal{X}$, let 
		$$
		P\cap \Balls(\calX):=\{P\cap \Ball(x,r)\mid \Ball(x,r)\in \Balls(\calX)\}.
		$$
	\end{definition}
	
	\noindent
	We are ready to define the shattering dimension.

	\begin{definition}[\bf{Shattering dimension of ball range spaces}]
		\label{def:shattering}
		The shattering dimension of $(\calX,\Balls(\calX))$, denoted by $\dim(\calX)$ is the minimum positive integer $t$ such that for every $P\subset \mathcal{X}$ with $|P|\geq 2$,
		\begin{eqnarray*}
			|P\cap \Balls(\calX)|\leq |P|^t
		\end{eqnarray*}
	\end{definition}
	
	
	\noindent
	A closely related notion to shattering dimension is the well-known Vapnik-Chervonenkis dimension~\cite{vapnik1971uniform}. 
	Technically, they are equivalent up to only a logarithmic factor (see e.g., \cite{SarielBook}). 
	The applications of shattering dimension in designing small-sized coresets for unconstrained \kzC stem from the work \cite{feldman2011unified}, which are followed by \cite{huang2018epsilon,baker2020coresets,BJKW21}.
	However, all these works require to bound the shattering dimension of a more complicated weighted range space. Only recently, this complication of weight in the range space is removed in the uniform sampling framework by~\cite{braverman2022power}, which only requires to consider the unweighted ball range space (Definition~\ref{def:ball}).

	We have the following lemma that upper bounds the covering number via the shattering dimension.
	
	%
	
	\begin{lemma}
		\label{lm:covering_shattering}
		For every $\alpha > 0$, we have $\cc{\alpha} \leq O\left(\dim(\calX)\cdot z^2\alpha^{-1}\epsilon^{-1}\right)$.
	\end{lemma}

	\begin{proof}
		We fix an unweighted dataset $P\subseteq \Ball(a,r_{\max})$ of $n\geq 2$ points. 
		For ease of statement, we use $r$ to represent $r_{\max}$ in the following.
		We upper bound $\cover(P,\alpha)$ by constructing an $\alpha$-covering $\calC$.
		For every $x\in \Ball(a,24zr/\epsilon)$, let $f_x(y)=d(x,y),y\in P$ denote its distance function. We round $f_x(y)$ to obtain an approximation $\tilde{f}_x$ such that $\forall y\in P, \tilde{f}_x(y)=\lfloor\frac{24z\cdot d(x,y)}{\alpha r}\rfloor\cdot \frac{\alpha r}{24z}$. Let $\calF_{\alpha}=\{\tilde{f}_x\mid x\in \Ball(a,24zr/\epsilon)\}$ and for each $\tilde{f}\in \calF_{\alpha}$, we add exactly one $x_0$ such that $\tilde{f}_{x_0}=\tilde{f}$ into $\calC$ (notice that there can be multiple $x$'s that have the same approximation $\tilde{f}_{x}$, we include only one of them).
		
		We claim that $\calC$ is already an $\alpha$-covering of $P$. To see this, note that for every $x\in \Ball(a,24zr/\epsilon)$, by construction, $\forall y\in P, |\tilde{f}_x(y)-f_x(y)|\leq \frac{\alpha r}{24z}$ and also by construction, there exists $c\in\calC$ such that $\tilde{f}_c=\tilde{f}_x$. So we conclude that for every $$
		\max_{p\in P}|d(x,y)-d(c,y)|\leq \frac{\alpha r}{24z}+\frac{\alpha r}{24z}= \frac{\alpha r}{12z}.
		$$
		
		It remains to upper bound the size of $\calC$. 
		To achieve this goal, we first construct a family of distance functions $G_{\alpha}$. Let $T=O(\frac{z^2}{\alpha\cdot \epsilon})$. We start with enumerating all possible collections of subsets $\{H_i\mid i=0,1,\cdots, T\}$ satisfying the followings
		\begin{enumerate}
			\item $H_i\in P\cap \Balls(\calX)$ (Definition~\ref{def:ball}),
			\item $H_T\subseteq H_{T-1}\subseteq \cdots \subseteq H_0\subseteq P$. 
		\end{enumerate}
		By the definition of the shattering dimension, we know that $|P\cap \Balls(\calX)|\leq |P|^{\dim(\calX)}$. So there are at most $(|P|^{\dim(\calX)})^{O(z^2\alpha^{-1}\epsilon^{-1})}=|P|^{O(\dim(\calX)\cdot z^2\alpha^{-1}\epsilon^{-1})}$
		collections of such $\{H_i\mid i=0,\cdots T\}$.
		
		Fix a collection $\mathcal{H}=\{H_i\mid i=0,1,\cdots, T\}$, we construct a corresponding distance function $g_{\mathcal{H}}$ as the following. For every $p\in P$, let $i_p$ denote the maximum integer $i\in \{0,\cdots, T\}$ such that $p\in H_{i}$. Let $g_{\mathcal{H}}(p)=i_p\cdot \frac{\alpha r}{24z}$. Let $G_{\alpha}$ denote the subset of all possible $g_{\calH}$'s constructed as above.
		
		Now we claim that $|\calC|\leq |G_{\alpha}|$ since $|G_{\alpha}|\leq n^{O(z^2\alpha^{-1}\epsilon^{-1}\cdot \dim(\cal M))}$, which completes the proof.
		It suffices to show that $\calF_{\alpha}\subseteq G_{\alpha}$. 
		To prove $\calF_{\alpha}\subseteq G_{\alpha}$, we fix $\tilde{f}_x\in \calF_{\alpha}$ and show that there must exist a realization $\calH=\{H_i\mid i=0,1,\cdots,T\}$ such that $g_{\calH}=\tilde{f}_x$. To see this, we simply let $H_i=\{y\in P\mid \tilde{f}_x(y)\leq i\cdot \frac{\alpha r}{24z}\}$. It is obvious that $H_i\in P\cap \Balls(\calX)$ and $H_T\subseteq \cdots \subseteq H_0$, thus $\calH$ is a valid realization of the enumeration. Moreover, by construction, we know that $g_{\calH}=\tilde{f}_x$.
	\end{proof}

	\noindent
	Applying $\cc{\eps} = O_z(\dim(\calX) \cdot \eps^{-2})$ to Theorem~\ref{thm:coreset}, we have the following corollary that bounds the coreset size via $\dim(\calX)$.

	\begin{corollary}[\bf{Relating coreset size to shattering dimension}]
		\label{cor:shattering}
		Let $(\calX,d)$ be a metric space, $k\geq 1, m\geq 0$ be  integers, and $z\geq 1$ be a constant.
		Let $\eps,\delta\in (0,1)$ and $\calB\subseteq\Delta_k$ be 
		a convex body specifying the assignment structure constraint. 
		There exists a randomized algorithm that given a dataset $P\subseteq \calX$ of size $n\geq 1$ and an $(2^{O(z)},O(1),O(1))$-approximation $C^\star\in \calX^k$ of $P$ for \kzmC, constructs an $(\eps,\calB,m)$-coreset for \kzC with general assignment constraints of size
		\[
		O(m) + 2^{O(z\log z)}\cdot \tilde{O}(\lip(\calB)^2\cdot (\dim(\calX) \cdot \eps^{-2} + k) \cdot k^2\eps^{-2z}) \cdot \log \delta^{-1},
		\]
		in $O(nk)$ time with probability at least $1-\delta$.
		Moreover, when $\calB = \Delta_k$, the coreset size can be further improved to
		\[
		O(m) + 2^{O(z\log z)}\cdot \tilde{O}(\dim(\calX) \cdot k^2\eps^{-2z-2}) \cdot \log \delta^{-1}.
		\]
	\end{corollary}
	
	\paragraph{Doubling Dimension}
	Doubling dimension is an important generalization of Euclidean and more generally $\ell_p$ spaces with the motivation to capture the intrinsic complexity of a metric space \cite{assouad1983plongements,GuptaKL03}. Metric spaces with bounded doubling dimensions are known as doubling metrics. For unconstrained clustering, small-sized coresets were found in doubling metrics~\cite{huang2018epsilon,cohen2021new}. 
	
	\begin{definition}[\bf{Doubling dimension~\cite{assouad1983plongements,GuptaKL03}}]
		\label{def:doubling}
		The doubling dimension of a metric space $(\calX, d)$ is the least integer $t$,
		such that every ball can be covered by at most $2^t$ balls of half the radius.
	\end{definition}
	
	\begin{lemma}
		\label{lm:covering_doubling}
		For every $\alpha > 0$, we have $\cc{\alpha} \leq O\left(\ddim(\calX)\cdot \log(z\alpha^{-1} \eps^{-1})\right)$.
	\end{lemma}
	
	\begin{proof} 
		The proof is standard and we provide it for completeness.
		Fix an unweighted dataset $P\subseteq \Ball(a,r_{\max})$ of $n\geq 2$ points.
		For ease of statement, we use $r$ to represent $r_{\max}$ in the following.
		By the definition of doubling dimension, $\Ball(a,24zr/\epsilon)$ has an $\frac{\alpha r}{12z}$-net of size $(\frac{z}{\alpha\cdot\epsilon})^{O(\ddim(\calX))}$. 
		By the triangle inequality, we can see that this net is an $\alpha$-covering of $\calX$. 
		So we have 
		\[
		\cover(P,\alpha) \leq (\frac{z}{\alpha\cdot\epsilon})^{O(\ddim(\calX))} \leq n^{O\left(\ddim(\calX)\cdot \log(z\alpha^{-1} \eps^{-1})\right)}.
		\]
		Due to the arbitrary selection of $P$ and $n\geq 2$, we have $\cc{\alpha}\leq O\left(\ddim(\calX)\cdot \log(z\alpha^{-1} \eps^{-1})\right)$, which completes the proof.
	\end{proof}
	
	\noindent
	Similarly, applying $\cc{\eps} = \tilde{O}(\ddim(\calX))$ to Theorem~\ref{thm:coreset}, we have the following corollary that bounds the coreset size via $\ddim(\calX)$.

	\begin{corollary}[\bf{Relating coreset size to doubling dimension}]
		\label{cor:doubling}
		Let $(\calX,d)$ be a metric space, $k\geq 1, m\geq 0$ be  integers, and $z\geq 1$ be a constant.
		Let $\eps,\delta\in (0,1)$ and $\calB\subseteq\Delta_k$ be 
		a convex body specifying the assignment structure constraint. 
		There exists a randomized algorithm that given a dataset $P\subseteq \calX$ of size $n\geq 1$ and an $(2^{O(z)},O(1),O(1))$-approximation $C^\star\in \calX^k$ of $P$ for \kzmC, constructs an $(\eps,\calB,m)$-coreset for \kzC with general assignment constraints of size
		\[
		O(m) + 2^{O(z\log z)}\cdot \tilde{O}(\lip(\calB)^2\cdot (\ddim(\calX) + k + \eps^{-1}) \cdot k^2\eps^{-2z}) \cdot \log \delta^{-1},
		\]
		in $O(nk)$ time with probability at least $1-\delta$.
		Moreover, when $\calB = \Delta_k$, the coreset size can be further improved to
		\[
		O(m) + 2^{O(z\log z)}\cdot \tilde{O}(\ddim(\calX) \cdot k^2\eps^{-2z}) \cdot \log \delta^{-1}.
		\]
	\end{corollary}

	\begin{remark}[\bf{Covering exponents for special metrics}] 
		\label{remark:covering_exponent}
		We list below several examples of metric spaces that have bounded covering exponent, which is obtained by using \Cref{cor:shattering} or \Cref{cor:doubling}.
		\begin{itemize}
			\item Let $(\calX,d)$ be an Euclidean metric $\R^d$.
			We have $\ddim(\calX) \leq d+1$ and hence, $\cc{\eps} \leq \tilde{O}(d)$ by \Cref{lm:covering_doubling}.
			Once we have the bound $\cc{\eps} \leq \tilde{O}(d)$, we can further assume $d = \tilde{O}(\eps^{-2} \log k)$ by applying a standard iterative size reduction technique introduced in~\cite{BJKW21}, which has been applied in other coreset works~\cite{cohen2022towards,braverman2022power}.\eat{\footnote{Note that by \Cref{def:covering_wo}, we need to apply terminal embedding~\cite{Narayanan2019OptimalTD,Cherapanamjeri2022TerminalEI} with multiplicative error $\eps' = O(\eps^2)$ instead of $O(\eps)$ in the iterative size reduction process. This results in the term $\eps^{-4}$ in our assumption for $d$, instead of $\eps^{-2}$ in~\cite{BJKW21,braverman2022power}.}}
			This idea yields a coreset of size $ 2^{O(z\log z)}\cdot \tilde{O}(\lip(\calB)^2 \cdot k^2\eps^{-2z-2}) \cdot \log \delta^{-1}$, which removes the dependence of $d$.
			%
			%
			%
			\item Let $(\calX,d)$ be a doubling metric with bounded doubling dimension $\ddim(\calX)$. 
			We directly have $\cc{\eps} \leq \tilde{O}(\ddim(\calX))$ by \Cref{lm:covering_doubling}.
			\item Let $(\calX,d)$ be a general discrete metric.
			Note that $\ddim(\calX) \leq \log |\calX|$, and hence, we have $\cc{\eps} \leq \tilde{O}(\log |\calX|)$ by \Cref{lm:covering_doubling}.
			\item Let $(\calX,d)$ be a shortest-path metric of a graph with bounded treewidth \tw.
			%
			By~\cite{bousquet2015vc,baker2020coresets}, we know that $\dim(\calM)\leq O(t)$, which implies that $\cc{\eps}\leq \tilde{O}(t \eps^{-2})$ by \Cref{lm:covering_shattering}.
			\item Let $(\calX,d)$ be a shortest-path metric of a graph that excludes a fixed minor $H$.
			By~\cite{bousquet2015vc}, we know that $\dim(\calM)\leq O(|H|)$, which implies that $\cc{\eps}\leq \tilde{O}(|H|\cdot  \eps^{-2})$ by \Cref{lm:covering_shattering}.
			%
		\end{itemize}
	\end{remark}

	\bibliographystyle{alphaurl}
	\bibliography{ref}
	
	\appendix
	\newpage
	
	\section{Application of \Cref{thm:coreset}: Simultaneous Coresets for Multiple $\calB$'s}
	\label{sec:simultaneous}
	
	Since our coreset can handle all capacity constraints simultaneously,
	one may be interested in whether our coresets can also handle multiple assignment structure constraints $\calB$'s simultaneously as well.
	We show that this is indeed possible for several families of $\calB$.
	A similar idea, called simultaneous or one-shot coresets, has also been considered in the literature~\cite{BachemLL18,braverman2019coresets}.
	This type of coresets is particularly useful when some key hyper-parameters, in our case $\calB$, are not given/known but need to be figured out by experiments,
	since a single coreset can be reused to support all the experiments.
	Fortunately, our coreset in Theorem~\ref{thm:intro_main} has a powerful feature that it does not use any specific structure of the given $\calB$ except an upper bound of $\lip(\calB)$ in the algorithm.
	Hence, for a family $\calF$ of $m\geq 1$ different $\calB$'s,
	one can apply Theorem~\ref{thm:intro_main} with a union bound,
	so that with an additional $\poly(\log(m))$ factor in the coreset size, 
	the returned coreset is a coreset for all $\calB$'s in $\calF$ simultaneously.
	Below, we discuss the size bounds of the simultaneous coresets for two useful families $\calF$ and their potential applications.
	
	\begin{enumerate}
		\item Let $\calF$ be the collection of all (scaled) uniform matroid basis polytopes ($\lip(\calB)\leq 2$).
		Since the ground set is $[k]$, we have $|\calF| = k$ (each corresponding to
		a cardinality constraint).
		Then we can achieve a simultaneous coreset $S$ for every $\calB\in \calF$ by increasing the coreset size in \Cref{thm:intro_main} by a multiplicative $\log k$ factor, say $\tilde{O}_z(\cc{\eps}\cdot  k^2\eps^{-2z})$.
		Consequently, $S$ is a simultaneous coreset for all fault-tolerant clusterings. 
		\item Let $\calF$ be the collection of all (scaled) partition matroid basis polytopes ($\lip(\calB)\leq 3$).
		Since the ground set is $[k]$, there are at most $k^k$ partition ways.
		For each partition, there are at most $k^k$ different selections of rank functions (see \Cref{def:rank}).
		Thus, we have $|\calF|\leq k^{2k}$ and we can achieve a simultaneous coreset for every $\calB\in \calF$ by increasing the coreset size in \Cref{thm:intro_main} by a multiplicative $k \log (2k)$ factor, say $\tilde{O}_z(\cc{\eps}\cdot  k^3\eps^{-2z})$.
	\end{enumerate}
	
	\section{Missing Proofs in Section \ref{sec:pre}}
	
	\subsection{Proof of Claim \ref{claim:capacity}: Capacitated Clustering}
	\label{sec:sufficient}
	
	\begin{proof}
		Fix a center set $C\in \calX^k$.
		Let $\sigma^\star$ be an assignment that achieves the optimal capacitated clustering cost on $P$ with respect to $C$, i.e., 
		\[
		\min_{\substack{\sigma: ~ \|\sigma\|_1 = n - m, \sigma\sim \calB \\ \quad \quad \quad \quad \ell_c \leq \|\sigma(\cdot, c)\|_1 \leq u_c, \forall c\in C} } \cost_z^\sigma(P,C) = \cost_z^{\sigma^\star}(P, C).
		\]
		Let $\avg^\star\in (n-m)\cdot \calB$ such that $\sigma^\star \sim \avg^\star$.
		Then we have
		\[
		\cost_z^{\sigma^\star}(P, C) = \cost_z(P,C, \calB, \avg^\star).
		\]
		Since $\ell_c \leq \|\sigma(\cdot, c)\|_1 \leq u_c$ for $\forall c\in C$, we know that $\ell_c\leq \avg_c \leq u_c$. 
		Thus, we have
		\[
		\cost_z^{\sigma^\star}(P, C) = \min_{\substack{\avg: ~ \avg^\star\in (n-m)\cdot \calB \\ \quad \quad \quad \ell_c \leq \avg_c \leq u_c, \forall c\in C} } \cost_z^\sigma(P,C).
		\]
		Similarly, we have
		\[
		\min_{\substack{\sigma: ~ \|\sigma\|_1 = n - m, \sigma\sim \calB \\ \quad \quad \quad \quad \ell_c \leq \|\sigma(\cdot, c)\|_1 \leq u_c, \forall c\in C} } \cost_z^\sigma(S,C) = \min_{\substack{ \avg: ~ \avg^\star\in (n-m)\cdot \calB \\ \quad \quad \quad \ell_c \leq \avg_c \leq u_c, \forall c\in C} } \cost_z^\sigma(S,C).
		\]
		By Definition \ref{def:coreset}, $\cost_z(S, C, \calB, \avg^\star)\leq (1+\eps) \cost_z(P, C, \calB, \avg^\star)$ .
		We also have
		\[
		\min_{\substack{\sigma: ~ \|\sigma\|_1 = n - m, \sigma\sim \calB \\ \quad \quad \quad \quad \ell_c \leq \|\sigma(\cdot, c)\|_1 \leq u_c, \forall c\in C} } \cost_z^\sigma(S,C) \leq \cost_z(S, C, \calB, \avg^\star)
		\]
		Combining the above two inequalities, we conclude that
		\[
		\min_{\substack{\sigma: ~ \|\sigma\|_1 = n - m, \sigma\sim \calB \\ \quad \quad \quad \quad \ell_c \leq \|\sigma(\cdot, c)\|_1 \leq u_c, \forall c\in C} } \cost_z^\sigma(S,C) \leq (1+\eps)\cdot \min_{\substack{\sigma: ~ \|\sigma\|_1 = n - m, \sigma\sim \calB \\ \quad \quad \quad \quad \ell_c \leq \|\sigma(\cdot, c)\|_1 \leq u_c, \forall c\in C} } \cost_z^\sigma(P,C).
		\]
		The other direction can be proved in the same way, which completes the proof.
	\end{proof}
	
	\subsection{Proof of Claim \ref{claim:fair}: Fair Clustering}
	\label{sec:fairness}
	
	\begin{proof}
		The idea is similar to that of \cite[Theorem 4.2]{huang2019coresets}, while we additionally consider a total capacity constraint.
		Fix a center set $C\in \calX^k$.

		We first consider the simple case that $\Gamma = 1$ in which all $\mathcal{G}_p$'s are the same.
		Fix a center $c$.
		For each $i\in \mathcal{G}_p$, we have that $\frac{|\sigma^{-1}(c) \cap G_i|}{|\sigma^{-1}(c)|} = 1$; and for each $i\notin \mathcal{G}_p$, we have that $\frac{|\sigma^{-1}(c) \cap G_i|}{|\sigma^{-1}(c)|} = 0$.
		Then if there exists some group $G_i$ with $(i\in \mathcal{G}_p) \wedge (u_i < 1)$, or with $(i\notin \mathcal{G}_p) \wedge (u_i > 0)$ holds, there is no feasible solution for the fair clustering problem with outliers.
		Otherwise, the fair clustering problem with outliers is reduced to a clustering problem with outliers.
		Thus, by the assumption of Claim \ref{claim:fair}, an $\eps$-coreset exists for the fair clustering problem with outliers of size at most $A$.

		For the general case that $\Gamma \geq 1$, we partition $P$ into $\Gamma$ disjoint subsets $P_1, \ldots, P_{\Gamma}$ in which all points $p\in P_j$ correspond to the same collection $\mathcal{G}_p$ for each $j\in [\Gamma]$.
		For each $j\in [\Gamma]$, we let $\mathcal{G}_j$ denote the corresponding collection $\mathcal{G}_p$ of $p\in P_j$.
		Suppose $S_j\subseteq P_j$ is an $(\eps,\Delta_k,m')$-coreset of $P_j$ of size at most $A$ for all $0\leq m'\leq m$.
		Then it suffices to prove that $S:= \bigcup_{j\in [\Gamma]} S_j$ is an $\eps$-coreset for the fair clustering problem with outliers.

		The remaining proof is similar to that of Claim \ref{claim:capacity}.
		Let $\avg^{(j)}\in \R^k_{\geq 0}$ be the capacity vector of group $j\in [\Gamma]$.
		Also, let $\sigma^{(j)}: P_j\times C\rightarrow \R_{\geq 0}$ be the assignment function of group $j$.
		Let $\sigma: P \times C \rightarrow \R_{\geq 0}$ be the union assignment function of $\sigma^{(1)},\ldots, \sigma^{(s)}$.
		We have for every $i\in [s]$ and $c\in C$, 
		\[
		|\sigma^{-1}(c)| = \sum_{j\in [\Gamma]} \avg^{(j)}_c, \text{ and } |\sigma^{-1}(c) \cap G_i| = \sum_{j\in [\Gamma]} \mathrm{I}\left[i\in \mathcal{G}_j \right]\cdot \avg^{(j)}_c,
		\]
		where $\mathrm{I}\left[\cdot \right]$ is the indicator function.
		Then a fairness constraint $\ell_i\leq \frac{|\sigma^{-1}(c) \cap G_i|}{|\sigma^{-1}(c)|} \leq u_i$ is equivalent to
		\[
		\ell_i\leq \frac{\sum_{j\in [\Gamma]} \mathrm{I}\left[i\in \mathcal{G}_j \right]\cdot \avg^{(j)}_c}{\sum_{j\in [\Gamma]} \avg^{(j)}_c} \leq u_i.
		\]
		Thus, to check whether $\sigma$ satisfies all fairness constraints, we only need the information of vectors $\avg^{(1)}, \ldots, \avg^{(\Gamma)}$.
		Consequently, there must exist vectors $\avg^{(1)}, \ldots, \avg^{(\Gamma)}$ with $\sum_{j\in [\Gamma]} \avg^{(j)} = \avg$, $|P_j| - m \leq \|\avg^{(j)}\|_1 \leq |P_j|$ for all $j\in [\Gamma]$ and 
		\[
		\ell_i\leq \frac{\sum_{j\in [\Gamma]} \mathrm{I}\left[i\in \mathcal{G}_j \right]\cdot \avg^{(j)}_c}{\sum_{j\in [\Gamma]} \avg^{(j)}_c} \leq u_i
		\]
		for all $i\in [s]$ and $c\in C$, such that
		\[
		\min_{\substack{\sigma: ~ \|\sigma\|_1 = n - m \\ \quad \quad \quad \quad \ell_i \leq \|\sigma(\cdot, c)\|_1 \leq u_i, \forall i\in [s], c\in C} } \cost_z^\sigma(P,C) = \sum_{j\in [\Gamma]} \cost_z(P_j, C, \Delta_k, \avg^{(j)}).
		\]
		By the definition of $S_j$, we know that
		\begin{align*}
			\sum_{j\in [\Gamma]} \cost_z(P_j, C, \Delta_k, \avg^{(j)}) \geq & \quad (1-\eps)\cdot \sum_{j\in [\Gamma]} \cost_z(S_j, C, \Delta_k, \avg^{(j)}) \\ 
			\geq & \quad (1-\eps)\cdot \min_{\substack{\sigma: ~ \|\sigma\|_1 = n - m \\ \quad \quad \quad \quad \ell_i \leq \|\sigma(\cdot, c)\|_1 \leq u_i, \forall i\in [s], c\in C} } \cost_z^\sigma(S,C).
		\end{align*}
		Thus, we conclude that
		\[
		\min_{\substack{ \sigma: ~ \|\sigma\|_1 = n - m \\ \quad \quad \quad \quad \ell_i \leq \|\sigma(\cdot, c)\|_1 \leq u_i, \forall i\in [s], c\in C} } \cost_z^\sigma(P,C) \geq (1-\eps)\cdot \min_{\substack{\sigma: ~ \|\sigma\|_1 = n - m \\ \quad \quad \quad \quad \ell_i \leq \|\sigma(\cdot, c)\|_1 \leq u_i, \forall i\in [s], c\in C} } \cost_z^\sigma(S,C).
		\]
		Similarly, we can prove that 
		\[
		\min_{\substack{\sigma: ~ \|\sigma\|_1 = n - m \\ \quad \quad \quad \quad \ell_i \leq \|\sigma(\cdot, c)\|_1 \leq u_i, \forall i\in [s], c\in C} } \cost_z^\sigma(P,C) \leq (1+\eps)\cdot \min_{\substack{\sigma: ~ \|\sigma\|_1 = n - m \\ \quad \quad \quad \quad \ell_i \leq \|\sigma(\cdot, c)\|_1 \leq u_i, \forall i\in [s], c\in C} } \cost_z^\sigma(S,C).
		\]
		This completes the proof.
	\end{proof}
	
	\section{Proof of Lemma~\ref{lm:relation}: Relation between Two Covering Notions}
	\label{sec:proof_relation}
	
	\begin{lemma}[\bf{Restatement of Lemma~\ref{lm:relation}}]
		\label{lm:relation_restated}
		Let $\calB$ be an assignment structure constraint.
		For every $\beta > 0$ and $t\in [k]$, we have
		\begin{align*}
			\log \cover(R,t,\beta,\calB) \leq & O\left(t\cdot \log \cover(n_R,\beta)) + zk\cdot \log (\lip(\calB) \cdot z\eps^{-1} \beta^{-1})\right) \\
			\leq & O\left(\cc{\beta}\cdot t\log n_R + zk\cdot \log (\lip(\calB) \cdot z\eps^{-1} \beta^{-1})\right).
		\end{align*}
		Moreover, when $\calB = \Delta_k$, we have
		\[
		\log \cover(R,t,\beta,\calB) \leq O\left(\cc{\beta}\cdot t\log n_R + zt\cdot \log (z k \eps^{-1} \beta^{-1})\right).
		\]
	\end{lemma}
	
	\begin{proof} 
		
		For preparation, we provide the following lemma.
		
		\begin{lemma}[\bf{Extension of $\lip(\calB)$ to $\conv(\calB^o)$}]
			\label{lm:extension}
			Suppose there are real numbers $a\geq b>0$, constraints $\avg\in a\calB,\avg'\in b\calB$, and $\sigma\sim (a\calB,\avg)$ then there exists $\sigma'\sim (b\calB,\avg')$ such that $\|\sigma-\sigma'\|\leq 3\lip(\calB)\cdot \|\avg-\avg'\|_1$.
		\end{lemma}
		
		\begin{proof}
			By definition of $\lip(\calB)$, there exists $\pi\sim (a\calB,\frac{a}{b}\avg')$ such that $$
			\|\pi-\sigma\|_1\leq \lip(\calB)\cdot \|\avg-\frac{a}{b}\avg'\|_1.
			$$
			Let $\sigma'=\frac{b}{a}\pi$, obviously $\sigma\sim(b\calB,\avg')$. Moreover,
			\begin{eqnarray*}
				\|\sigma'-\sigma\|_1&\leq& \|\sigma'-\pi\|_1+\|\sigma-\pi\|\\
				&\leq & \frac{|a-b|}{b}\cdot \|\sigma'\|_1+\lip(\calB)\cdot |\avg-\frac{a}{b}\avg'|\\
				&\leq & |a-b|+\lip(\calB)\cdot (\|\avg-\avg'\|_1+\frac{|a-b|}{b}\cdot \|\avg'\|_1)\\
				&=& |a-b|+\lip(\calB)\cdot (\|\avg-\avg'\|_1+|a-b|)\\
				&\leq & 3\lip(\calB)\cdot \|\avg-\avg'\|_1
			\end{eqnarray*}
			where for the last inequality, we use the triangle inequality to obtain $\|\avg-\avg'\|_1\geq |\|\avg\|_1-\|\avg'\|_1|=|a-b|$.
		\end{proof}
		
		\noindent
		We first show how to construct an $\beta$-covering $\calF\subset \Phi_t$ w.r.t. $(R,\calB)$ with the desired size and covering property.
		For center sets, let $\calC$ denote the $\beta$-covering of $B(c_i^\star, \frac{48z r}{\eps})$ and for technical convenience we also add $c_i^\star$ into $\calC$.
		Let $\calC[t]$ denote the collection of $k$-tuples $(c_1,c_2,\ldots,c_k)$ satisfying that
		\begin{itemize}
			\item exactly $k-t$ centers $c = c_i^\star$;
			\item the remaining $t$ centers are selected from $\calC$.
		\end{itemize}
		For capacity constraints, let $\mathcal{N}$ denote an $l_1$-distance $\frac{\beta\cdot \eps^z}{12\cdot (48z)^z\cdot \lip(\cal B)}$-net of $\conv(\calB^o)$, namely, for every $\avg\in \conv(\calB^o)$, there exists $ \avg'\in \calN$ such that $\|\avg-\avg'\|_1\leq \frac{\beta\cdot \eps^z}{12\cdot (48z)^z\cdot \lip(\cal B)}$. Since $\calB\subseteq \Delta_k$, we know that $|\mathcal{N}|\leq (\frac{z\cdot \lip(\cal B)}{\beta \cdot\eps^z})^{O(k)}$. 
		We let $\calF:=\calC[t] \times n_R\calN$ to be the Cartesian product of $\calC$ and $n_R\calN$.
		By construction, we have
		\[
		\log |\calF| \leq t \log k + t \log |\calC| + \log |\calN| \leq O\left(t\cdot \log \cover(n,\beta) + zk\cdot \log (\lip(\calB) \cdot z \beta^{-1} \eps^{-1})\right).
		\]
		Then it remains to show that $\calF$ is an $(t,\beta)$-covering w.r.t. $(R,\calB)$. 
		
		Now we fix a pair $(C,h)\in \Phi_t$. 
		We show that there exists $(C',\avg')\in \mathcal{F}$ such that for every $Q\subseteq R, w_Q(Q)=n_R$, 
		\[
		\cost_z(Q,C,\calB,\avg)\in \big(1\pm (\beta+\eps)\big)\cdot \cost_z(Q,C',\calB,\avg')\pm \beta n_R r^z.
		\]
		By Lemma~\ref{lm:approximation_cost}, it suffices to prove that
		\begin{eqnarray}
			\label{eqn:covering} 
			\cost_z(Q,\nu(C),\calB,\avg)\in (1\pm \beta)\cdot \cost_z(Q,C',\calB,\avg')\pm \beta n_R r^z.
		\end{eqnarray}

		\noindent
		For the construction of $C'$, we let $c' = c_i^\star$ for every $c\in C$ with $\nu(c) = c_i^\star$, and let $c'$ be the closest point in $\calC$ for the remaining $c\in C$.
		For the construction of $h'$, we know that there exists $\tilde{\avg}\in n_R\calN$ such that $\|\tilde{\avg}-\avg\|_1\leq \frac{\beta\cdot \eps^z}{12\cdot (48z)^z\cdot \lip(\cal B)}\cdot n_R$. 
		We define a capacity constraint $\avg'$ on $C'$ such that $\forall c\in C, \avg'(c')=\tilde{\avg}(c)$.

		In the following, we
		fix a weighted $Q\subseteq \Ball(a,r)$ with total weight $n_R$ and prove Inequality~\eqref{eqn:covering}. 
		For ease of analysis, we slightly abuse the notation by using $C$ to replace $\nu(C)$ in the following.
		Assume $\sigma$ is the corresponding optimal assignment for $\cost_z(Q,C,\calB,\avg)$ and expand it as
		$$
		\cost_z(Q,C,\calB,\avg)=\sum_{p\in Q} \sum_{c\in C} \sigma(p,c)\cdot d^z(p,c).
		$$
		It suffices to prove the following inequality:
		\begin{eqnarray}
			\label{eqn:cov:refined}
			\cost_z(Q,C',\calB,\avg')\in (1\pm \beta)\cdot \sum_{p\in Q} \sum_{c\in C} \sigma(p,c)\cdot d^z(p,c)\pm \beta n_R r^z.
		\end{eqnarray} 
		
		We prove the two directions separately. Firstly, let $\sigma'$ denote an optimal assignment of $\cost_z(Q,C',\calB,\avg')$, we prove that
		\begin{eqnarray}
			\label{eqn:cov:d1}
			\cost_z(Q,C',\calB,\avg')\geq (1-\beta)\cdot \sum_{p\in Q}\sum_{c\in C}\sigma(p,c)\cdot d^z(p,c)-\beta n_R r^z.
		\end{eqnarray}
		
		For the sake of contradiction, suppose (\ref{eqn:cov:d1}) does not hold. 
		By Lemma~\ref{lm:extension}, we know that there exists an assignment $\sigma''\sim (\calB,\avg')$ such that 
		\begin{eqnarray}
			\label{eq:3}
			\sum_{p\in Q}\sum_{c\in C}|\sigma''(p,c')-\sigma'(p,c)|\leq \frac{\beta}{4} \cdot \big(\frac{\eps}{48z}\big)^z \cdot n_R
		\end{eqnarray}
		It suffices to show that $\sigma''$ is actually a better assignment than $\sigma$ to conclude a contradiction. 
		We can see that
		\begin{align}
			\label{eq:5}
			\begin{aligned}
				& \cost_z^{\sigma''}(Q,C) &\\
				= &\sum_{p\in Q}\sum_{c\in C} \sigma''(p,c')\cdot d^z(p,c) &\\
				=&\sum_{p\in Q}\sum_{c\in C} \sigma'(p,c')\cdot d^z(p,c)+\sum_{p\in Q}\sum_{c\in C} |\sigma'(p,c')-\sigma''(p,c)|\cdot d^z(p,c)^z &\\
				\leq &\sum_{p\in Q}\sum_{c\in C} \sigma'(p,c')\cdot d^z(p,c)+\|\sigma''-\sigma'\|_1\cdot \big(\frac{48zr}{\eps}\big)^z & (\text{Defn. of $\barCfar$})\\
				\leq & \sum_{p\in Q}\sum_{c\in C} \sigma'(p,c')\cdot d^z(p,c)+ \frac{\beta}{4} n_R r^z. & (\text{Ineq.~\eqref{eq:3}})
			\end{aligned}
		\end{align}
		By Lemma~\ref{lm:triangle}, we know that
		\begin{align}
			\label{eq:6}
			\begin{aligned}
				&\sum_{p\in Q}\sum_{c\in C} \sigma'(p,c')\cdot d^z(p,c) &\\
				=&\sum_{p\in Q}\sum_{c\in C} \sigma'(p,c')\cdot \big( d^z(p,c')+d^z(p,c)-d^z(p,c')\big) &\\
				\leq & (1+\beta) \cdot\sum_{p\in Q}\sum_{c\in C}\sigma'(p,c')\cdot d^z(p,c')+(\frac{3z}{\beta})^{z-1}\cdot n_R\cdot (\frac{\beta r}{12z}) & (\text{\Cref{lm:triangle}})\\
				\leq & (1+\beta) \cdot\cost_z(Q,C',\calB,\avg')+ \frac{\beta}{4} n_R r^z. & \\
			\end{aligned}
		\end{align}
		Summing up the above three inequalities, we conclude that 
		\begin{align*}
			\begin{aligned}
				\cost_z^{\sigma''}(Q,C) &\leq &&(1+\beta)\cdot \cost_z(Q,C',\calB,\avg')+\frac{\beta}{2} n_R r^z & (\text{Ineqs.~\eqref{eq:5} and~\eqref{eq:6}})\\
				& < && \cost_z^{\sigma}(Q,C), & (\text{by assumption})
			\end{aligned}
		\end{align*}
		which has been a contradiction to the fact that $\sigma$ is an optimal assignment for $\cost_z(Q,C,\calB,\avg)$.
		
		To prove the other direction of (\ref{eqn:cov:refined}), it suffices to construct an assignment $\sigma'\sim (\calB,\avg')$ such that
		\begin{eqnarray*}
			\cost_z^{\sigma'}(Q,C)\leq (1+\beta)\sum_{p\in Q}\sum_{c\in C}\sigma(p,c)\cdot d^z(p,c) +\beta n_R r^z.
		\end{eqnarray*}
		To this end, we simply choose $\sigma'$ to be an assignment consistent with $(\calB,\avg')$ and 
		$$\sum_{p\in Q}\sum_{c\in C}|\sigma'(p,c)-\sigma(p,c)|\leq \frac{\beta}{4}\cdot (\frac{\eps}{48z})^z\cdot n_R.
		$$ 
		Again, by Lemma~\ref{lm:extension}, such $\sigma'$ exists.
		Then by a similar argument as in the first direction, we can verify that 
		$$
		\sum_{p\in Q}\sum_{c\in C} \sigma'(p,c')\cdot d^z(p,c')\leq (1+\beta)\sum_{p\in Q}\sum_{c\in C} \sigma(p,c)\cdot d^z(p,c)+\beta n_R r^z
		$$ 
		which completes the proof of the first inequality of Lemma~\ref{lm:relation_restated}.
		
		When $\calB = \Delta_k$, the only difference is the construction of $\calN$. 
		Note that for any $C\subset B(c_i^\star, \frac{48z r}{\eps})$ and any $h,h'\in n_R \conv(\calB^o)$ satisfying that 1) $h_c = h'_c$ for every $c\in C$ with $c\neq c_i^\star$; and 2) $\sum_{c\in C: c = c_i^\star} h_c = \sum_{c\in C: c = c_i^\star} h'_c$, the following holds:
		\[
		\cost_z(Q, C, \Delta_k, h) = \cost_z(Q, C, \Delta_k, h').
		\]
		This observation enables us to construct $\calN$ as follows:
		\begin{enumerate}
			\item Enumerate all subsets $A\subseteq [k]$ of size $t-1$ and let $\Delta_A = \left\{h\in \Delta_k: \forall i\in A, h_i = 0\right\}$. 
			\item Construct an $l_1$-distance $\frac{\beta\cdot \eps^z}{12\cdot (48z)^z\cdot \lip(\cal B)}$-net for every $\Delta_A$ and let $\calN$ be their union.
		\end{enumerate}
		We still let $\calF:=\calC[t] \times n_R\calN$, which can be shown an $(t,\beta)$-covering w.r.t. $(R,\Delta_k)$.
		By construction, we have
		\[
		\log |\calF| \leq t \log k + t \log |\calC| + \log |\calN| \leq O\left(t\cdot \log \cover(n,\beta) + zt\cdot \log ( z k \beta^{-1} \eps^{-1})\right),
		\]
		which completes the proof of the second inequality of Lemma~\ref{lm:relation_restated}.
	\end{proof}
	
\end{document}